\documentclass[reqno,11pt]{amsart}
\usepackage[utf8]{inputenc}
\usepackage[hyperref=true,
            url=false,
            isbn=false,
            backref=false,
            doi=false,
            giveninits=true,
            style=numeric-comp,
            maxcitenames=99,
            maxbibnames=100,
            labelalpha, 
            block=none]{biblatex}
\AtEveryBibitem{\clearfield{note}{}}

\DeclareFieldFormat*{title}{\textit{#1}}
\DeclareFieldFormat[article]{journaltitle}{#1}
\DeclareFieldFormat[article]{number}{(#1)}
\DeclareFieldFormat[article]{volume}{\textbf{#1}}
\DeclareFieldFormat[incollection]{booktitle}{#1}
\DeclareFieldFormat[inproceedings]{booktitle}{#1}

\renewbibmacro{in:}{%
  \ifboolexpr{%
     test {\ifentrytype{article}}%
     or
     test {\ifentrytype{inproceedings}}%
     or
     test {\ifentrytype{incollection}}%
  }{}{\printtext{\bibstring{in}\intitlepunct}}%
}

\AtEveryBibitem{%
    \ifentrytype{article}{%
        \renewbibmacro{volume+number+eid}{%
            \printfield{volume}%
            \setunit{\addcomma\space}%
            \printfield{eid}}%
        }%
        {}%
    }

\calclayout

\usepackage{amsthm,amsfonts,amstext,amssymb,mathrsfs,amsmath,latexsym,mathtools} 
\usepackage{enumerate}
\usepackage{bm}
\usepackage{mdwlist}
\usepackage[abs]{overpic} 
\usepackage[mathscr]{euscript}
\usepackage{mathrsfs}
\DeclareMathAlphabet{\mathpzc}{OT1}{pzc}{m}{it}
\usepackage{parskip}

\usepackage{hyperref} 
\usepackage{comment}

\usepackage{tikz}
\usetikzlibrary{3d}

\usepackage[font=small]{caption}
\usepackage[labelformat=simple]{subcaption}
\renewcommand{\thesubfigure}{\alph{subfigure}}
\usepackage{color}
\usepackage{esint} 
\usepackage[colorinlistoftodos,prependcaption,textsize=small]{todonotes}

\definecolor{red}{RGB}{255,0,0}
\definecolor{green}{RGB}{0,100,0}
\definecolor{blue}{RGB}{0,0,255}  

\usepackage{cleveref}
 
\crefname{equation}{equation}{equations}
\crefname{figure}{Figure}{Figures}

\newtheorem{thm}{Theorem}[section]
 \newtheorem{thmx}{Theorem}
\newtheorem{prop}[thm]{Proposition}
\newtheorem{lem}[thm]{Lemma}
\newtheorem{cor}[thm]{Corollary}

\theoremstyle{definition}
\newtheorem{definition}[thm]{Definition}

\theoremstyle{remark}
\newtheorem{remark}[thm]{Remark} 
\numberwithin{equation}{section} 

\numberwithin{equation}{section} 

\DeclareMathOperator{\re}{Re}
\DeclareMathOperator{\im}{Im}

\renewcommand{\Im}{\mathop{\rm Im}}

\DeclareMathOperator{\supp}{supp} 
\DeclareMathOperator{\discr}{Discr}
\DeclareMathOperator{\const}{const}
\DeclareMathOperator{\res}{Res}

\DeclareMathOperator{\tr}{Tr}
\DeclareMathOperator{\diag}{diag}

\DeclareMathOperator{\polyn}{Polyn}

\newcommand{\eqq}{\coloneqq}

\newcommand{\ds}{\displaystyle}
\newcommand{\E}{\mathbb{E}}
\newcommand{\HH}{\mathbb{H}}

\newcommand{\N}{\mathbb{N}}
\newcommand{\C}{\mathbb{C}}
\newcommand{\R}{\mathbb{R}}
\newcommand{\Z}{\mathbb{Z}}

\newcommand{\boh}{\mathit{o}}
\newcommand{\Boh}{\mathcal{O}}
\newcommand{\weak}{\stackrel{*}{\longrightarrow}}

\newcommand\restr[2]{{
		\left.\kern-\nulldelimiterspace 
		#1 
		\vphantom{\big|} 
		\right|_{#2} 
}}

\usepackage{tikz}
\usetikzlibrary{patterns}
\usetikzlibrary{decorations.pathmorphing}
\usetikzlibrary{decorations.pathreplacing,decorations.markings}

\tikzstyle{vecArrow} = [thick, decoration={markings,mark=at position
   1 with {\arrow[semithick]{open triangle 60}}},
   double distance=1.4pt, shorten >= 5.5pt,
   preaction = {decorate},
   postaction = {draw,line width=1.4pt, white,shorten >= 4.5pt}]
\tikzstyle{innerWhite} = [semithick, white,line width=1.4pt, shorten >= 4.5pt]

\makeatletter
\def\grd@save@target#1{%
  \def\grd@target{#1}}
\def\grd@save@start#1{%
  \def\grd@start{#1}}
\tikzset{
  grid with coordinates/.style={
    to path={%
      \pgfextra{%
        \edef\grd@@target{(\tikztotarget)}%
        \tikz@scan@one@point\grd@save@target\grd@@target\relax
        \edef\grd@@start{(\tikztostart)}%
        \tikz@scan@one@point\grd@save@start\grd@@start\relax
        \draw[minor help lines] (\tikztostart) grid (\tikztotarget);
        \draw[major help lines] (\tikztostart) grid (\tikztotarget);
        \grd@start
        \pgfmathsetmacro{\grd@xa}{\the\pgf@x/1cm}
        \pgfmathsetmacro{\grd@ya}{\the\pgf@y/1cm}
        \grd@target
        \pgfmathsetmacro{\grd@xb}{\the\pgf@x/1cm}
        \pgfmathsetmacro{\grd@yb}{\the\pgf@y/1cm}
        \pgfmathsetmacro{\grd@xc}{\grd@xa + \pgfkeysvalueof{/tikz/grid with coordinates/major step}}
        \pgfmathsetmacro{\grd@yc}{\grd@ya + \pgfkeysvalueof{/tikz/grid with coordinates/major step}}
        \foreach \x in {\grd@xa,\grd@xc,...,\grd@xb}
        \node[anchor=north] at (\x,\grd@ya) {\pgfmathprintnumber{\x}};
        \foreach \y in {\grd@ya,\grd@yc,...,\grd@yb}
        \node[anchor=east] at (\grd@xa,\y) {\pgfmathprintnumber{\y}};
      }
    }
  },
  minor help lines/.style={
    help lines,
    step=\pgfkeysvalueof{/tikz/grid with coordinates/minor step}
  },
  major help lines/.style={
    help lines,
    line width=\pgfkeysvalueof{/tikz/grid with coordinates/major line width},
    step=\pgfkeysvalueof{/tikz/grid with coordinates/major step}
  },
  grid with coordinates/.cd,
  minor step/.initial=.2,
  major step/.initial=1,
  major line width/.initial=0.25mm,
}

\pgfdeclarepatternformonly{north east lines wide}%
   {\pgfqpoint{-1pt}{-1pt}}%
   {\pgfqpoint{10pt}{10pt}}%
   {\pgfqpoint{5pt}{5pt}}%
   {
     \pgfsetlinewidth{.8pt}
     \pgfpathmoveto{\pgfqpoint{0pt}{0pt}}
     \pgfpathlineto{\pgfqpoint{9.1pt}{9.1pt}}
     \pgfusepath{stroke}
    }

\tikzset{
  on each segment/.style={
    decorate,
    decoration={
      show path construction,
      moveto code={},
      lineto code={
        \path [#1]
        (\tikzinputsegmentfirst) -- (\tikzinputsegmentlast);
      },
      curveto code={
        \path [#1] (\tikzinputsegmentfirst)
        .. controls
        (\tikzinputsegmentsupporta) and (\tikzinputsegmentsupportb)
        ..
        (\tikzinputsegmentlast);
      },
      closepath code={
        \path [#1]
        (\tikzinputsegmentfirst) -- (\tikzinputsegmentlast);
      },
    },
  },
  mid arrow/.style={postaction={decorate,decoration={
        markings,
        mark=at position .5 with {\arrow[#1]{stealth}}
      }}},
  end arrow/.style={postaction={decorate,decoration={
        markings,
        mark=at position 1 with {\arrow[#1]{stealth}}
      }}},
  start arrow/.style={postaction={decorate,decoration={
        markings,
        mark=at position 0 with {\arrow[#1]{stealth}}
      }}},
}

\makeatother
\tikzset{every state/.style={minimum size=0pt}}

\title[Spectral Curves, Variational Problems and Random Matrices]{Spectral curves, variational problems and the hermitian matrix model with external source}

\author[A. Mart\'{\i}nez-Finkelshtein]{Andrei Mart\'{\i}nez-Finkelshtein}

\address[AMF]{Department of Mathematics, Baylor University, Waco TX, USA, and Department of Mathematics, University of Almer\'{\i}a, Almer\'{\i}a, Spain}

\email{A\_Martinez-Finkelshtein@baylor.edu}

\author[G.~Silva]{Guilherme L.~F.~Silva}

\address[GS]{Instituto de Ciências Matemáticas e de Computação, Universidade de São Paulo (ICMC-USP), São Carlos, SP, Brazil}

\email{silvag@usp.br}

 

\begin{document}

\begin{abstract} 
We consider the hermitian random matrix model with external source and general polynomial potential, when the source has two distinct eigenvalues but is otherwise arbitrary. 

All such models studied so far have a common feature: an associated cubic equation (``spectral curve''), one of whose solutions can be expressed in terms of the Cauchy (a.k.a.~Stieltjes) transform of the limiting eigenvalue distribution $\lambda$. 

This is our starting point: we show that to any such a spectral curve (not necessarily given by a random matrix ensemble) it corresponds a unique vector-valued measure with three components on the complex plane, characterized as a solution of a variational problem stated in terms of their logarithmic energy. We  describe  all  possible geometries of the supports of these measures: the third component, if non-trivial, lives on a contour on the plane and separates the supports of the other two measures, both on the real line. 

This general result is applied to the random matrix model with external source, under an additional assumption of uniform boundedness of the zeros of a sequence of average characteristic polynomials, when the size of the matrices goes to infinity (equivalently, uniform boundedness of certain recurrence coefficients). It is shown that any limiting zero distribution for such a sequence can be written in terms of a  solution of a spectral curve, and thus admits the variational description obtained in the first part of the paper. As a consequence of our analysis, we  obtain that the density of this limiting measure can have only a handful of local behaviors: sine, Airy and their higher order type behavior, Pearcey or yet the fifth power of the cubic (but no higher order cubics can appear).

We also compare our findings with the most general results available in the literature, showing that once an additional symmetry is imposed, our vector critical measure contains enough information to recover the solutions to the constrained equilibrium problem that was known to describe the limiting eigenvalue distribution in this symmetric situation.

\end{abstract}

\maketitle

\setcounter{tocdepth}{1}
\tableofcontents

\section{Introduction}

The main motivation of this work is the random matrix ensemble defined on the space $\mathcal M_N$ of $N\times N$ hermitian matrices equipped with the probability distribution
\begin{equation}\label{external_source_model}
\frac{1}{Z_N}e^{-N\tr(V(\bm M)-\bm A \bm M)}d \bm M,
\end{equation}
where $V$ is a polynomial, $\bm A$ is a fixed $N\times N$ hermitian matrix (the {\it external source}), $d\bm M$ is the Lebesgue measure on $\mathcal M_N\simeq \R^{N^2}$, and $Z_N$ is the normalization constant, also called partition function. Due to the external source, this ensemble is not unitary invariant, but because of the invariance of $d\bm M$ and the trace under unitary conjugation, without loss of generality we can assume that  $\bm A$ is diagonal and that $V(0)=0$. 

The most studied case is  when $\bm A$ has exactly two distinct eigenvalues, $a_1$ and $a_2$ with multiplicities $n_1$ and $n_2$. Since for any $c\in \R$, 
$$
V(\bm M)-\bm A\bm M=(V(\bm M)-c\bm M)+(c\bm I-\bm A)\bm M,
$$
we can symmetrize the problem at the cost of changing the linear term of $V$ by choosing
$$
c=\frac{a_1+a_2}{2},\quad a=\frac{a_1-a_2}{2},
$$
so that the eigenvalues of the new external source $\bm A-c\bm I$ are $-a$ and $a$. In other words, without loss of generality we consider the situation when
\begin{equation} \label{A}
\bm 	A= \diag (\underbrace{ a_1, a_1, \dots, a_1}_{n_1 \text{ times}}, \underbrace{ a_2, a_2, \dots, a_2}_{n_2 \text{ times}}), \quad a_1=-a_2=a \ge 0.
\end{equation}

The asymptotic analysis of statistical properties (as $N\to \infty$) of the spectra of random matrices perturbed by an additive external source can be traced back to Pastur \cite{pastur_spectrum_random_matrices}. In the late 1990's Zinn-Justin \cite{zinn_justin_external_field} showed that the eigenvalues of the model \eqref{external_source_model} are determinantal, and later Bleher and Kuijlaars \cite{bleher_kuijlaars_external_source_multiple_orthogonal} proved that the correlation kernel that appears in this determinantal expression can be written in terms of multiple orthogonal polynomials. The now well-known matrix-valued Riemann-Hilbert characterization for these polynomials (involving matrices whose size depends on the number of distinct eigenvalues of $\bm A$)  turned out to be crucial for the asymptotic analysis of external source models. Aptekarev, Bleher, and Kuijlaars were the first to develop this approach in a series of papers \cite{bleher_kuijlaars_external_source_gaussian_I,aptekarev_bleher_kuijlaars_external_source_gaussian_II, bleher_kuijlaars_external_source_gaussian_III}, where they analyzed the symmetric Gaussian case, when $V(\bm M)=\bm M^2/2$ and   $\bm A$ is given by \eqref{A},  with $n_1=n_2\to \infty$, recovering and extending results obtained a little earlier by Brézin and Hikami \cite{brezin_hikami_singularity_gap, brezin_hikami_level_spacing_random_matrices_external_source} and Tracy and Widom \cite{tracy_widom_2006_pearcey}. It is worth mentioning that the Gaussian random matrix model with an external source is equivalent to the so-called deformed GUE ensemble and with the Dyson Brownian motion
 (with parameter $\beta=2$), studied by Duits,  Erd\H{o}s,  Johansson, Landon, Yau and others (see e.g.~\cite{claeys_neuschel_venker, landon_Sosoe_Yau, duits_johansson, erdos_schnelli, erdos_yau, landon_yau}).

Riemann-Hilbert characterization of multiple orthogonal polynomials as a tool for the asymptotic analysis of the random matrix ensemble \eqref{external_source_model} has been successfully exploited in   \cite{bleher_delvaux_kuijlaars_external_source,bleher_kuijlaars_external_source_gaussian_I, bleher_kuijlaars_external_source_multiple_orthogonal, aptekarev_bleher_kuijlaars_external_source_gaussian_II,aptekarev_lysov_tulyakov,bertola_et_al_external_source_I, bertola_et_al_external_source_II}. Further results have also been obtained by other authors, see for instance  \cite{baik_wang_external_source_I,baik_wang_external_source_II,baik_wang_external_source_III}, and especially the recent monograph \cite{brezin_hikami_external_source_book}.

Let us summarize the key ingredients in some of the above mentioned asymptotic results.

Motivated by \cite{bleher_kuijlaars_external_source_gaussian_I,aptekarev_bleher_kuijlaars_external_source_gaussian_II, bleher_kuijlaars_external_source_gaussian_III}, McLaughlin \cite{mclaughlin_external_source} made an important observation that in an ``ideal situation'', when the original Riemann-Hilbert problem can be appropriately normalized with the help of certain ``$g$-functions'', there exists a function $\xi(z)$ that encodes the Cauchy (or Stieltjes) transform of the limiting density of eigenvalues of \eqref{external_source_model} for the ray sequences of $(n_1, n_2)$ with the slope $\alpha$, i.e.~in the asymptotic regime
\begin{equation}\label{asymptRegime}
N:=n_1+n_2 \to \infty,\quad \frac{n_1}{N}\to \alpha \in (0,1).
\end{equation}
This function satisfies a cubic equation, also known as \textit{spectral curve} or {\it master loop equation} (to be more precise, the spectral curve is the locus of solutions of equation \eqref{asymptRegime}, but we allow ourselves to use this term in a loose sense)  of the form
\begin{equation}\label{spectral_curv0}
\xi^3 - V'(z) \xi^2 + p_1(z) \xi + p_0(z)=0, 
\end{equation}
where $p_1$ and $ p_0$ are analytic functions of $z$. If $V$ is a polynomial with real coefficients, then the $p_j$'s are also polynomials, and $\xi=\xi(z)$ is an algebraic function. Furthermore, for a monic quartic potential $V(z)=z^4/4$, McLaughlin obtained an ansatz for \eqref{spectral_curv0} and used this to show that in this case his ideal situation indeed takes place. However, he also observed that for general $V$ this is not always the case.

It is worth pointing out that spectral curves for matrix models other than \eqref{external_source_model} have also been obtained before. We refer the reader to the works  \cite{kuijlaars_tovbis_supercritical_normal_matrix_model, duits_painleve_kernels,bertola_eynard_harnad_duality_multi_matrix, bertola_boutroux,xu_fan_chen2016, johansson_1998}
for an account on the existence and  relevance of spectral curves in random matrix theory. 

With an extension of the Riemann-Hilbert approach of \cite{bleher_kuijlaars_external_source_gaussian_I, bleher_kuijlaars_external_source_gaussian_III, aptekarev_bleher_kuijlaars_external_source_gaussian_II}, Bleher, Delvaux and Kuijlaars \cite{bleher_delvaux_kuijlaars_external_source} considered the asymptotic density of the eigenvalues of  \eqref{external_source_model} with $\bm A$ as in \eqref{A}, imposing the additional symmetry that the potential $V$ is an \textit{even} polynomial and $n_1=n_2\to \infty$ (which yields $\alpha = 1/2$ in \eqref{asymptRegime}). They showed that in this case the limiting density of eigenvalues can be described in terms of the solution of a variational problem involving two measures, one of them subject to a given upper bound (the so-called constrained vector equilibrium). They proved the existence and uniqueness of the asymptotic density, showing also the existence of the spectral curve \eqref{spectral_curv0}. In addition, in the case of even quartic potential $V$ they obtained the coefficients of \eqref{spectral_curv0} explicitly.

Independently, but almost simultaneously, Aptekarev, Lysov and Tulyakov \cite{aptekarev_lysov_tulyakov_2,aptekarev_lysov_tulyakov} tackled the same problem for the even quartic potential when $n_1=n_2\to \infty$. Their starting point was an ansatz for the spectral curve of the model, which was all they needed for their asymptotic analysis. Such ansatz is possible in virtue of the symmetry and low complexity of their potential $V$. In contrast with \cite{bleher_delvaux_kuijlaars_external_source}, they found a different variational description for the asymptotic density of the eigenvalues of  \eqref{external_source_model}, now in terms of a vector equilibrium (no constraints) involving three measures, one of them living on a contour in the complex plane. Similar complex vector equilibrium with three components, as discussed in \cite{aptekarev_lysov_tulyakov}, has appeared before in the asymptotic analysis of a certain class of multiple orthogonal polynomials, see e.g.~\cite{aptekarev_kuijlaars_van_assche_hermite_pade,MR2796829, AptekarevLysov2010}. They also observed in \cite{aptekarev_lysov_tulyakov} that the constrained equilibrium problem had no evident extension to the non-symmetric case of non-even potential $V$, and expressed their belief that their  equilibrium with the complex contour should be more suitable for a general situation.

In short, in the limited number of cases where a characterization of the limiting density of eigenvalues of  \eqref{external_source_model} with $\bm A$ as in \eqref{A} has been established, it was described either  in terms of a spectral curve of the form \eqref{spectral_curv0},  or in terms of solutions of  variational problems for vector-valued measures on the plane. Both characterizations, technically equivalent, can be used to build certain functions ($g$-functions) that are crucial for the asymptotic Riemann-Hilbert analysis. An advantage of the  variational approach is that the components of the vector-valued measure can be used to describe the asymptotic eigenvalue distribution of the random matrix ensemble explicitly, but so far the most general rigorous analysis available is still restricted to a symmetric situation.

In our previous work \cite{martinez_silva_critical_measures} we showed that if the variational problem of the type we consider here has a solution (in an appropriate sense), then there exists an associated spectral curve. The first part of the present work (Sections \ref{sec:vector_crit_measure} and \ref{section_qd}) provides a converse statement in the setup motivated by the matrix model \eqref{external_source_model} but not necessarily restricted to it. Namely, assuming that $V$ is an \textit{arbitrary} real polynomial and that the cubic equation \eqref{spectral_curv0} satisfies certain natural conditions, we prove the \textit{existence} of the solution of an associated variational problem involving three measures on the complex plane (vector-valued \textit{critical measure}). It turns out that the components of this measure exhibit a symmetry property, known as the  \textit{S-property}, which should ultimately characterize the support of the critical measure.

When compared to the just mentioned works related to \eqref{external_source_model}, the main novelty of our approach  is the introduction of a quadratic differential on the Riemann surface associated to \eqref{spectral_curv0} and the study of its critical graph. This idea has been developed in our previous works \cite{martinez_silva_critical_measures, martinez_silva2}, and also in \cite{bleher_silva}, and allows us to, starting solely from \eqref{spectral_curv0}, extract the support for the measures that constitute the solution to our variational problem. 

It turns out that in the situation presented here, the critical graph is sufficiently rigid so that 
\begin{itemize}
\item there are only two qualitatively different configurations that allow us to classify two asymptotic regimes, described as \textit{saturated} and \textit{unsaturated} regimes, see Section~\ref{sec:saturated}.
\item there are only two possible singular behaviors of the critical measures, corresponding to either square or cube roots in the density of its components, as explained in Theorem~\ref{thm:local_behavior} in Section~\ref{sec:vector_crit_measure}. The latter happens only in the unsaturated regime, and only with low powers of the cube root.
\end{itemize}

In the second part of our work (Sections \ref{sec:eigenvalues} and \ref{sec:constrained}), we apply our results to the analysis of the asymptotic density of the eigenvalues of the ensemble \eqref{external_source_model} with $\bm A$ as in \eqref{A}, comparing it with the totally symmetric situation considered by Bleher, Delvaux and Kuijlaars \cite{bleher_delvaux_kuijlaars_external_source}, when  the potential $V$ is an even polynomial and $\alpha = 1/2$. We show the equivalence of the variational problem for vector-valued measures with three components and the constrained vector equilibrium from \cite{bleher_delvaux_kuijlaars_external_source}. This yields an alternative description of the asymptotic distribution found in  \cite{bleher_delvaux_kuijlaars_external_source}, that extends the analogous result of \cite{aptekarev_lysov_tulyakov} from the quartic to any even potential $V$. An advantage of our electrostatic model is that it is valid for arbitrary polynomial $V$ and $\alpha\in (0,1)$, making it natural to expect that vector-valued critical measures describe the eigenvalue distribution even in the most general case, in accordance with the conjecture in \cite{aptekarev_lysov_tulyakov}.

One of the main powers of the above-mentioned Riemann-Hilbert approach is that it provides a comprehensive asymptotic description in different scaling regimes, including the limiting eigenvalue distribution, the local scaling limits such as the classical Airy and Sine kernel universality-type results, etcetera. But this power can be seen as a drawback as well, as if one is interested in only one of such quantities, one still needs to complete all technical steps that are firmly tied to the other quantities as well. 

As an alternative to the Riemann-Hilbert methods, several recent works have been developed to analyze the limiting eigenvalue distribution of random matrix models (in fact, these works deal with (subclasses of) determinantal point processes that encompass the eigenvalues of the model \eqref{external_source_model}) and its fluctuations \cite{breuer_duits_2017, lambert_2018, hardy_2015, breuer_duits_2016, hardy_2018}. In common, their analysis is based on the recurrence relation coefficients  for the functions used to construct the correlation kernel of the particle system. 

In the third and last part of this work (Sections~\ref{sec:general} and \ref{sec:generalcase}), we use a combination of the just-mentioned approach and of an analogue of a WKB analysis to study the limiting eigenvalue distribution from the perspective of the recurrence coefficients. We start by introducing a differential equation for the wave functions associated to the model, and for which its characteristic equation becomes a ``finite $N$'' version of the spectral curve. Assuming that the zeros of the average characteristic polynomial remain uniformly bounded (fact that we believe is always true, and needs not to be imposed as an additional hypothesis, see Remark~\ref{rem:boundedness}), we show that the associated recurrence coefficients remain bounded as well. This, in turn, yields uniform boundedness of  the coefficients in the finite $N$ spectral curve. As a consequence, when $N\to \infty$, we establish existence along any convergent subsequence  of the spectral curve studied in the first part, describing the limiting eigenvalue distribution of the matrix model.

In particular, this shows that in this situation the asymptotic eigenvalue distribution has either the same local behavior as in the one-matrix model (power of square-root behavior), corresponding to $\bm A= \bm 0$ and giving rise to Painlevé-type hierarchy limiting kernels, or  vanishes as a cubic root or the fifth power of the cubic root in the interior of its support. The cubic root behavior is associated to the known Pearcey kernel (this situation appears, for instance, when $n_1=n_2$, $a=1$ and $V(z)=z^2/2$), and although an example of a potential with limiting density of eigenvalues vanishing with power $5/3$ is implicit in the literature (see the comments after Theorem~\ref{thm:local_behavior} below), to our knowledge the corresponding limiting process has not yet been described in the literature. In short, our results mean that there is no higher-order Pearcey kernel at the bulk of the spectrum, other than the ones associated with $1/3$ and $5/3$ powers, nor powers of cubic root like behavior at the soft edges for the model \eqref{external_source_model}.

Without further ado, let us discuss our main findings.

\section{Statement of main results}\label{sec:main_results}

Throughout this paper, we use the notation $\R_+ := (0,+\infty)$,  $i\R := \{ ix:\, x\in \R\} $, 
\begin{equation} \label{notationCH}
\C _{\pm}:=\{ z\in \C \; \mid \; \pm \re z>0 \},\quad \HH_\pm:=\{ z\in \C \; \mid \; \pm \im z>0 \}.
\end{equation}
We define the zero counting measure of any polynomial $q$ by
\begin{equation}\label{counting_measure_definition}
\mu(q)\eqq \frac{1}{\deg q}\sum_{q(w)=0}\delta_w,
\end{equation}
where each zero on the right-hand side is counted according to its multiplicity, so that $ |\mu(q)|:=\mu(\mathbb C)=1$. We will understand by asymptotic zero distribution any limit of $\mu(q)$ in the weak-* sense (that we denote by $\weak$) as $\deg q\to \infty$. 

The Cauchy transform of any finite Borel measure $\nu$ on $\C$ is given by
\begin{equation}\label{def:Cauchytransf}
C^{\nu}(z)\eqq\lim_{\varepsilon\to 0+}\int_{|z-s|>\varepsilon} \frac{d\nu(s)}{s-z}, \quad z\in \C.
\end{equation}

\subsection{Vector critical measures from the spectral curve} \label{sec:vector_crit_measure}

\

\subsubsection{The spectral curve}

\

As we pointed out in the Introduction, our primary motivation is  the asymptotic behavior of the eigenvalues of the random matrices from the ensemble \eqref{external_source_model}--\eqref{A} in the asymptotic regime \eqref{asymptRegime}. This model is a particular case of a multiple orthogonal polynomial ensemble, and as discussed before, under mild conditions studying the global asymptotic regime of eigenvalues is equivalent to studying the weak-* convergence, as $N=n_1+n_2\to \infty$, of the zero-counting measure (see the definition \eqref{counting_measure_definition}) $\mu(\pi_{n_1, n_2})$ for the  average characteristic polynomial  \cite{hardy_2018}
\begin{equation} \label{averageCharPoly}
\pi_{n_1, n_2}(z)\eqq \E\det(z\bm I- \bm M).
\end{equation}
For  even polynomial potential $V$ and for $n_1=n_2=N/2$, the authors of \cite{bleher_delvaux_kuijlaars_external_source} found the spectral curve of the matrix model in the form of the algebraic equation  \eqref{spectral_curv0}, 
where $p_0$ and $p_1$ are polynomials that, in general, cannot be fully determined only from $V$ and $a$. Moreover, they established that   $\mu(\pi_{n_1, n_2}) \weak \nu $  as $N\to \infty$, and that  $\xi=\xi(z)=C^{\nu}(z)+V'(z)$ is one of the solutions of \eqref{spectral_curv0}, where  $C^{\nu}$ is the Cauchy transform of  $ \nu$, defined in \eqref{def:Cauchytransf}.

Such a spectral curve is a central object and starting point of our paper, so we turn it into a formal definition that synthesizes some general properties of such algebraic equations described in the existing literature.
\begin{definition}\label{definition_spectral_curve}
Given a polynomial $V$ of arbitrary degree $m\geq 2$, 
\begin{equation}\label{def:potential_v_coeff}
V(z)=\sum_{k=1}^m \frac{v_k}{k}z^k,  \quad v_1, \dots, v_{m}\in \R, \quad v_m\neq 0,  
\end{equation}
and parameters $(a, \alpha) \in \R_+ \times (0,1)$,  an associated {\it admissible spectral curve} is an irreducible algebraic equation of the form
\begin{equation}\label{spectral_curve}
F(\xi,z):=\xi^3 +p_2(z) \xi^2 + p_1(z) \xi + p_0(z)=0, 
\end{equation}
where $p_0, p_1$ and $p_2$ are real polynomials, with 
\begin{equation}\label{eq:normalization_p0}
\begin{split}
p_2(z) & = - V'(z) , \\
\deg p_1&  \leq m-2, \\
 p_0(z) &=v_m a^2 z^{m-1}+\left( v_{m-1}a^2+v_ma\left(2\alpha-1\right) \right)z^{m-2}+\text{lower degree terms},
\end{split}
\end{equation}
such that there exists a probability measure $\lambda$ on $\R$ for which the function
\begin{equation}\label{solution_cauchy_transform}
\xi (z)=C^{\lambda}(z)+V'(z)
\end{equation}
is a solution to \eqref{spectral_curve} for $z\in \mathbb C\setminus\R$.
\end{definition}

\begin{remark}
Let us insist that Definition~\ref{definition_spectral_curve} focuses on the algebraic equation \eqref{spectral_curve} only and does not assume a priori the existence of an underlying random matrix model. However, as the results from Sections \ref{sec:general} and \ref{sec:generalcase} show, admissibility is a natural assumption in the context of Random Matrix Theory (see also the models studied in \cite{bleher_kuijlaars_external_source_gaussian_I, bleher_kuijlaars_external_source_gaussian_III,  bleher_delvaux_kuijlaars_external_source, bleher_kuijlaars_external_source_multiple_orthogonal, aptekarev_bleher_kuijlaars_external_source_gaussian_II,aptekarev_lysov_tulyakov,bertola_et_al_external_source_I, bertola_et_al_external_source_II, aptekarev_lysov_tulyakov_2}).

Although playing a role in our approach, the conditions~\eqref{eq:normalization_p0} should not be seen as a restriction but rather as a normalization, fixing the potential $V$ and the parameters $\alpha$ and $a$.
The crucial condition is that one of the solutions of \eqref{spectral_curve}  can be written in terms of the Cauchy transform of a finite positive measure on $\R$. Algebraic equations admitting solutions in terms of Cauchy transforms of positive measures are rather special, appearing frequently in random matrix theory but also in other contexts, see  e.g.~\cite{BorceaBogvadShapiro2009, BorceaBogvadShapiroCorr,martinez_rakhmanov, shapiro_tater, shapiro_takemura_tater, shapiro_solynin_2017, bogvad_shapiro_2016, rao_edelman_2008, anderson_zeitouni_2008}. 
	
Notice that the model \eqref{external_source_model} is only well-defined if $m$ is even and $v_m>0$, but some of our main findings do not need such a requirement. Obviously, when we apply our results to this random matrix model, we will then impose these conditions {\it en route}.
\end{remark}

\subsubsection{Critical measures}
  
  \

Now we introduce the second ingredient relevant to the formulation of the main results of this section: the variational problem for vector  measures. For that, we describe both the family of the measures and the energy functional acting on them.

The mutual (logarithmic) energy of two finite Borel measures $\mu$ and $\nu$ on $\C$ is
$$
I(\mu,\nu)\eqq \iint \log\frac{1}{|s-z|}d\mu(s)d\nu(z);
$$
we refer to $I(\mu,\mu)$ simply as the energy of $\mu$.

For the real polynomial  $V$ as in \eqref{def:potential_v_coeff} define
\begin{equation} \label{defV1}
V_j(z):=V(z)-a_jz, \quad j=1,2, \qquad V_3(z):=V_2(z)-V_1(z)=2a z,
\end{equation}
where  we use the convention that $a_1=-a_2=a$, and set
\begin{equation}\label{definition_phis}
\phi_j(z):=\re V_j(z),\quad z\in \C,\quad j=1,2,3.
\end{equation}

We define by $\mathcal T$ the class of rectifiable contours $\Gamma$ that are symmetric under complex conjugation, intersect $\R$ at a single point and extend to $\infty$ along their two ends in such a way that
\begin{equation}\label{asymptotic_condition_class_T}
\lim_{\substack{z\to \infty \\ z\in \Gamma}} \re z=+\infty.
\end{equation}
Associated to a contour $\Gamma \in \mathcal T$ and a parameter $\alpha \in (0,1)$ we introduce $\mathcal M_\alpha (\Gamma)$, the set of vectors of non-negative measures $\vec\mu=(\mu_1,\mu_2,\mu_3)$ satisfying the following conditions:
\begin{itemize}
	\item Each component $\mu_j$ is a Borel measure with finite energy and for which $\int \phi_j d\mu_j$ is finite;
	\item $\ds{\supp\mu_1\subset \R }$,  $ \supp\mu_2\subset \R$ and $ \supp\mu_3 \subset \Gamma$;
	\item $\ds{|\mu_1|+|\mu_2|=1, \ |\mu_1| - |\mu_3|=\alpha \in (0,1)}$.
\end{itemize}
Notice  that any one of the two identities relating the masses of $\mu_1, \mu_2$ and $\mu_3$ can be equivalently replaced by
$$
|\mu_2|+|\mu_3|=1-\alpha.
$$

We introduce the \textit{interaction matrix}
\begin{equation}\label{def:matrixA}
\bm B=(\mathfrak b_{j,k})_{1\leq j,k\leq 3}=
\begin{pmatrix}
1 	  			& 1/2 	&  	-1/2 \\
1/2 		 	& 1		&  	1/2 \\
-1/2 			& 1/2	& 1
\end{pmatrix}.
\end{equation}
Its structure shows that the interaction between $\mu_1$ and $\mu_2$ or between $\mu_2$ and $\mu_3$  is of repulsive or Angelesco character, while between $\mu_1$ and $\mu_3$ the interaction is of attractive or Nikishin type, see e.g.~\cite{MR2647568}.

For $\alpha\in (0,1)$, $\Gamma\in \mathcal T$ and  for $\vec\mu\in\mathcal M_\alpha (\Gamma)$, the \textit{vector energy}
\begin{equation}\label{energy_three_measures}
E_S(\vec\mu)\eqq \sum_{j,k=1}^3 \mathfrak b_{j,k}I(\mu_j,\mu_k) + \sum_{j=1}^3 \int \phi_j \; d\mu_j,
\end{equation}
is  well defined and finite.
The \textit{vector equilibrium measure} of $\Gamma$ is the  vector $\vec\mu_\Gamma \in \mathcal M_\alpha(\Gamma)$ for which
$$
E_S(\vec\mu_\Gamma)=\inf_{\vec\mu\in\mathcal M_\alpha (\Gamma)} E_S(\vec\mu);
$$
if $m$ is even and $v_m>0$ then our conditions on $\mathcal M_\alpha (\Gamma)$ guarantee both existence and uniqueness of $\vec\mu_\Gamma$. 
Although the first two components of $\vec\mu_\Gamma$ live on $\R$, we refer to $\vec\mu_\Gamma$ as the vector equilibrium measure of $\Gamma$ instead of $(\R, \Gamma)$. 
We stress also that $\vec\mu_\Gamma$ does depend on $V$, $\alpha$ and $a$, but we do not make this dependence explicit in our notation.

Closely related to equilibrium measures is the notion of \textit{vector critical measures}, i.e.~saddle points of the energy functional $E_S(\vec\mu)$. To be more precise, given a complex-valued function $h\in  C_c^2(\C)$ (i.e., twice real differentiable with continuous second derivatives and compactly supported on $\C$), we denote by $\mu^t$ the pushforward of $\mu$ by the map $z\mapsto z+th(z)$, $t\in \R$, and set $\vec \mu^t=(\mu_1^t,\mu_2^t,\mu_3^t)$. 
\begin{definition}\label{definition_critical_measure}
	A vector of measures $\vec\mu\in \mathcal M_\alpha(\Gamma)$ is called \textit{critical} for the energy functional \eqref{energy_three_measures} if 
%
\begin{equation}\label{critical_measure_limit_definition}
\lim_{t\to 0} \frac{E_S(\vec \mu^t)-E_S(\vec \mu)}{t}=0,
\end{equation}
for every $h\in C_c^2(\C)$.
\end{definition}

See \cite{martinez_rakhmanov} for the notion of scalar critical measure and its applications, and its extension in \cite{martinez_silva_critical_measures} to the vector setting. 

Our first major result shows that each admissible spectral curve has a unique associated  critical vector measure for the energy functional \eqref{energy_three_measures}, and that its components  satisfy certain symmetry conditions. In its formulation we use the notion of the \textit{logarithmic potential} of a Borel measure $\mu$ on $\C$, 
$$
U^{\mu}(z)\eqq\int \log\frac{1}{|s-z|}d\mu(s), \quad z\in \C.
$$
\begin{thmx}\label{thm_existence_critical_measure}
Suppose that given $V$, $\alpha$ and $a$,  \eqref{spectral_curve} is an admissible spectral curve as described in Definition \ref{definition_spectral_curve}. Then there exists a unique vector critical measure $\vec \mu_*=(\mu_1^*,\mu_2^*,\mu_3^*) \in \mathcal M_\alpha(\Gamma_*)$ for some $\Gamma_*\in \mathcal T$, such that the admissibility condition  in Definition~\ref{definition_spectral_curve} is satisfied for
\begin{equation} \label{lambda_mu}
\lambda = \mu_1^*+\mu_2^*.
\end{equation}

Furthermore, there exists $x_*\in \R$ for which $\supp\mu_1\subset [x_*,+\infty)$ and $\supp\mu_2\subset (-\infty,x_*]$. In particular, these supports intersect in at most one point.

If $\mu_3^*\neq 0$, then $\supp\mu_3^*$ is a bounded and connected arc with connected complement in $\C$, symmetric with respect to $\R$, and such that  $\supp\mu_3^*\cap \HH_+$ is an analytic arc. Still when $\mu_3^*\neq 0$, either $\supp\mu_1^*\cap\supp\mu_2^*\cap\supp\mu_3^*$ consists of a single point (necessarily $x_*$ as above) or the supports of the measures are pairwise disjoint.

The components of the vector critical measure $\vec \mu_*$ satisfy the following variational identities: for $i=1, 2, 3$, 
	\begin{equation}\label{eq:variational_identitiesThm}
	 \sum_{k=1}^3 \mathfrak b_{i,k}U^{\mu_k^*}(z) +\frac{1}{2}\phi_i(z) \equiv \ell
	\end{equation}
along each connected component of $\supp \mu_i^*$, for some constant $\ell$ that might depend on the connected component.
	
Finally, on $\supp \mu_3^*$ the measure $\vec \mu_*$ exhibits the \emph{S-property}
\begin{equation}\label{vector_s_property}
\left( \frac{\partial }{\partial n_+}-\frac{\partial }{\partial n_-} \right)\left( \sum_{k=1}^3 \mathfrak b_{3,k}U^{\mu_k^*}(z) +\frac{1}{2}\phi_3(z) \right)=0,\quad z\in \supp\mu_3^*,
\end{equation}
where $n_\pm$ denote the two unit normal vectors to $\supp\mu_3^*$ at $z$ in the opposite directions. 
\end{thmx}

Obviously, in the result just stated if $\supp \mu_3^*=\emptyset$, the corresponding conditions in \eqref{eq:variational_identitiesThm} and \eqref{vector_s_property} are void. 

The contour $\Gamma_*$ will be more precisely described in Propositions \ref{prop:construction_s_curves_saturated} and \ref{prop:construction_trajectories_s_contour_unsaturated} below, in terms of trajectories of an associated quadratic differential.

\begin{remark}
The fact that expressions in \eqref{eq:variational_identitiesThm} are constant along each connected component of the support of the corresponding measure is a characterizing property of critical measures (see \cite[Theorem~1.8]{martinez_silva_critical_measures}). In the case of the equilibrium measure on $\Gamma_*$ we would have the equality of these constants on each connected component of $\supp \mu_i^*$. The latter is equivalent to imposing an additional, the so-called \textit{Boutroux condition} on the spectral curve \cite{bertola_boutroux}, see the detailed discussion in Section  \ref{sec:6}, especially Remark~\ref{boutrouxRem}. 

 On the other hand, not every critical measure is an equilibrium measure: there is an additional property that can be written as a set of inequalities satisfied by the left
 hand side of \eqref{eq:variational_identitiesThm}. For instance, it is possible to construct spectral curves, in the sense of Definition \ref{definition_spectral_curve} but with $a=0$, for which the critical measure in Theorem \ref{thm_existence_critical_measure} does not coincide with the corresponding equilibrium measure on $\R$. It is natural to expect from our techniques that this situation can happen even when $a>0$. In other words, the statement of Theorem~\ref{thm_existence_critical_measure} is sharp, and without a stronger condition on the spectral curve we cannot assure that the critical measure is in fact also an equilibrium measure. 
	
 When the matrix $\bm A=\bm 0$, the model \eqref{external_source_model}--\eqref{A} reduces to the classical Hermitian random matrix model. For such models, it is known that the spectral curve exists as an algebraic equation of degree $2$ (so that, actually, the expression \eqref{spectral_curve} becomes reducible, see Proposition~\ref{prop:exclusion_caseIII} below). In this situation, it follows from our arguments that $\mu^*_2=\mu^*_3=0$ and $\mu^*_1$ is the standard critical or equilibrium measure in the external field $V$, see also Remark~\ref{rmk:irreducible} below.
\end{remark}
  
It turns out that the admissibility of a spectral curve, described in Definition \ref{definition_spectral_curve}, limits the variety of possible local behaviors of the measure $\lambda$.
\begin{thmx}\label{thm:local_behavior}
Suppose that given $V$, $\alpha$ and $a$,  \eqref{spectral_curve} is an admissible spectral curve as described in Definition \ref{definition_spectral_curve}. Then the measure $\lambda=\mu_1^*+\mu_2^*$ is absolutely continuous with respect to the Lebesgue measure on $\R$, its support is a finite union of disjoint intervals, and its density vanishes on $\supp \lambda$ only at a finite number of points where it enjoys the following properties:
\begin{enumerate}[(i)]
\item If $x_0$ is a zero of $\frac{d\lambda}{dx}$ in the interior of $\supp\lambda$, then for some $c>0$ either 
\begin{equation}\label{sine_like_behavior}
\frac{d\lambda}{dx}(x)=c|x-x_0|^{k}(1+\Boh(x-x_0)),\quad x\to x_0,
\end{equation}
for a positive integer $k$, or for $\nu\in \{ 1, 5\}$, 
\begin{equation}\label{pearcey_like_behavior}
\frac{d\lambda}{dx}(x)=c|x-x_0|^{\frac{\nu}{3}}(1+\Boh(x-x_0)),\quad x\to x_0.
\end{equation}
Furthermore, the latter case happens only when $\mu_3^*=0$ and  $\supp\mu_1^* \cap \supp\mu_2^*=\{x_0\}=\{x_*\}$ with $x_*$ as in Theorem~\ref{thm_existence_critical_measure},  which determines this point uniquely.
\item If $x_0$ is an endpoint of $\supp\lambda$ then for some $c>0$ and $k\in \mathbb N\cup \{0\}$,
\begin{equation}\label{painleve_type_behavior}
\frac{d\lambda}{dx}(x)=c|x-x_0|^{\frac{2k+1}{2} }(1+\Boh(x-x_0)),\quad x\to x_0 \mbox{ along } \supp\lambda.
\end{equation}
\end{enumerate}
\end{thmx}

In a slight abuse of terminology, in (ii) above and in the following, we will usually speak about ``endpoints of  $\supp\lambda$'' meaning endpoints of a connected component of  $\supp\lambda$.

The local behaviors \eqref{sine_like_behavior}--\eqref{painleve_type_behavior} are believed to determine the local statistics of eigenvalues near $x_0$. The behaviors \eqref{sine_like_behavior} and \eqref{painleve_type_behavior} are known to arise in the one matrix model (when $\bm A=\bm 0$) and their ``complexified'' versions, and have been vastly studied in the literature \cite{claeys_kuijlaars_vanlessen, bleher_deano_painleve_I, bleher_delvaux_kuijlaars_external_source, bleher_kuijlaars_external_source_gaussian_III, bleher_its, bertola_tovbis_quartic_weight, duits_painleve_kernels, dai_xu_2019, dai_2018, atkin_claeys_mezzadri_2016, its_kuijlaars_ostensson_2008, duits_kuijlaars_painleve_I, bleher_deano_2016}. 

The behavior \eqref{pearcey_like_behavior} with $\nu=1$ is known to give rise to the Pearcey point process under suitable scaling of the eigenvalues, and as mentioned before is known to occur for $V(z)=z^2/2$ \cite{bleher_kuijlaars_external_source_gaussian_III,tracy_widom_2006_pearcey}. More recently, this type of behavior appeared in Wigner-type random matrix models as well \cite{ajanki_erdos_krugen_2017_cpam, ajanki_erdos_krugen_2017, erdos_kruger_schroder_2018}. The case $\nu=5$ is implicitly present but not discussed in the literature: for the quartic symmetric potential $V$ discussed in \cite[Figure~7.1]{bleher_delvaux_kuijlaars_external_source} and the unique bifurcation point in the parameter space (that is, the unique choice of parameters that belongs to the boundary of cases I, II and II), the discriminant of the spectral curve (see equation (7.1) therein) has only 2 simple zeros and a zero at the origin of multiplicity 10, which shows that the solution has the desired local asymptotics. 

The appearance of such behaviors in the theorem above is not surprising, but the fact that these are the only possible behaviors, in particular the non-existence of a criticality associated to the vanishing of the density with order $\nu/3$ and $\nu\neq 1,5$, is the new contribution in Theorem~\ref{thm:local_behavior}.

\subsection{Asymptotic eigenvalue distribution for the Hermitian matrix model with external source} 

\

We apply the results above to the Hermitian matrix model with external source when $\bm A$ is of the form \eqref{A}. So for Sections~\ref{sec:eigenvalues}--\ref{sec:general}, to ensure the model \eqref{external_source_model} is well defined, we also tacitly assume that $m=\deg V$ is even and $v_m>0$.

\subsubsection{The totally symmetric case and the constrained equilibrium} \label{sec:totallySym}\label{sec:eigenvalues}

\

As it was mentioned in the Introduction, when $V$ is an even polynomial, Bleher, Delvaux and Kuijlaars \cite{bleher_delvaux_kuijlaars_external_source} showed that the limiting eigenvalue distribution of the random matrix ensemble \eqref{external_source_model}--\eqref{A} in the asymptotic regime \eqref{asymptRegime} with  $\alpha=1/2$ is given by a certain probability measure $\nu_1$ on $\R$, characterizing this measure in terms of a constrained vector equilibrium that we describe next.

We denote by $\sigma$ the measure supported on $i\R$, absolutely continuous with respect to the arc-length and  with constant density
\begin{equation}\label{def:constrained_measure}
\frac{d\sigma}{|dz|}=\frac{a}{\pi}, \quad z\in i\R. 
\end{equation}

We then define $\mathcal M^c$ as the set of pairs of measures $\vec\nu=(\nu_1,\nu_2)$ satisfying additionally
\begin{itemize}
	\item Each component $\nu_j$ is a Borel measure with finite logarithmic energy;
	\item $\ds{\supp\nu_1\subset \R, \ \supp\nu_2=i\R}$;
	\item $\ds{|\nu_1|=1}$, $\ds{|\nu_2|=1/2}$, and $\nu_2\leq \sigma$ (i.e., $\sigma - \nu_2$ is a non-negative Boreal measure on $i \R$).
\end{itemize}

For any $\vec\nu =(\nu_1,\nu_2) \in \mathcal M^c$, we  consider the energy
$$
E^c(\vec\nu)=I(\nu_1,\nu_1)+I(\nu_2,\nu_2)-I(\nu_1,\nu_2)+\int (V(x)-a |x|)d\nu_1(x).
$$
Standard arguments from logarithmic potential theory show that there is a unique minimizer $\vec\nu^*=(\nu_1^*,\nu_2^*)$ of $E^c$, called the constrained \textit{vector equilibrium measure} in the class $\mathcal M^c$. One of the results of \cite{bleher_delvaux_kuijlaars_external_source} for the case when $V$ is an even polynomial and $\alpha=1/2$ is that  the limiting zero distribution for the average characteristic polynomial \eqref{averageCharPoly} is precisely the component $\nu_1^*$ of $\vec\nu^*$. They also found the spectral curve, with $\nu_1^*=\lambda$, writing expressions for the coefficients $p_0$ and $p_1$ in \eqref{spectral_curve} in terms of $\vec\nu^*$. 

As for $\nu_2^*$, the inequality constraint $\nu_2^*\leq \sigma$ can become an equality on a subinterval $[-y_*, y_*]$ of $i\R$, which is known as the \textit{saturated} case or regime for this random matrix model.

In the following theorem we give an alternative description of $\vec\nu^*$, now in terms of the vector critical measure $\vec\mu^*$, see Definition~\ref{definition_critical_measure}.
\begin{thmx}\label{thm_s_property_constrained_problem}
	For an even polynomial $V$, $\alpha=1/2$ and $a>0$, let \eqref{spectral_curve} be the admissible spectral curve corresponding to the totally symmetric random matrix model described above. If $\vec \mu_*=(\mu_1^*,\mu_2^*,\mu_3^*) \in \mathcal M_\alpha(\Gamma_*)$, $\Gamma_*\in \mathcal T$, is the corresponding vector critical measure  as in  Theorem \ref{thm_existence_critical_measure}, then $\supp (\mu_2 +\mu_3) $ is contained in the closure of the left half plane, and for the constrained equilibrium measure  $\vec \nu_*= (\nu_1^*,\nu_2^*)$ it holds that
	$$
	\nu_1^*=\mu_1^*+\mu_2^*,
	$$
	while $\nu_2^*$ is the balayage of $\mu_2^*+\mu_3^*$ onto the right half plane.
\end{thmx}

For the special case when $V$ is quartic, Theorem~\ref{thm_s_property_constrained_problem} was proven (for the vector equilibrium measures) by Aptekarev, Lysov and Tulyakov in \cite[Proposition~2.2]{aptekarev_lysov_tulyakov}.

\begin{remark}
The proof of this theorem is carried out in Section \ref{sec:constrained} by constructing a pair of measures $\vec \nu_*= (\nu_1^*,\nu_2^*)$ from the critical measure $\vec \mu_*=(\mu_1^*,\mu_2^*,\mu_3^*)$. This construction is valid for {\it any} symmetric spectral curve (see Definition~\ref{def:symmetric_spectral_curve}), and once this symmetric spectral curve is assumed to be the one in \cite{bleher_delvaux_kuijlaars_external_source} this constructed pair $\vec \nu_*= (\nu_1^*,\nu_2^*)$ immediately reduces to the constrained equilibrium pair.

Curiously, this construction could also be extended even for general $V$ and $\alpha\in (0,1)$, with the modification that $\nu_2^*$ would now be supported on a different, but asymptotically vertical contour, and the upper constraint is no longer a positive (not even real-valued) measure. This will be detailed in a future work.
\end{remark}

For the notion of balayage in the logarithmic potential theory, see e.g. \cite{AMFPotentialLectures} or \cite{stahl_totik_book}. 

\subsubsection{The general case} \label{sec:general}

\

The previous results show that our vector critical measures give an alternative description of the asymptotic eigenvalue distribution of the random matrix ensemble \eqref{external_source_model}--\eqref{A} in the limit \eqref{asymptRegime} with $V$ even and  $\alpha=1/2$. In this section we claim that under certain additional assumptions this description in fact still holds for general $V$ and $\alpha$.

It is convenient to introduce a multi-index notation. We denote $\vec n=(n_1,n_2)\in \N^2$, where for us $0\in \N$, and make $|\vec n|=n_1+n_2$. Also, we set $\vec e_1=(1,0)$ and $\vec e_2=(0,1)$. We consider sequences of multi-indices $(\vec n_k)$ with an up-right path structure, that is,
\begin{equation}\label{eq:up_right}
\begin{aligned}
&\bullet \vec n_k=(n_{k,1},n_{k,2}) \quad \mbox{with} \quad |\vec n_k|=k, \\[3pt] 
& \bullet \text{for every } k,\; \vec n_{k+1}=\vec n_{k}+\vec e_j \quad \mbox{for some } j\in \{1,2\}, \\[3pt]
&\bullet \lim_{k\to\infty} \frac{n_{k,1}}{k}=\alpha\in (0,1), \qquad \mbox{and} \\[3pt]
& \bullet \text{there exists } d\ge 2 \text{ for which } n_{k+d,j}\geq n_{k,j}+1, \text{ for } j=1,2 \text{ and every } k.
\end{aligned}
\end{equation}

The first three conditions above are the usual ones for an up-right path. The last condition says that the path stays within a fixed distance from the ray with slope $\alpha$; it will ensure that the sequence of multiple orthogonal polynomials indexed by $(\vec n_k)$ satisfies a recurrence relation with finitely many terms. The same condition also appeared in \cite{hardy_2015}, see Equation~(3.2) therein and also our comments after Theorem~\ref{thm:as_convergence} below.

A remarkable fact established in \cite{bleher_kuijlaars_external_source_multiple_orthogonal} is that for the random matrix ensemble \eqref{external_source_model}--\eqref{A}, the average characteristic polynomial 
$ \pi_{n_1, n_2}=\pi_{\vec n}$, defined in \eqref{averageCharPoly}, satisfies multiple orthogonality conditions. This means that $\pi_{\vec n}=P^{(N)}_{\vec n}$, $N=|\vec n|$ , where $(P^{(N)}_{\vec k})$ is the sequence of multiple orthogonal polynomials (MOP) with respect to the exponential weights $e^{-NV_1}$ and $e^{-NV_2}$:
\begin{equation}\label{defPNmultiple}
\begin{aligned}
& \int_\R P^{(N)}_{\vec k}(x)x^l e^{-NV_j(x)}dx=0,\quad l=0,\dots, k_j-1, \; j=1,2, \\
& P^{(N)}_{\vec k}(x)=x^{|\vec k|}+\mbox{lower order terms};
\end{aligned}
\end{equation}
we remind the reader that $V_1$ and $V_2$ were defined in \eqref{defV1}.

A closely related quantity is the vector of wave functions
\begin{equation}\label{def:wavefunctionPsi}
\renewcommand*{\arraystretch}{1.5}
\bm \Psi_{\vec n}(z)= 
\begin{pmatrix}
\pi_{\vec n}(z) \\
P^{(N)}_{\vec n-\vec e_1}(z) \\
P^{(N)}_{\vec n-\vec e_2}(z)
\end{pmatrix}
e^{-NV(z)}
=
\begin{pmatrix}
P^{(N)}_{\vec n}(z) \\
P^{(N)}_{\vec n-\vec e_1}(z) \\
P^{(N)}_{\vec n-\vec e_2}(z)
\end{pmatrix}
e^{-NV(z)}
,\quad \quad |\vec n|=N=n_1+n_2.
\end{equation}

\begin{thmx}\label{prop:wave_functions}
The vector of wave functions $\bm \Psi_{\vec n}$ satisfies the first-order differential equation
\begin{equation} \label{ODEmatrix}
\frac{d}{dz}\, 	\bm \Psi_{\vec n}(z) = N \bm R_{\vec n}(z) \bm \Psi_{\vec n}(z), \quad z\in \C,
\end{equation}
where $\bm R_{\vec n}$ is a matrix-valued polynomial (depending on the multi-index $\vec n$) of the form 
\begin{equation}\label{Rexplicit}
\bm R_{\vec n}(z)=
\begin{pmatrix}
- V'(z)+\Boh(z^{m-2}) & \Boh(z^{m-2}) & \Boh(z^{m-2}) \\
\Boh(z^{m-2}) & \Boh(z^{m-3}) & \Boh(z^{m-3}) \\
\Boh(z^{m-2}) & \Boh(z^{m-3}) & \Boh(z^{m-3}) 
\end{pmatrix}, \quad z\to\infty,
\end{equation}
and $V$ was defined in \eqref{def:potential_v_coeff}. Moreover, the characteristic polynomial of $\bm R_{\vec n}$ takes the form
\begin{equation}\label{eq:finite_spectral_curve}
\det\left(\xi\bm I+\bm R_{\vec n}\right)=\xi^3+q_{\vec n}^{(2)}(z)\xi^2+q_{\vec n}^{(1)}(z)\xi+q_{\vec n}^{(0)}(z),
\end{equation}
where the coefficients $q_{\vec n}^{(j)}=p_j$, $j=0,1,2$, satisfy the conditions \eqref{eq:normalization_p0} with $\alpha=n_1/N$.
\end{thmx}

 \begin{remark}
 	In the expressions \eqref{Rexplicit}--\eqref{eq:finite_spectral_curve} all terms of the form $\Boh(z^{k})$ may depend on $\vec n$.
 	
For the quadratic potential (shifted Gaussian ensemble) obtained when $m=2$, most terms in the expansion \eqref{Rexplicit} can actually be computed exactly (see \eqref{eq:expansion_R_D} below) and give 
$$
\bm R_{\vec n}(z)=
\begin{pmatrix}
- V'(z) & 0 & 0 \\
\Boh(1) & -a_1  & 0 \\
\Boh(1) & 0 & -a_2
\end{pmatrix}
$$
which is enough to determine the coefficients of $\det (\xi\bm I + \bm R_{\vec n}(z) )$ uniquely. 
 \end{remark}

The main consequence of our findings for the model \eqref{external_source_model} is our next result. To state it, recall that for the model \eqref{external_source_model} we assume $V$ to be a polynomial of even degree $m\geq 2$, $\bm A$ to be of the form \eqref{A} and by $N\to\infty$ we mean under the regime \eqref{asymptRegime}. 

\begin{thmx}\label{thm:as_convergence}
Suppose that $(\vec n_N)$ is a sequence of multi-indices satisfying \eqref{eq:up_right} with corresponding value $d$, and denote by $ \widetilde \pi_{\vec n_N}$ the   average characteristic polynomial  \eqref{averageCharPoly} for the modified  random matrix ensemble \eqref{external_source_model}--\eqref{A}, where $N$ in \eqref{external_source_model} is replaced by $N-\max \{m-2,d\}$. If all the zeros of the sequence  $(\widetilde \pi_{\vec n_N})$ are contained in a fixed compact set of $\C$, then all the zeros of $(\pi_{\vec n_N})$ remain uniformly bounded as $N\to \infty$, and for any limiting zero distribution $\lambda$ of a subsequence $(\pi_{\vec n_k})$ of $(\pi_{\vec n_N})$, 
$$
\lambda = \lim_{k\to\infty}\mu(\pi_{\vec n_k}),
$$
the function
$$
\xi (z)=C^{\lambda}(z)+V'(z)
$$
is a solution of a cubic equation \eqref{spectral_curve} that constitutes an admissible spectral curve in the sense of Definition \ref{definition_spectral_curve}. Furthermore, this cubic equation is obtained as a limit of the finite $N$ spectral curve \eqref{spectral_curve} along a subsequence of $(\vec n_k)$.
\end{thmx}

With \eqref{defPNmultiple} in mind, the average characteristic polynomial $\widetilde \pi_{N}$ is actually the MOP $\widetilde \pi_{N}=P^{(N-\kappa)}_{\vec n_N}$, with $\kappa=\max \{m-2,d\}$. So the assumption of Theorem~\ref{thm:as_convergence} is that the zeros of the sequence of MOP's $(P^{(N)}_{\vec n_{N+\kappa}})$ slightly off the given up-right path $(\vec n_N)$ remain bounded. Nevertheless, we opted for the formulation above as it is written in terms of the average characteristic polynomial, which has a more concrete physical meaning.

As shown in Proposition~\ref{prop:realzeros} below, the zeros of $\pi_{n_N}$ are all real, so the measure $\lambda$ in the Theorem above is always real as well. Furthermore, the boundedness of zeros of $P^{(N)}_{\vec n_{N+d}}$ implies that the coefficients for the recurrence relation satisfied by the sequence $(P^{(N)}_{\vec n_k})$ remain appropriately bounded as $N\to\infty$. As a consequence of Hardy's result \cite[Corollary~1.5]{hardy_2015} along the same subsequence $k$ of $N$ any sequence of spectral measures
$$
\nu_k = \frac{1}{k}\sum_{j=1}^k \delta_{\lambda_k}
$$
for the (random) eigenvalues $\lambda_1,\hdots,\lambda_k$ of the model \eqref{external_source_model} converges almost surely to $\lambda$ as well. So the spectral curve is, in fact, describing not only the limiting measure for the average characteristic polynomial, but also the density of states for the random eigenvalues, again along subsequences at least.

As it follows from Theorem~\ref{thm:as_convergence} above, the characteristic equation for $\bm R_{\vec n}$ in \eqref{Rexplicit}, obtained by equating the expression in \eqref{eq:finite_spectral_curve}  to $0$, can be seen as a finite $N$ version of the spectral curve. Similar versions of it (and also of the ODE \eqref{ODEmatrix}) for the case $\bm A=\bm 0$ and also for two-matrix models have been obtained in the past by Bertola, Eynard and Harnad \cite{bertola_eynard_harnad_duality_multi_matrix, bertola_eynard_harnad_2003_diff_systems, bertola_eynard_harnad_2003_partition, bertola_eynard_harnad_2006_semiclassical, bertola_eynard_harnad_2003_duality}.  

\begin{remark}\label{rem:boundedness}
	As it was mentioned above, the uniform boundedness of the zeros of $(\pi_{\vec N})$ is equivalent to the boundedness of the recurrence coefficients of  the associated sequence of multiple orthogonal polynomials $(P_{\vec n_N})$, and thus, we work in a framework and under hypotheses analogous to those in recent developments on fluctuations of eigenvalues of large random matrices \cite{breuer_duits_2017, lambert_2018, hardy_2015, breuer_duits_2016, hardy_2018}.
	
We strongly believe that under our assumptions on $V$ and $\bm A$ and the explicit scaling with $N$ in \eqref{external_source_model} the zeros of $(\pi_{\vec N})$ remain bounded, so the conclusions of Theorem~\ref{thm:as_convergence} should always take place.  In fact, boundedness was established already in the Gaussian, $V(x)=x^2/2$, and in the symmetric cases, $V(-x)=V(x)$ and $\alpha=1/2$, at least under regularity assumptions, see \cite{bleher_delvaux_kuijlaars_external_source}. 

\end{remark}

\subsection{Overview of the paper}

\

The core of the paper is Section~\ref{section_qd}, culminating in the proof of Theorems~\ref{thm_existence_critical_measure} and \ref{thm:local_behavior}. This is achieved in several steps, as follows:
\begin{itemize}
	\item Section~\ref{section_RS} is a preliminary discussion on the construction of the Riemann surface $\mathcal R$ associated to the spectral curve \eqref{spectral_curve}.

	\item Sections~\ref{sec:canonical_qd}--\ref{sec:saturated} contain our most novel contributions, namely the construction of a quadratic differential $\varpi$ on $\mathcal R$  and a detailed study of its trajectories, which will ultimately describe the supports of the measures in Theorem~\ref{thm_existence_critical_measure}.

	\item In Section~\ref{sec:local_behavior} we explore the consequences of the structure of trajectories of $\varpi$ for the local behavior of the density of $\lambda$ for the spectral curve \eqref{spectral_curve}.

	\item In Section~\ref{sec:riemann_surface_II}, we return to the Riemann surface $\mathcal R$, now constructing its sheet structure in a very explicit way, once again thanks to the study of the trajectories of $\varpi$.

	\item Section~\ref{sec:6} combines the results of the previous sections, proving Theorems~\ref{thm_existence_critical_measure} and \ref{thm:local_behavior}.
\end{itemize}

In Section~\ref{sec:constrained} we analyze the spectral curve under symmetry conditions of Section~\ref{sec:eigenvalues}; at the end of this section Theorem~\ref{thm_s_property_constrained_problem} is proved.

In Section~\ref{sec5} we prove some more or less known facts about zeros and recurrence relations satisfied by the polynomials $P_{\vec k}^{(N)}$ defined by \eqref{defPNmultiple}.

Section~\ref{sec:generalcase} is devoted to the proof of our last two main theorems. In Section~\ref{sec:rhp} we use the Riemann-Hilbert problem characterizing the orthogonality \eqref{defPNmultiple} as an algebraic tool to prove Theorem~\ref{prop:wave_functions}. The calculations to get \eqref{ODEmatrix} are somewhat standard, but we push them a little further to obtain also the properties of the matrix of coefficients $\bm R_{\vec n}$, as  claimed in Theorem~\ref{prop:wave_functions}.

Last but not least, Section~\ref{sec:counting_measures} contains the proof of Theorem~\ref{thm:as_convergence}. It involves an asymptotic analysis of the coefficients in the finite $N$ spectral curve \eqref{eq:finite_spectral_curve}, based on the relation between $\bm R_{\vec n}$ and the recurrence coefficients, relation which is implicitly given through the Riemann-Hilbert problem for both type I and II multiple orthogonal polynomials.


\section{Construction of the vector critical measures}\label{section_qd}

In this section, our starting point is an admissible spectral curve  \eqref{spectral_curve}, see Definition~\ref{definition_spectral_curve}. Our main goal is to prove Theorems~\ref{thm_existence_critical_measure} and \ref{thm:local_behavior}, which will require several steps. 

\subsection{The Riemann Surface associated to the Spectral Curve} \label{section_RS}

\

For each $z\in \C$ the equation \eqref{spectral_curve} has three solutions (not necessarily distinct) that we denote by $\xi_1(z),\xi_2(z)$ and $\xi_3(z)$, recalling that $\xi_1(z)$ is the fixed solution \eqref{solution_cauchy_transform}. 

\begin{prop}
The functions $\xi_1,\xi_2$ and $\xi_3$ admit an expansion of the form
\begin{equation}\label{asymptotics_xi}
\begin{aligned}
\xi_1(z) & =V'(z)-\frac{1}{z}+\Boh(z^{-2}),\\
\xi_2(z) & =a+\frac{\alpha}{z}+\Boh(z^{-2}),\\
\xi_3(z) & =-a+\frac{1-\alpha}{z}+\Boh(z^{-2}),
\end{aligned}
\quad  \qquad \mbox{ as } z\to\infty.
\end{equation}
In particular, these solutions are meromorphic at $z=\infty$, and the measure $\lambda$ from \eqref{solution_cauchy_transform} has compact support.
\end{prop}
\begin{proof}
The discriminant of $F$ with respect to $\xi$ is
$$
\discr_\xi (F)=4(V')^3p_0-27p_0^2-18V'p_0p_1+(V')^2p_1^2-4p_1^3=4v_m^4a^2z^{4m-4}+\Boh(z^{4m-3}),
$$
so in particular $\discr_\xi (F)>0$ for $z\in \R$ with sufficiently large absolute value, implying that $\xi_1,\xi_2$ and $\xi_3$ are not branched near $z=\infty$. This means that $\xi_1$ is meromorphic at $z=\infty$, so we also get that $\lambda$ has to be compactly supported. Furthermore, these solutions then admit a Laurent expansion around $z=\infty$, which in virtue of \eqref{solution_cauchy_transform} take the form
\begin{equation}\label{eq:expansion_xis_prel}
\begin{aligned}
\xi_1(z) & =V'(z)-\frac{1}{z}+\Boh(z^{-2}),\\
\xi_2(z) & =\kappa_2 z^{m_2}+\eta_2 z^{m_2-1}+\Boh(z^{m_2-2}),\\
\xi_3(z) & =\kappa_3 z^{m_3}+\eta_3 z^{m_3-1}+\Boh(z^{m_3-2}),
\end{aligned}
\quad  \qquad \mbox{ as } z\to\infty.
\end{equation}
for some $m_2,m_3\in \Z$ and real numbers $\kappa_2,\kappa_3,\eta_2,\eta_3$, with $\kappa_2,\kappa_3\neq 0$ (so, at this point, the second and the third equation in \eqref{eq:expansion_xis_prel} are just to set up the notation).

Being the solutions to the algebraic equation \eqref{spectral_curve}, these functions $\xi_1,\xi_2$ and $\xi_3$ satisfy the relations
\begin{equation}\label{alg_relations}
\left\{
\begin{aligned}
& \xi_1(z)+\xi_2(z)+\xi_3(z)=V'(z), \\
& \xi_1(z)\xi_2(z)+\xi_1(z)\xi_3(z) + \xi_2(z)\xi_3(z)=p_1(z), \\
& \xi_1(z)\xi_2(z)\xi_3(z)=-p_0(z).
\end{aligned}
\right.
\end{equation}

Using \eqref{eq:expansion_xis_prel} to expand the LHS of the last equation in \eqref{alg_relations} and then comparing with the leading coefficient of the RHS, we arrive at the identities
$$
m_2=-m_3,\quad \kappa_2\kappa_3=-a^2.
$$
Now, from the first equation in \eqref{alg_relations} we get that the polynomial part of $\xi_2$ near $z=\infty$ equals minus the polynomial part of $\xi_3$, and consequently we must have $m_2=m_3=0$ and also $\{\kappa_2,\kappa_3\}=\{a,-a\}$. Thus, imposing w.l.g. that $\kappa_2=a$, we get $\kappa_3=-a$, and a further use of the first equation in \eqref{alg_relations} also provides the relation 
\begin{equation}\label{eq:aux_proof_coeff_expansion_xi}
\eta_2+\eta_3=1.
\end{equation}
Using again the last equation in \eqref{alg_relations} and comparing with \eqref{eq:expansion_xis_prel}, we see that we should have
$$
p_0(z)=a^2v_mz^{m-1}+(a^2v_{m-1}+av_m(\eta_2-\eta_3))z^{m-2}+\Boh(z^{m-3}),
$$
so comparing with \eqref{eq:normalization_p0} this means that
$$
\eta_2-\eta_3=2\alpha-1.
$$
Combining it with \eqref{eq:aux_proof_coeff_expansion_xi} we finally obtain $\eta_2=\alpha$, $\eta_3=1-\alpha$, as claimed.
\end{proof}

Because $C^{\lambda }$ satisfies an algebraic equation and $\supp\lambda \subset \R$, we get (see for instance \cite{borcea_et_al_algebraic_cauchy_transform}) that $\supp\lambda $ is the union of a finite number of bounded intervals, namely
\begin{equation}\label{eq:intervals_support_mu0}
\supp\lambda =\bigcup_{k=1}^l [a_k,b_k],\quad a_1<b_1<\hdots<a_l<b_l.
\end{equation}
Since $F(\xi, z)$ in \eqref{spectral_curve} is monic in $\xi$, we conclude that $\xi_1$, $\xi_2$ and $\xi_3$ can have poles only at $\infty$, and actually by \eqref{asymptotics_xi}, only $\xi_1$ has a pole at $\infty$. Hence, the measure $\lambda $ is absolutely continuous with respect to the Lebesgue measure on $\R$ with a bounded density (the latter follows from Plemelj's formula applied to \eqref{solution_cauchy_transform}, see also \eqref{Sokhotsky-Plemelj} below). 

We construct the associated Riemann surface $\mathcal R$ to \eqref{spectral_curve} as a three-sheeted branched cover of $\overline \C$
$$
\mathcal R=\mathcal R_1\cup\mathcal R_2\cup \mathcal R_3,
$$
where $\mathcal R_k$'s are copies of $\overline \C$ cut along certain arcs (the branch cuts), chosen in such a way that $\xi_k$ is meromorphic on $\mathcal R_k$. The branch cut structure for $\mathcal R_1$ is completely determined by \eqref{solution_cauchy_transform}, so that
\begin{equation}\label{construction_sheet_1}
\mathcal R_1=\overline \C\setminus \supp\lambda =\overline\C\setminus \bigcup_{k=1}^l [a_k,b_k].
\end{equation}

The Equation  \eqref{spectral_curve} may have nonreal branch points, and in virtue of \eqref{symmxi1} these are necessarily branch points of $\xi_2$ and $\xi_3$ but not of $\xi_1$. Because \eqref{spectral_curve} is an equation with real coefficients for $z$ real, these nonreal branch points come in complex conjugate pairs. We connect each such pair by a simple and bounded real-symmetric contour; these countours will play the role of nonreal branch cuts. There are finitely many such pairs, so we can make sure that these branch cuts are chosen in such a way that the upper and lower half planes minus these branch cuts is simply connected. This way, the branches of solutions $\xi_2$ and $\xi_3$ labeled as the series \eqref{asymptotics_xi} admit analytic continuations to the upper and lower half planes, and we claim that these satisfy
\begin{equation}\label{eq:realsymmetry_sols}
\xi_{j}(\overline z)=\overline{\xi_{j}(z)},\quad j=1,2,3,
\end{equation}
whenever $z$ is nonreal and does not belong to one of the branch cuts. Indeed, the Equation \eqref{spectral_curve} is cubic in $\xi$, has real coefficients and its discriminant is
\begin{multline*}
-27p_0(z)^2-4p_1(z)^3+18p_0(z)p_1(z)p_2(z)+p_1(z)^2+p_2(z)^2-4p_0(z)p_2(z)^3 \\=4v_m^4 a^2 z^{4m-4}+\Boh(z^{4m-5}), \quad z\to\infty,
\end{multline*}
where for the identity we used \eqref{eq:normalization_p0}. This shows that the discriminant is positive for $z$ real and sufficiently large, thus proving \eqref{eq:realsymmetry_sols} for $z$ large, and by analytic continuation \eqref{eq:realsymmetry_sols} follows.

We construct the sheets $\mathcal R_2$ and $\mathcal R_3$ using the analytic continuation just explained. Hence, these two sheets are connected through nonreal branch cuts (if any), and also along intervals of the real line. At this stage, the branch cuts on the real line either connect two of the sheets $\mathcal R_1,\mathcal R_2$ and $\mathcal R_3$, or possibly the three of them.

By \eqref{solution_cauchy_transform}, we have 
	\begin{equation}\label{symmxi1}
\xi_{1}(\bar z)=\overline{ \xi_{1}(z)}, \quad z \in \C\setminus \supp\lambda,
	\end{equation}
and in combination with \eqref{eq:realsymmetry_sols} this yields
\begin{equation}\label{symmxi123}
\xi_{j,+}(x)=\overline{\xi_{j,-}(x)}, \quad j=1,2,3,
\end{equation}
for any $x \in \R$ and not on the nonreal branch cuts; for the same values of $x$ these boundary values $\xi_{j,\pm}$ are continuous as function on $\R$. Finally, at points of intersection of the nonreal branch cuts with the real line, Equation~\eqref{symmxi123} is understood in the sense of lateral limits.

But we can say more:

\begin{lem}\label{lem_imaginaryperiods}
	Let $p\in \R$ be a regular point.
	
	If $p \in \R \setminus \supp\lambda$ then 
\begin{equation}\label{symm1}
	\xi_{1}(p) \in \R \quad \text{and} \quad \im \xi_{2}(p)=- \im \xi_{3}(p).
	\end{equation}
	If $p \in   \supp\lambda$ then  there exists a permutation $j, k$ of $\{2, 3 \}$ (depending on $p$)  such that    
	\begin{equation}\label{lem_orthogonal_paths_eq_0}
	   \xi_{1\pm}(p)=\overline{\xi_{1\mp}(p)}=\xi_{j \mp}(p)=\overline{\xi_{j\pm}(p)}, \quad     \xi_{k\pm }(p)\in \R,
	\end{equation}
	and
	\begin{equation}\label{Sokhotsky-Plemelj}
\lambda'(p):=	\frac{d\lambda }{dx}(p)=\frac{1}{2\pi i}(\xi_{1+}(p)-\xi_{1-}(p))=\frac{1}{\pi}\im\xi_{1+}(p) .
	\end{equation}
	Moreover, at such $p \in   \supp\lambda$,
	\begin{equation}\label{prep_cycles}
\frac{1}{2\pi i}	\left[ \left( \xi_2-\xi_3\right)_+(p) - \left( \xi_2-\xi_3\right)_-(p) \right] =\frac{1}{ \pi  }	\Im \left( \xi_2-\xi_3\right)_+(p)  =
\begin{cases}
-\lambda'(p), & \text{if in \eqref{lem_orthogonal_paths_eq_0}, $j=2$,} \\
\lambda'(p), & \text{if in \eqref{lem_orthogonal_paths_eq_0}, $j=3$.}
\end{cases}
	\end{equation}
\end{lem}
\begin{proof}
	\eqref{symm1} as well as the identity $ \xi_{1\pm}(p)=\overline{\xi_{1\mp}(p)}$ for $p \in   \supp\lambda$ are direct consequences of \eqref{symmxi1}  and that \eqref{spectral_curve} has real coefficients. 
	We have observed already that  $\lambda $ is absolutely continuous with respect to the Lebesgue measure on $\R$ with a bounded density $d\lambda/dx$ that can be recovered via the Sokhotsky-Plemelj's formula, which yields \eqref{Sokhotsky-Plemelj}. 
	
Next we verify \eqref{lem_orthogonal_paths_eq_0}. Because in this case we are assuming $p\in \supp\lambda$ is a regular point (that is, not a branch point of \eqref{spectral_curve}), we learn from the already proven Formula \eqref{Sokhotsky-Plemelj} that the boundary values $\xi_{1+}(p)$ and $\xi_{1-}(p)$ do not coincide, so $\xi_{1+}(p)$ coincides with at least one among $\xi_{2-}$ and $\xi_{3-}$, say $\xi_{1+}(p)=\xi_{2-}(p)$ (the case when $\xi_{1+}(p)=\xi_{3-}(p)$ being analogous). Combined with the symmetry \eqref{symmxi123} we obtain that $\xi_{1-}(p)=\xi_{2+}(p)$. Using again Formula~\eqref{Sokhotsky-Plemelj} we see that $\xi_{1+}$ and $\xi_{2+}$ have nonzero imaginary parts, and because the Equation \eqref{spectral_curve} is real and cubic, we have to have the other solution real, that is, $\xi_{3+}\in \R$. This completes the proof of \eqref{lem_orthogonal_paths_eq_0}.

	
To conclude, using once more the symmetry \eqref{symmxi123} we see that if in \eqref{lem_orthogonal_paths_eq_0} $j=2$ then $\xi_{3+}(p) = \xi_{3-}(p)$, and $\xi_{2+}(p) = \xi_{2-}(p)$ otherwise. Using this in \eqref{prep_cycles} we see that the left hand side expresses $-  \lambda'(p)$ in the former case, and $   \lambda'(p)$ otherwise.	
\end{proof}

The global solution $\xi$ to \eqref{spectral_curve} is  constructed by taking  
\begin{equation}\label{eq:global_solution}
\restr{\xi}{\mathcal R_j}=\xi_j, \quad j=1,2,3.
\end{equation}
We denote by $\pi:\mathcal R\to \overline \C$ the canonical projection from $\mathcal R$ to $\overline \C$. For a point $q\in \overline \C$, we denote its preimage by $\pi$ on $\mathcal R_j$ by $q^{(j)}$. That is, $q^{(j)}$ is uniquely determined by
$$
\mathcal R_j\cap \pi^{-1}(q)=\{q^{(j)}\}.
$$

Obviously, if $q$ is not a branch point, then the three points $q^{(1)}$, $q^{(2)}$ and $q^{(3)}$ are distinct, but if $q$ is a branch point connecting the sheets $\mathcal R_j$ and $\mathcal R_k$ then $q^{(j)}=q^{(k)}$. Moreover, if $U\subset \overline \C$, we denote
$$
U^{(k)}=\pi^{-1}(U)\cap \mathcal R_k, \quad k=1,2,3,
$$
and if $a,b$ are any points on $\overline \C$, we denote by $[a,b]$ the non-oriented straight line segment with endpoints $a$ and $b$, so that $[a,b]$ also makes sense for $a, b \in \R$, $b<a$, and in our notation the sets $[a,b]$ and $[b,a]$ are the same. We also set
$$
[a^{(k)},b^{(k)}]:=[a,b]^{(k)}.
$$
Similar notation is used also for open intervals.
Finally, if $V\subset \mathcal R_k$ for some $k$, then we denote by $V^*$ its complex conjugate on the same sheet $\mathcal R_k$.

\subsection{The canonical quadratic differential \texorpdfstring{$\varpi$}{}}\label{sec:canonical_qd}

\

Associated to the Riemann surface $\mathcal R$ we define the canonical quadratic differential $\varpi$ through
\begin{equation}\label{definition_varpi}
\varpi= 
\begin{cases}
-(\xi_2(z)-\xi_3(z))^2dz^2, & \mbox{ on } \mathcal R_1, \\
-(\xi_3(z)-\xi_1(z))^2dz^2, & \mbox{ on } \mathcal R_2, \\
-(\xi_1(z)-\xi_2(z))^2dz^2, & \mbox{ on } \mathcal R_3.
\end{cases}
\end{equation}

In the earlier 2000's, some works in the literature explored quadratic differentials in the context of asymptotic analysis, but in all of them they only defined quadratic differentials on (subdomains of) the complex plane \cite{baik_deift_mclaughliun_miller_zhou_2001, kuijlaars_stahl_vanassche_wielonsky_2003, martinez_orive_jacobi_single_contour, martinez_thabet_quadratic_differentials, bertola_boutroux, bertola_mo} . To our knowledge, this explicit construction of $\varpi$ on the associated Riemann surface $\mathcal R$ as above in terms of solutions to an algebraic equation first appeared, independently and almost simultaneously, in our previous work \cite{martinez_silva_critical_measures} and in the paper  \cite{AptekarevTulyakov2016} by Aptekarev and Tulyakov.  Since then, it has found applications in different contexts, see also \cite{bleher_silva, martinez_silva2, aptekarev_vanassche_yatsselev}. This quadratic differential will play a substantial role in what comes next. A general account on the theory of quadratic differentials can be found in the books by Jenkins \cite{jenkins_book}, Pommerenke \cite[Chapter~8]{Pommerenkebook} and Strebel \cite{strebel_book}. The general theory and background reviewed in \cite[Appendix~B]{martinez_silva_critical_measures} are sufficient for our purposes; in fact, we closely follow the notation and notions introduced therein.

Recall that a zero/pole of a quadratic differential defined in local coordinates as $f(z)dz^2$ is simply a zero/pole of the function $f$, its multiplicity coincides with that as a zero/pole of $f(z)$, and this definition is independent of the choice of the local parameter $z$. Zeros and poles of a quadratic differential are its \textit{critical} points, all the rest are its \textit{regular} points. By the definition of $\varpi$, the projection by $\pi$  of the zeros of $\varpi$ onto $\C$ coincides with the zeros of the discriminant of \eqref{spectral_curve}. As a consequence, the poles and some of the zeros of $\varpi$, together with their multiplicities, can actually be easily described. 
\begin{itemize}
\item $\varpi$ has a pole of order $4$ at $\infty^{(1)}$ and poles of order $2m+2$ at $\infty^{(2)}$ and $\infty^{(3)}$, and no other poles on $\mathcal R$;
\item It has zeros at each of the points $a_j^{(k)}$, $b_{j}^{(k)}$, $k=1,2,3$. 
\end{itemize}

The multiplicities of the poles $\infty^{(1)},\infty^{(2)}$ and $\infty^{(3)}$ are obtained from the asymptotic expansion \eqref{asymptotics_xi}.
The multiplicity of each of the zeros at the endpoints of $\supp\lambda $ is determined by the local behavior of the density of $\lambda $ (see Lemma \ref{lem_imaginaryperiods}). 
For instance, if
$$
\frac{d\lambda }{dz}(z)=\const  |z-q|^{1/2}(1+\Boh(z)),\quad z\to q \in \{a_1,b_1,\hdots,a_l,b_l\},
$$
then $\xi_1(q)=\xi_j(q)$ for some $j\in \{2,3\}$. If in addition $\xi_1(q)\neq \xi_k(q)$ for $k\in \{2,3\}\setminus \{j\}$, then in such a case $q^{(k)}$ is a zero of multiplicity $1$ of $\varpi$ and the branch point $q^{(1)}=q^{(j)}$ is a zero of multiplicity $2$. Points where the density of $\lambda $ vanish with a higher order, or also higher order branch points, can be analyzed in a similar manner. We will perform such calculations when the need arises, and refer the reader to \cite[Section~4.4]{martinez_silva_critical_measures} for an explicit calculation as well.

\subsection{Trajectories and orthogonal trajectories of \texorpdfstring{$\varpi$}{varpi}} \label{sec:trajandorthtraj}

\

Locally, any quadratic differential is the square of a meromorphic differential, so we can talk about integration of $\sqrt{-\varpi}$ along curves on $\mathcal R$. A curve $\gamma \subset \mathcal R$ that does not contain critical points (poles or zeros) of $\varpi$ is called an {\it arc of vertical trajectory} (or simply arc of trajectory) of $\varpi$ if
\begin{equation}\label{eq:definition_trajectory}
\re \int_{\widetilde\gamma} \sqrt{-\varpi} \equiv \const
\end{equation}
along any subarc ${\widetilde\gamma}\subset \gamma$. Similarly, $\gamma \subset \mathcal R$ is an {\it arc of orthogonal (horizontal) trajectory} of $\varpi$ if 
\begin{equation}\label{eq:definition_orthogonal_trajectory}
\im \int_{\widetilde\gamma} \sqrt{-\varpi} \equiv \const
\end{equation}
along any subarc ${\widetilde\gamma}\subset \gamma$. Maximal arcs of trajectories are simply called trajectories. A trajectory is called {\it critical} if it has well-defined endpoints on its two directions, and at least one of these endpoints is a zero or a simple pole of $\varpi$. To stress when a trajectory is not critical, we call it {\it regular}. {\it Ditto} for orthogonal trajectories and critical orthogonal trajectories. The orthogonal trajectories of a quadratic differential $\varpi$ are the same as the trajectories of $-\varpi$, so the qualitative behavior of trajectories and orthogonal trajectories is essentially the same.

The union of critical trajectories of $\varpi$ is called its {\it critical graph}, and the union of critical orthogonal trajectories goes by the name of {\it critical orthogonal graph}. The structure of both graphs of a quadratic differential can be quite rich, both at local and global level, and has found important applications in several areas of mathematics; 
we refer the reader to the book by Jenkins \cite[Theorem~3.5]{jenkins_book}, or also to our previous work \cite[Theorem~B1]{martinez_silva_critical_measures} which rephrases the former in terms closer to the present work. It should be noted though that describing the structure of a critical (or critical orthogonal) graph in each concrete case can be quite challenging. 

A straightforward consequence of the assumption \eqref{solution_cauchy_transform} involving the real measure $\lambda $ is the following proposition:
\begin{lem} \label{lemma:3.2}
	Any interval in $\R^{(1)}\setminus (\supp\lambda )^{(1)}$ is a finite union of arcs of critical trajectories and orthogonal trajectories of $\varpi$.
\end{lem}
\begin{proof}
	Fix a small interval $[a^{(1)},b^{(1)}]\subset \R^{(1)}\setminus (\supp\lambda )^{(1)}$ that does not contain zeros of $\varpi$. By \eqref{symm1}, 
$\xi_1$ is real-valued on $[a^{(1)},b^{(1)}]$ and $\xi_2, \xi_3$ are complex conjugate of each other. Thus, we have two possibilities:
	\begin{enumerate}[a)]
		\item  $\xi_2$ and $\xi_3$ are real on $[a,b]$. In this case $(\xi_2(x)-\xi_3(x))dx$ is real-valued along $[a,b]$, hence $[a^{(1)},b^{(1)}]$ is an arc of orthogonal trajectory.
		\item   $\xi_2$ and $\xi_3$ are non-real, and by \eqref{symm1}, they are complex conjugate of each other. In this situation,  $(\xi_2(x)-\xi_3(x))dx$ is purely imaginary along $[a,b]$, so $[a^{(1)},b^{(1)}]$ is an arc of trajectory.
	\end{enumerate}
	This concludes the proof of the Lemma.
\end{proof}
\begin{remark}
	Later on (see Corollary \ref{cor:_gaps_orthogonal_trajectories}) we will see that in fact $\R^{(1)}\setminus (\supp\lambda )^{(1)}$ is made only of  orthogonal trajectories of $\varpi$, i.e.~it is a subset of the critical orthogonal graph of $\varpi$; thus, only the possibility a) takes place.
\end{remark}

In our work we are interested in the detailed structure of the critical and critical orthogonal graphs of $\varpi$ for two reasons. First, as we will see from Section~\ref{sec:6}, the support of the third component $\mu_3^*$ of the vector critical measure $\vec \mu_*$ (see Theorem~\ref{thm_existence_critical_measure}) is subset of the critical graph of $\varpi$. Second, the construction of the sets $\Gamma_*$ and $\sigma_*$ underlying Theorems \ref{thm_existence_critical_measure} and \ref{thm_s_property_constrained_problem} uses both the critical and the critical orthogonal graphs of $\varpi$. Thus, our goal for the rest of this section is to describe the relevant critical trajectories. In order to do so, we will invoke some basic facts or ``principles'' that we list next. These principles are very similar in spirit as the ones used in \cite[Section~4.5.1]{martinez_silva_critical_measures}, and as we will briefly point out in a moment, they follow easily from basic properties of trajectories.
\begin{itemize}
\item[\bf P1.] If $D\subset \mathcal R$ is a simply connected domain without poles of $\varpi$, then any trajectory of $\varpi$ in $D$ either ends at a zero of $\varpi$ in $D$ or intersects the boundary $\partial D$. 
\item[\bf P2.] A regular trajectory and a regular orthogonal trajectory intersect at most once. If they are critical, the same holds true in any simply connected domain without poles of $\varpi$, and this intersection has to happen at one of their endpoints.
\item[\bf P3.] If $D\subset \mathcal R$ is a simply connected domain, not necessarily bounded, not containing poles of $\varpi$, whose boundary $\partial D$  is a finite union of (arcs of) trajectories, then $\partial D$ contains at least one pole of $\varpi$, and at this point arcs of $\partial D$ meet at a non-zero inner angle. Same conclusion is true if $ \partial D$ is a union of orthogonal trajectories.
\item[\bf P4.] Orthogonal trajectories of $\varpi$ are asymptotically horizontal at $\infty^{(1)}$. That is, if $\gamma$ is an orthogonal trajectory that extends to $\infty^{(1)}$, then 
$$
\lim_{\substack{z\to \infty \\ z\in \gamma\cap \C_+^{(1)}}} \arg z = 0 \; (\mbox{mod }2\pi), \quad \lim_{\substack{z\to \infty \\ z\in \gamma\cap \C_-^{(1)}}} \arg z = \pi \; (\mbox{mod }2\pi).
$$

\item[\bf P5.] Trajectories of $\varpi$ are asymptotically vertical at $\infty^{(1)}$. That is, if $\tau$ is a trajectory that extends to $\infty^{(1)}$, then 
$$
\lim_{\substack{z\to \infty \\ z\in \tau\cap \HH_\pm^{(1)}}} \arg z =\pm\frac{\pi}{2} \; (\mbox{mod }2\pi). 
$$
\end{itemize}

Principle \textbf{P1} above is the same as \cite[Lemma~8.4]{Pommerenkebook}. Principle \textbf{P2} follows from \cite[Theorem~5.5]{strebel_book} and \cite[Theorem~8.3]{Pommerenkebook}. Principle \textbf{P3} is a consequence of Teichmuller's Lemma \cite[Theorem~14.1]{strebel_book}, and Principles \textbf{P4} and \textbf{P5} follow from the local structure of trajectories near poles of order $4$ \cite[Theorem~7.4]{strebel_book}.

\subsection{Critical Orthogonal Graph of \texorpdfstring{$\varpi$}{varpi}} \label{sec:critical_graph}

\

A priori it looks like for a general polynomial potential $V$ and for a parameter $\alpha \in (0,1)$ the structure of the critical orthogonal graph of $\varpi$ can be arbitrarily complicated. Our goal in this section is to show that the natural assumptions we imposed on the spectral curve in Definition~\ref{definition_spectral_curve} have a crucial consequence:  the quadratic differential $\varpi$ has at most one zero, and of order at most 1, in $\HH^{(1)}_+$ (see Proposition~\ref{prop_unique_zero} below). The existence or not of such a zero characterizes the two possible scenarios or regimes of the asymptotic eigenvalue distribution, determined by the existence or not of the non-trivial third component $\mu_3^*$ of the critical measure, and consequently also determines the occurrence of the Pearcey type behavior \eqref{pearcey_like_behavior}.

The main goal of this section is to describe the structure of orthogonal trajectories of $\varpi$ on $\mathcal R_1$. 
Due to the real symmetry, it is enough to focus our attention on the upper half-plane in $\mathcal R_1$. 

Thus, consider the collection $\mathcal Y $ of zeros of $\varpi$ in $\HH_+^{(1)}$ and also the collection $\mathcal X $ of zeros of $\varpi$ on $\R^{(1)}$. These are disjoint and finite collections of points. Furthermore, by construction, branch points of $\mathcal R_1$ necessarily belong to $\mathcal X$, see the discussion after \eqref{definition_varpi}.

Given a set $\gamma$ on $\mathcal R_1$, we denote by $\overline \gamma$ its closure; in the case of a smooth arc with finite length, $\overline \gamma$ is this arc together with its endpoints. Let $\Xi$ be the union of maximal arcs of orthogonal trajectories $\gamma$ of $\varpi$ on $\HH^{(1)}_+$ such that $\overline \gamma \cap \mathcal Y\neq \emptyset$. In other words,  any  critical orthogonal trajectory $\gamma\subset \Xi$ emanates from a certain zero $y^{(1)}\in \mathcal Y$ along $\HH^{(1)}_+$, and (applying Principle {\bf P1} to $\HH^{(1)}_+$) exactly one of the following possibilities holds:
\begin{enumerate}[(1)]
\item $\gamma$ extends to $\infty^{(1)}$.
\item $\gamma$ ends up at another zero $\tilde y^{(1)}\in \mathcal Y$.
\item $\gamma$ ends up at a zero $x^{(1)}\in \mathcal X$.
\item $\gamma$ connects $y^{(1)}$ to a point $d^{(1)}\in \R^{(1)}$, which is not a zero of $\varpi$.
\end{enumerate}

Consider also the set $\widetilde{\mathcal X} = \mathcal X \setminus \overline \Xi \subset \mathcal X$ of points in $\mathcal X$ that do not belong to the closure of any contour in $\Xi$, and let $\widetilde{\Xi}$ be the union of maximal arcs of orthogonal trajectories $\gamma$ of $\varpi$ on $\HH^{(1)}_+$ such that $\overline \gamma \cap \widetilde{\mathcal X}\neq \emptyset$. 

From any $x^{(1)}\in \widetilde{\mathcal X}$ emanate at least three arcs of critical orthogonal trajectories, so at least one of them, say $\tau$, enters $\HH^{(1)}_+$. By the definition of $\widetilde{\mathcal X}$,  $\tau$ does not intersect zeros in $\HH_+^{(1)}$, so applying the Principle {\bf P1} again we know that exactly one of the following is true:
\begin{enumerate}[(1)]
\setcounter{enumi}{4}
\item $\tau$ connects $x^{(1)}$ to another point $d^{(1)}\in \R^{(1)}$, possibly with $d^{(1)}\in \mathcal X$ (we will see that this possibility contradicts Lemma \ref{lem_fundamental_orthogonal_paths_1}, and thus ultimately does not hold), or 
\item $\tau$ extends to $\infty^{(1)}$.
\end{enumerate}

Finally, define 
$$
\mathcal J^{(1)} :=  \Xi \cup  \widetilde \Xi .
$$
\begin{lem}\label{lem_graph_properties}
The set $\mathcal J^{(1)}$ is a graph with finitely many vertices (that contain $\mathcal X \cup \mathcal Y$) and edges and with the following properties.
\begin{enumerate}[(i)]
\item All its edges are (arcs of) critical orthogonal trajectories in $\HH^{(1)}_+$.
\item Every point in $\mathcal Y$ is a vertex of $\mathcal J^{(1)}$, and its degree is equal to its multiplicity as a zero of $\varpi$ plus two.
\item There are edges emanating from every vertex of $\mathcal J^{(1)}$.
\item Every vertex $y^{(1)}\in \mathcal Y$ of $\mathcal J^{(1)}$ is connected to $\R^{(1)}\cup \{\infty^{(1)}\}$.
\item The connected components of $\HH^{(1)}_+\setminus \mathcal J^{(1)}$ are simply connected. If, in addition, the boundary of a connected component does not intersect $\R^{(1)}$, then it is also unbounded.
\item Two different sides of an edge belong to the boundaries of two distinct connected components of $\HH^{(1)}_+\setminus \mathcal J^{(1)}$.
\end{enumerate}
\end{lem}
\begin{proof}
The properties (i)--(iii) follow immediately from the construction of $\mathcal J^{(1)}$. 

If (iv) were not true, then $\mathcal J^{(1)}$ would contain a bounded closed cycle $\gamma\subset \HH_+^{(1)}$, which is a finite union of critical orthogonal trajectories and with simply connected interior $D$. Because $\varpi$ has no poles in $\HH_+^{(1)}$, the simply connected domain $D$ contains no poles on its closure, thus contradicting the Principle {\bf P3}.

Let us prove (vi) before (v). We first look at simple paths (that is, without self-intersections) contained in $\mathcal J^{(1)}$. There is a partial order on simple paths in $\mathcal J^{(1)}$, induced by inclusion. Therefore we can talk about maximal simple paths, that is, simple paths that are not contained in any other simple path with more edges. The endpoints of maximal simple paths are degree one vertices of $\mathcal J^{(1)}$, and therefore have to belong to $\R^{(1)}\cup \{\infty^{(1)}\}$. This is so because any vertex in $\HH^{(1)}_+$ has to belong to $\mathcal Y$, and as such, is a zero of $\varpi$. The claim then follows once we recall that from any zero emanate at least three orthogonal trajectories, i.e.~edges of $\mathcal J^{(1)}$.

Now, suppose that $D$ is a connected component of $\HH^{(1)}_+\setminus \mathcal J^{(1)}$ whose boundary contains an edge $\tau$ whose both sides are part of $\partial D$, and consider a maximal simple path $\gamma$ containing $\tau$. By the previous observation, the endpoints of $\gamma$ cannot belong to $\HH_+^{(1)}$. This means that both endpoints of the simple path $\gamma$ belong to $\R^{(1)}\cup \{\infty^{(1)}\}$. Thus, $\gamma$ has to split $\HH_+^{(1)}$ into two disjoint components, say $D_1$ and $D_2$, and any edge in $\gamma$, in particular $\tau$, separates $D_1$ and $D_2$. Because $D$ is connected we must have either $D\cap D_1=\emptyset$ or $D\cap D_2=\emptyset$, and as such only one side of $\tau$ is in the interior of $D$, a contradiction.

Finally, we now prove (v). If a connected component $D$ of $\HH^{(1)}_+\setminus \mathcal J^{(1)}$ were not simply connected, then its boundary $\partial D$ would have a component $\gamma$ other than the outer boundary of $D$. There are two possibilities for $\gamma$: either it is a tree, or it contains a closed path $\widetilde \gamma$. The first possibility cannot hold, as in this case it would have an edge that contradicts (vi). As for the second possibility, the interior of $\widetilde\gamma$ would then be a simply connected domain $\widetilde D$ without poles on its closure. As before, the domain $\widetilde D$ contradicts the Principle {\bf P3}. This last reasoning also applies to show that $D$ has to be unbounded when its boundary does not intersect $\R^{(1)}$.
%
%
%
%

\end{proof}

The proof of Lemma~\ref{lem_graph_properties} only made use of general principles \textbf{P1}--\textbf{P5} stated above. The next properties do rely on the crucial assumption  \eqref{solution_cauchy_transform} from the Definition~\ref{definition_spectral_curve}. 
\begin{lem}\label{lem_fundamental_orthogonal_paths_1}
If $\gamma\subset \HH_+^{(1)}$ is a union of arcs of critical orthogonal trajectories then at most one endpoint of $\gamma$ belongs to $\R^{(1)}$.
\end{lem}
\begin{proof}
We start observing that the claim of the lemma is equivalent to assuring that there are no paths $\tau$ on the graph $ \mathcal J^{(1)}$ that intersect $\R^{(1)}$ more than once.  Assuming the contrary, we must conclude that there exists a connected component $D$ of $\HH_+^{(1)}\setminus \mathcal J^{(1)}$ with boundary of the form $\partial D=\gamma\cup [c^{(1)},d^{(1)}]$, $c<d$, where
\begin{enumerate}[(i)]
\item $\gamma$ is a path in $\mathcal J^{(1)}$ without self-intersections connecting $c^{(1)}$ and $d^{(1)}$, obtained as the union of finitely many orthogonal trajectories,
\item $(c^{(1)},d^{(1)}) \cap \mathcal X=\emptyset$, i.e.~there are no zeros of $\varpi$ on $(c^{(1)},d^{(1)})$.
\end{enumerate}
First property is self-evident, while (ii) follows from Lemma~\ref{lem_graph_properties}. An immediate consequence of this is that no endpoint $q$ of $\supp\lambda $ can lie on $(c,d)$: by construction, $q^{(1)}$ is a branch point of $\mathcal R_1$ and consequently a zero of $\varpi$. In other words, either $(c,d)\cap \supp\lambda  =\emptyset$ or  $[c,d] \subset \supp\lambda$.

We will now show that such a domain $D$ cannot exist.

If $(c,d)\cap \supp\lambda  =\emptyset$, then the interval $[c,d]$ is not a (part of a) branch cut for $\xi_1$. Consider the closed path $\gamma \cup \gamma^*$ on $\mathcal R$ (where $\gamma^*$ is the complex conjugate  of $\gamma$): it is the boundary of a bounded and simply connected domain in $\mathcal R_1$, and this contradicts the Principle {\bf P3}. 

Assume that
\begin{equation}\label{lem_orthogonal_paths_inclusion}
[c,d] \subset \supp\lambda .
\end{equation}
The open set $D$ is simply connected and does not contain vertices of $\mathcal J^{(1)}$. consequently it does not contain zeros of $\varpi$ (see Lemma \ref{lem_graph_properties}) and as such also no branch points of $\xi_2$ and $\xi_3$. Hence,
$$
\oint_{\partial D} \sqrt{-\varpi}=\oint_{\partial D}(\xi_2(s)-\xi_3(s)) =0.
$$
We can then use (i) above to get
$$
\int_{\gamma} \sqrt{-\varpi} \in \R.
$$
Thus,
\begin{equation}\label{lem_orthogonal_paths_eq_1}
\int_{c}^{d}(\xi_{2+}(x)-\xi_{3+}(x))dx \in \R.
\end{equation}
However, by \eqref{prep_cycles} and by our assumption about $[c,d]$, $\Im (\xi_{2+}(x)-\xi_{3+}(x))$ preserves sign on $[c,d]$, in contradiction with \eqref{lem_orthogonal_paths_eq_1}. 
\end{proof}

\begin{remark}\label{rem_regular_orth_traj}
Lemma \ref{lem_fundamental_orthogonal_paths_1} deals only with arcs obtained as union of critical orthogonal trajectories. However, more is true: there cannot be an arc $\gamma$ of orthogonal trajectory on $\HH^{(1)}_+$ intersecting $\R^{(1)}$ more than once, regardless whether $\gamma$ is critical or regular. The proof of this fact follows along exactly the same lines as in the proof above: one considers $D=[c^{(1)},d^{(1)}]\cup \tau$ a domain contained in the region bounded by $\gamma\cap \R^{(1)}$, constructed as in (i)--(ii) above, and the rest of the argument remains the same.
\end{remark}

\begin{lem}\label{lem_arc_s_contour}
Given $y^{(1)}\in \mathcal Y$, there is exactly one path $\tau\subset \mathcal J^{(1)}$ that connects $y^{(1)}$ with $  \R^{(1)}$.
\end{lem}
\begin{proof}
The uniqueness follows from Lemma \ref{lem_fundamental_orthogonal_paths_1}.

Let us prove existence by contradiction, assuming that no such path exists. There are at least three distinct arcs of orthogonal trajectories emanating from $y^{(1)}$, that form part of three different paths, which by our assumption and by Lemma~\ref{lem_graph_properties} (iv), connect $y^{(1)}$ to $\infty^{(1)}$. From Principle {\bf P5} we know that there are only two allowed directions for these paths to go to $\infty^{(1)}$, and consequently two of them, say $\gamma_1$ and $\gamma_2$, have to extend to $\infty^{(1)}$ along the same angle. Hence, $\gamma_1\cup \gamma_2$ is the boundary of a domain with inner angle $0$ at $\infty^{(1)}$ and no poles on its interior, which contradicts Principle {\bf P3}, concluding the proof.
\end{proof}

One of the crucial conclusions of our analysis is the following proposition.
\begin{prop}\label{prop_unique_zero}
The quadratic differential $\varpi$ has at most one zero in $\HH^{(1)}_+$, and if it exists, it is  simple.
\end{prop}
In our previous notation, this means that either $\mathcal Y=\emptyset$, or it consists of a single point, which is necessarily a simple zero of $\varpi$. 
\begin{proof}
Let $y^{(1)}\in \mathcal Y$. We first show by contradiction that zeros with higher multiplicity are not allowed. Suppose that $y^{(1)}$ is a zero of multiplicity at least $2$. Consequently there exist four arcs of orthogonal trajectories emanating from $y^{(1)}$, and from Lemma \ref{lem_graph_properties} (iv) there correspond four paths of orthogonal trajectories $\gamma_1$, $\gamma_2$, $\gamma_3$ and $\gamma_4$ in $\mathcal J^{(1)}$, which are all pairwise distinct.

Any two distinct paths $\gamma_j$ and $\gamma_k$ intersect at $y^{(1)}$ and at no other point on $\HH_+^{(1)}$, otherwise $\HH_+^{(1)}\setminus \mathcal J^{(1)}$ would have bounded connected components, contradicting Lemma \ref{lem_graph_properties} (v). Also, at most one of them intersects $\R^{(1)}$ (see Lemma \ref{lem_fundamental_orthogonal_paths_1}). Recalling again Lemma \ref{lem_graph_properties} (iv), we conclude that there are at least three among these paths that go to $\infty^{(1)}$. Because of Principle {\bf P4}, two of such arcs, say $\gamma_j$ and $\gamma_k$, have to extend to $\infty^{(1)}$ in the same asymptotic direction. But then $\gamma_k\cup \gamma_j$ is the boundary of a simply connected domain with zero interior angle at $\infty^{(1)}$ and no poles in its interior, contradicting Principle {\bf P3}.

The analysis above, combined with Lemma~\ref{lem_arc_s_contour}, also gives the behavior of orthogonal trajectories emanating from a simple zero: there are three arcs of orthogonal trajectories, two of them extend to $\infty^{(1)}$ horizontally (and with opposite angles) and the remaining one necessarily connects the zero to the real axis.

The case of two simple distinct zeros $y^{(1)}, \widetilde y^{(1)}\in \HH^{(1)}_+$, connected by a path $\gamma $, is topologically equivalent to a double zero, already discarded.

To conclude the proof, it remains to analyze the case of  two simple distinct zeros $y^{(1)}, \widetilde y^{(1)}\in \HH^{(1)}_+$, $y^{(1)}\neq \widetilde y^{(1)}$, not connected by a path in $\mathcal J^{(1)}$.

The behavior of paths of orthogonal trajectories we just observed tells us that one of the three arcs of orthogonal trajectories emanating from $y^{(1)}$ is part of a path connecting $y^{(1)}$ to $\R^{(1)}$, and the other two arcs extend to $\infty^{(1)}$ in the opposite asymptotic directions.   Hence, the complement of these three paths in $ \HH^{(1)}_+$ is the union of three disjoint domains. By assumption, $\widetilde y^{(1)}$ cannot belong to the boundary of these domains, so it belongs to one of these domains. But since the topology of the three arcs of orthogonal trajectories emanating from $y^{(1)}$ and $\widetilde y^{(1)}$  is the same, a simple analysis shows that these two sets of paths of orthogonal trajectories necessarily intersect, which contradicts our assumption that $y^{(1)}$ and $\widetilde y^{(1)} $ are not connected by a path in $\mathcal J^{(1)}$.

The proof is finally complete.
\end{proof}

\begin{prop}\label{prop_orthog_traj_gaps}
	If $q^{(1)}$ is a regular point of $\varpi$ on $\R^{(1)}\setminus (\supp\lambda )^{(1)}$, then the arc of orthogonal trajectory that passes through $q^{(1)}$ is a real interval. In consequence, each interval in $\left( \R \setminus \supp\lambda \right)^{(1)}$ that does not contain singular points of $\varpi$, is an arc of orthogonal trajectory of $\varpi$.
\end{prop}
\begin{proof}
Suppose that the orthogonal arc through $q^{(1)}$, say $\sigma$, is not an interval. Symmetry tells us that the orthogonal trajectory $\sigma$ is symmetric with respect to complex conjugation. Furthermore, as a consequence of Remark~\ref{rem_regular_orth_traj}, it stays within $\mathcal R_1$ and either extends to $\infty^{(1)}$ in its two ends, and in such a case necessarily horizontally and with the same angle in both ends, or it ends at the critical point $y_*^{(1)}$. Without loss of generality, we can assume the latter does not happen. Indeed, in a neighborhood of $q^{(1)}$ on the real line, every point is regular. If $\sigma$ happens to contain $y_*^{(1)}$ then the same would not happen for the orthogonal arc emanating from any other point in this neighborhood, so we could take an arbitrarily small perturbation of $q^{(1)}$ and proceed forward with the orthogonal trajectory emanating from the perturbed point. 
	
Moving forward with the proof, pick a point $p$, $\re p\ge 0$, such that $ p^{(1)}\in \sigma$. Then also $\overline p^{(1)}\in \sigma$, and we can define
	$$
	p \mapsto  \int_{p^{(1)}}^{\overline p^{(1)}} (\xi_2(s)-\xi_3(s))ds 
	$$
	where we integrate along $\sigma$. By the definition of orthogonal trajectory,  $(\xi_2(s)-\xi_3(s))ds$ is real along  $\sigma$. Furthermore, by our previous discussion,  $\sigma$ does not contain critical points
	of $\varpi$, so $(\xi_2(s)-\xi_3(s))ds$ also has constant sign along  $\sigma$. Thus, the just defined function is real-valued, vanishing at  $p=q\in \R$, and strictly monotone along $\sigma \cap \HH_+^{(1)}$. 
	In particular, there exists  $\delta>0$ and $M>0$ such that
	\begin{equation}\label{eq:integration_orth_traj_lower_bound}
	\left| \re  \int_{p^{(1)}}^{\overline p^{(1)}} \sqrt{-\varpi} \right| = \left| \re  \int_{p}^{\overline p} (\xi_2(s)-\xi_3(s))ds \right|\geq \delta>0,\quad \mbox{whenever } |p|\geq M.
	\end{equation}
	On the other hand, now that we know that $\varpi$ has exactly one zero on $\HH_+^{(1)}$ (and its conjugate on the lower half plane, these two particular points are the only possible branch points of $\xi_2$ and $\xi_3$ outside the real axis, so we can deform integrals and write
	$$
	\int_{p}^{\overline p} (\xi_2(s)-\xi_3(s))ds=\sum_\gamma \oint_{\widetilde \gamma}  (\xi_2(s)-\xi_3(s))ds +  \int_\gamma (\xi_2(s)-\xi_3(s))ds
	$$
	where $\gamma$ is any other real-symmetric contour intersecting $\R$ only once, at a point which is not a branch point of $\xi_2$ nor $\xi_3$, and the $\widetilde \gamma$'s are closed loops encircling pairs of branch points of $\xi_2$ and $\xi_3$ on $\R$, and possibly also $\infty$. We stress that this way we can make sure none of the $\widetilde\gamma$'s encircle the non-real branch points of $\xi_2$ and $\xi_3$.
	 
	 Now, by \eqref{prep_cycles} (and the asymptotics \eqref{asymptotics_xi} in case we have to encircle $\infty$) all integrals over the $\widetilde \gamma$'s are purely imaginary, so that 	
for sufficiently large $M$ (such that $\supp \lambda \subset \{|z|<M \}$), the contour $\gamma$ above can be replaced by the arc of a circle $|s|=|p|$ in the right half plane that joins $p$ and $\overline p$. Since the asymptotic expansion  of $\xi_2$ and $\xi_3$ in \eqref{asymptotics_xi} is uniform in a neighborhood of $\infty$, we can integrate it term-by-term and conclude that as $|p|\to +\infty$, 
	$$
	\re  \int_{p^{(1)}}^{\overline p^{(1)}}  (\xi_2(s)-\xi_3(s))ds=\re\left( 2a(\overline p-p)+(2\alpha -1) \log (p/  \overline p) \right)+\Boh(|p|^{-1})=\Boh(|p|^{-1})
	$$
	which contradicts  \eqref{eq:integration_orth_traj_lower_bound} for $|p|$ sufficiently large.

	Finally, the last statement follows from the fact that every regular point $q^{(1)}\in \R^{(1)}\setminus (\supp\lambda )^{(1)}$ of $\varpi$  belongs to an arc of an orthogonal trajectory.
\end{proof}

This result shows that in fact only one of the two possibilities mentioned in Lemma~\ref{lemma:3.2} holds. 

\begin{cor}\label{cor:_gaps_orthogonal_trajectories}
	Any gap $(b_k^{(1)}, a_{k+1}^{(1)})\subset \R^{(1)}\setminus (\supp\lambda )^{(1)}$ is a union of orthogonal trajectories and only contains zeros of $\varpi$ of even multiplicity.
\end{cor}
\begin{proof}
	The fact that the whole interval $(b_k^{(1)}, a_{k+1}^{(1)})$ is a union of orthogonal trajectories is contained in Proposition~\ref{prop_orthog_traj_gaps}. Any zero $p^{(1)}$ of $\varpi$ in $(b_k^{(1)}, a_{k+1}^{(1)})$ has to have even multiplicity, because the angle between the two orthogonal trajectories $(p^{(1)}, (p+\varepsilon)^{(1)})$ and $((p-\varepsilon)^{(1)}, p^{(1)})$ is $\pi$, which has to be an integer multiple of $2\pi/(\mbox{multiplicity of } p^{(1)}+2)$.
\end{proof}

\subsection{Saturated and unsaturated regimes} \label{sec:saturated}

\

As will be concluded later, the existence or not of the simple zero claimed by Proposition \ref{prop_unique_zero} determines whether the measure $\mu_3^{*}$ in Theorem~\ref{thm_existence_critical_measure} is nontrivial or vanishes identically. In terms of the constrained equilibrium problem in Theorem~\ref{thm_s_property_constrained_problem}, it will determine whether the upper constraint for $\nu_2^*$ is active or not. These facts and the traditional terminology explained in Section \ref{sec:totallySym}  motivate the next definition.

\begin{definition}\label{definition_saturated_unsaturated}
If $\varpi$ does not vanish in $\HH^{(1)}_+$, we say that we are in the {\it unsaturated regime}. Otherwise, when $\varpi$ has a simple zero in $\HH^{(1)}_+$ that we denote  by $y^{(1)}_*$, we say that we are in the {\it saturated regime}.
\end{definition}

In this section, we collect some extra information on the zeros and trajectories of $\varpi$ in $\mathcal R_1$. Later on, these results will help in the construction of a more appropriate and explicit branch cut structure for the remaining sheets $\mathcal R_2$ and $\mathcal R_3$ of $\mathcal R$, also providing a solid ground to the proof of our main results.

The general theory of quadratic differentials (see e.g.~\cite{jenkins_book},  \cite[Chapter~8]{Pommerenkebook} or \cite{strebel_book}) tells us that the complement of the critical (vertical or orthogonal) graph of a quadratic differential $\varpi$ on a compact Riemann surface $\mathcal R$ is a finite union of canonical domains. These are the so-called  half-plane, strip, ring and circle domains, as well as density domains, filled with recurrent trajectories. We refrain from giving a systematic explanation of this theory and instead refer to the literature above or to our previous work \cite[Appendix~B]{martinez_silva_critical_measures}. Here we only describe the half-plane domains that we will need next, and whose existence is assured by this general theory.

Associated to the pole $\infty^{(1)}$ of order $4$, there are two half-plane domains bounded by orthogonal trajectories. In virtue of the symmetry of the problem, the projection of these half-plane domains to $\overline \C$ are related by complex conjugation, so we denote them by $\mathcal H$ and $\mathcal H^*$. They enjoy the following properties:
\begin{enumerate}
	\item[$\mathcal H 1)$] The boundary $\partial \mathcal H$ is a finite union of critical orthogonal trajectories, contains the pole $\infty^{(1)}$ and at least one zero of $\varpi$, but no other poles. Furthermore, $\partial \mathcal H$ is a loop on $\mathcal R$ that starts and ends at $\infty^{(1)}$, creating an inner angle $\pi$ there, and it is asymptotically horizontal in its both ends, see Principle \textbf{P.4} in Section~\ref{sec:trajandorthtraj}.
	\item[$\mathcal H 2)$] For any fixed point $p\in \partial \mathcal H$ and a choice of sign $\epsilon\in \{1,-1\}$, the canonical map
	$$
	z\mapsto \epsilon\int_p^z \sqrt{-\varpi}
	$$
	is a conformal map from $\mathcal H$ to $\C_{+}$.
	\item[$\mathcal H 3)$] The domain $\mathcal H$ does not contain critical orthogonal trajectories of $\varpi$.
	\item[$\mathcal H 4)$] Given any point $p\in \mathcal H$  (necessarily regular), the orthogonal trajectory that passes through $p$ is an analytic arc that extends to $\infty^{(1)}$ along its two ends, with angles $0$ and $\pi$, and it is entirely contained in $\mathcal H$.
\end{enumerate}

By symmetry, only one among $\mathcal H$ and $\mathcal H^*$ contains points $p\in \HH_+^{(1)}$ with $\pi(p)$ having positive and large imaginary part; we fix this one to be $\mathcal H$.

For each regime (saturated or not) our goal is:
\begin{enumerate}
	\item[a)] to describe the domain $\mathcal H$;
	\item[b)] to construct an appropriate contour  $\Gamma_*\in \mathcal T$,  see Theorem~\ref{thm_existence_critical_measure}.
\end{enumerate}

\subsubsection{Saturated regime}

\

We start by analyzing the case of the saturated regime, i.e.~of the existence of the simple zero $y^{(1)}_*$ of $\varpi$ in $\HH^{(1)}_+$, see Definition~\ref{definition_saturated_unsaturated}.

\begin{prop}\label{prop:traj_orth_saturated}
In the saturated regime, the trajectories and orthogonal trajectories emanating from $y_*^{(1)}$ are as follows:
\begin{enumerate}[(a)]
\item Exactly one of the orthogonal trajectories from $y_*^{(1)}$ intersects the real axis; we call the intersection point $c_*$. The remaining two trajectories extend to $\infty^{(1)}$ in opposite asymptotically horizontal directions.

If $c_*\in \R\setminus \supp \lambda$ then $c_*$ is a zero of $\varpi$ of even multiplicity.

\item Exactly one of the trajectories from $y_*^{(1)}$ extends to $\infty^{(1)}$ with the asymptotic angle $\pi/2$. The remaining two trajectories intersect the real axis; we call the leftmost and the rightmost intersection points $a_L$ and $a_R$, respectively.
\end{enumerate}
We have
\begin{equation} \label{alessb}
a_L<c_*<a_R
\end{equation}
and 
\begin{equation} \label{abEmpty}
(a_L,c_*) \cap \supp\lambda \neq \emptyset, \quad (c_*,a_R) \cap \supp\lambda \neq \emptyset.
\end{equation}
\end{prop}
The structure of trajectories  emanating from $y_*^{(1)}$ in $\HH_+^{(1)}$ are depicted on Figure~\ref{figure_traj_orth_traj_c}.  
\begin{proof}

The description of the orthogonal trajectories emanating from $y_*^{(1)}$ follows exactly as in the proof of Proposition~\ref{prop_unique_zero}, we label these orthogonal trajectories by $\tau_0$, $\tau_1$ and $\tau_2$ as displayed in Figure~\ref{figure_traj_orth_traj_c}.

\begin{figure}[t]
\begin{tikzpicture}[scale=0.8]
 \draw [line width=0.8mm,black] (5.5,0)--(-5.5,0) node [pos=0.98,above] {$\R^{(1)}$};
 %
 %
 \fill (1,3) circle[radius=4pt] node [right,shift={(1pt,5pt)}] (cstar) { $\, y_*^{(1)}$};
 %
 %
 \draw [thick] (5,4.5) to[out=180,in=50,edge node={node [pos=0.5,above] {$\tau_1$}}] (1,3); 
 \draw [thick,dashed] (1,3) to[out=50-180,in=85,edge node={node [pos=0.5,left] {$\gamma_1$}}] (-1,0);
 \draw [thick,dashed] (-1,6) to[out=260,in=110,edge node={node [pos=0.5,above] {$\gamma_0$}}] (1,3); 
 \draw [thick] (1,3) to[out=110-180,in=95,edge node={node [pos=0.5,right] {$\tau_0$}}] (2,0); 
 \draw [thick] (-5,4) to[out=0,in=170,edge node={node [pos=0.5,above] {$\tau_2$}}] (1,3); 
 \draw [thick,dashed] (1,3) to[out=170-180,in=105,edge node={node [pos=0.5,above] {$\gamma_2$}}] (4,0);  
 \node at (-0.45,0.41) {$a_L^{(1)}$};
  \node at (2.4,0.36) {$c_*^{(1)}$};
  \node at (4.45,0.41) {$a_R^{(1)}$};
\end{tikzpicture}
\caption{Saturated regime: the structure of trajectories (dashed lines, labeled as $\gamma_0$, $\gamma_1$, $\gamma_2$) and orthogonal trajectories (solid lines, labeled as $\tau_0$, $\tau_1$, $\tau_2$) emanating from $y_*^{(1)}$ in $\HH_+^{(1)}$. The trajectories on the lower half plane can be obtained by complex conjugation.}
\label{figure_traj_orth_traj_c}
\end{figure}
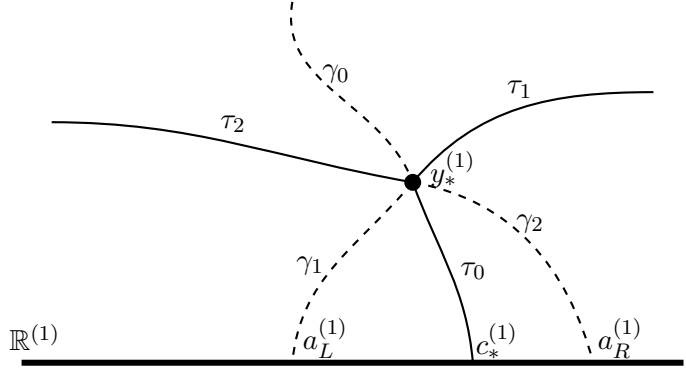

Between any two consecutive orthogonal trajectories $\tau_j$ and $\tau_{j+1}$ emanates exactly one trajectory from $y_*^{(1)}$, say $\gamma_{j+2}$ (with the convention that $\gamma_{j+3}=\gamma_j$), see Figure~\ref{figure_traj_orth_traj_c}. Observe also that $\gamma_j$ is the analytic continuation of $\tau_j$ through $y_*^{(1)}$. An application of Principle {\bf P2} formulated in Section~\ref{sec:trajandorthtraj} tells us that $\gamma_j$ has to remain in the domain bounded by $\tau_j$ and $\tau_{j+1}$. Having in mind that trajectories cannot extend to $\infty^{(1)}$ horizontally, we then apply the Principle {\bf P1} to conclude that $\gamma_1$ and $\gamma_2$ connect $y_*^{(1)}$ to the real line, whereas $\gamma_0$ extends to $\infty^{(1)}$ with angle $\pi/2$.

If $c_*\in \R\setminus \supp \lambda$ then from $c_*$ emanate at least 4 orthogonal trajectories ($\tau_0$, $\tau_0^*$, as well as the real line,  see Corollary \ref{cor:_gaps_orthogonal_trajectories}), and hence $c_*$ must be a zero of $\varpi$ of even multiplicity.

We have \eqref{alessb} by construction. Furthermore,  $\gamma_1\cup [a_L,c_*]\cup \tau_0$ encloses a bounded domain without branch points of $\xi_2$ and $\xi_3$ in its interior; if $(a_L,c_*) \cap \supp\lambda = \emptyset$ then   $  [a_L,c_*]\cup \tau_0$ are orthogonal trajectories, while $\gamma_1$ is a trajectory of $\varpi$, so that 
$$
\Im \int_{\gamma_1} \sqrt{-\varpi} = \frac{1}{i} \int_{\gamma_1} \sqrt{-\varpi} \neq 0,
$$
which leads us into a contradiction. The other assertion in \eqref{abEmpty}  is proved in the same manner. 
\end{proof}

\begin{cor}\label{corol:trajectories_real_line_sat}
Suppose we are in the saturated regime. From any zero $p^{(1)}\in \R^{(1)}$ of $\varpi$ emanates at most one orthogonal trajectory on $\HH_+^{(1)}$. In particular, any such a zero $p^{(1)}$ satisfies:
\begin{itemize}
	\item  $p$ lies in the convex hull  of the support of $\lambda$, and
	\item if  $p^{(1)} \in \R^{(1)}\setminus (\supp \lambda )^{(1)}$, then it is a zero of $\varpi$  of order exactly $2$.
\end{itemize}
In particular, if $c_* \notin \supp \lambda$ then $c_*^{(1)}$  is a zero of $\varpi$  of order exactly $2$.
\end{cor}
\begin{proof}
	 If $p = c_*$, defined in Proposition~\ref{prop:traj_orth_saturated}, then such an orthogonal trajectory is exactly one, that we denoted $\tau_0$. Indeed, any other orthogonal trajectory from $c_*^{(1)}$ into  $\HH_+^{(1)}$ must diverge to $\infty$, either in the direction of $\tau_1$ or $\tau_2$, see Figure~\ref{fig:orthogonal_half_plane}, which contradicts Principle \textbf{P3}.
	
	Suppose that $p^{(1)}\neq c_*^{(1)}$ is a zero of $\varpi$ on $\R^{(1)} $ from which emanate at least two orthogonal trajectories on $\HH_+^{(1)}$, say $\gamma_1$ and $\gamma_2$. 	
	Then   $\gamma_1$ and $\gamma_2$  both have to stay within the closure of the same connected component of $\HH_+^{(1)}\setminus (\tau_0\cup \overline{\mathcal H})$. However, we know also that   these two trajectories have to extend to $\infty^{(1)}$ along different angles (see Lemma~\ref{lem_fundamental_orthogonal_paths_1}), which leads us into a contradiction   (see Figure~\ref{fig:orthogonal_half_plane}).	
	
	If, in addition, $p^{(1)}\in \R^{(1)}\setminus (\supp \lambda )^{(1)}$, then  the angle between each orthogonal trajectory emanating from $p^{(1)}$ has to be $2\pi/(\mbox{order of }p^{(1)})$, and the order must be even as assured in Corollary~\ref{cor:_gaps_orthogonal_trajectories}. If this order is at least four then at least two such orthogonal trajectories have to emerge on $\HH_+^{(1)}$, contradicting the first part. This shows that zeros on $\R^{(1)}\supp\lambda^{(1)}$ have to have order $2$.
	
	If $p$ is not in the convex hull of $\supp\lambda$, say $p\in (-\infty,a_1)$, then the orthogonal trajectory emerging from $p^{(1)}$ vertically has to extend to $\infty^{(1)}$ along the angle $\pi$, and the domain on $\HH^{(1)}$ bounded by this trajectory and the orthogonal $(-\infty^{(1)},a_1^{(1)})$ violates Principle {\bf P3}. This shows the first bullet of the Corollary, concluding the proof.
\end{proof}

In the saturated regime, the description of $\mathcal H$ is an immediate consequence of Proposition~\ref{prop:traj_orth_saturated} and it is displayed in Figure~\ref{fig:orthogonal_half_plane}. 

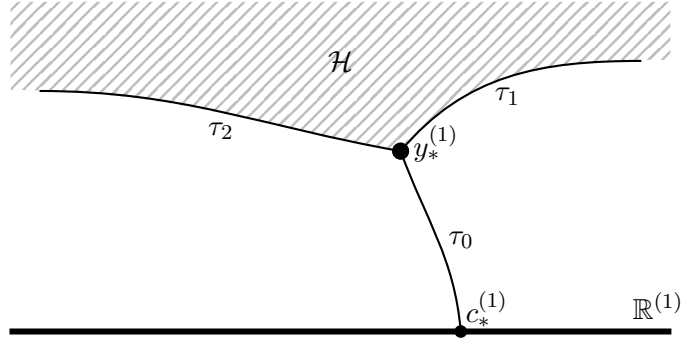
\begin{figure}[t]
	\begin{tikzpicture}[scale=0.8]
	\draw [line width=0.8mm,black] (-5.5,0)--(5.5,0) node [pos=0.98,above] {$\R^{(1)}$};
	%
	%
	\fill (1,3) circle[radius=4pt] node [right,shift={(1pt,3pt)}] (cstar) {$y_*^{(1)}$};
	%
	%
	\draw [thick,white,pattern=north east lines wide,pattern color=lightgray] (5.5,4.5) to [out=180,in=50] (1,3)
	to [out=170,in=0] (-5.5,4)
	to [out=90,in=90-180] (-5.5,5.5)
	to [out=0,in=180] (5.5,5.5)
	to [out=-90,in=90] (5.5,4.5);
	%
	\fill (2,0) circle[radius=3pt] node [above right,shift={(-2pt,-2pt)}] (intstar) {$c_*^{(1)}$};
	%
	%
	\draw [thick] (5,4.5) to[out=180,in=50,edge node={node [pos=0.5,below] {$\tau_1$}}] (1,3); 
	\draw [thick] (1,3) to[out=110-180,in=95,edge node={node [pos=0.5,right] {$\tau_0$}}] (2,0); 
	\draw [thick] (-5,4) to[out=0,in=170,edge node={node [pos=0.5,below] {$\tau_2$}}] (1,3);  
	\fill (1,3) circle[radius=4pt];
	\node at (0,4.5) {$\mathcal H$};
	\end{tikzpicture}
	\caption{Saturated regime: the orthogonal trajectories emanating from $y_*$ in $\HH_+^{(1)}$ and the associated orthogonal half-plane domain $\mathcal H$ (shaded area). The point $c_*$ is the intersection of $\tau_0$ with $\R^{(1)}$. We dropped the superscript $^{(1)}$ for ease of notation.}
	\label{fig:orthogonal_half_plane}
\end{figure}

Now we describe the contour $\Gamma_*$ used in the definition of the vector critical measures, see Section~\ref{sec:vector_crit_measure}:
\begin{prop}\label{prop:construction_s_curves_saturated}
Suppose we are in the saturated regime. 
Define the contour $\Gamma_*$ as
$$
\Gamma_* := \Delta_3\cup \tau_*, \qquad \tau_*:=\pi\left(\tau_1\cup \tau_1^*\right), \quad \Delta_3:= \pi\left((\gamma_1 \cup  \gamma_1^*)\cap \mathcal R_1\right).
$$ 
Then $\Gamma_*$ is piece-wise analytic, belongs to the class $\mathcal T$ (see \eqref{asymptotic_condition_class_T}) and satisfies
\begin{align*}
& (\xi_{2}(z)-\xi_3(z))dz\in \R \quad \mbox{along } \tau_*,\\
& (\xi_{2}(z)-\xi_3(z))dz\in i\R \quad \mbox{along } \Delta_3.
\end{align*}
Furthermore, $\Gamma_*$ intersects $\R$ at exactly one point that we denote by $x_*$.

Finally,  on $\Gamma_*$ the difference $\xi_2-\xi_3$ vanishes only at the points $y_*$ and $\overline{y_*}$, where   $(\xi_{2}-\xi_3)^2$ has a simple zero, as well as possibly at $x_*$.
\end{prop}

The contour $\Gamma_*$ is displayed in Figure~\ref{fig:gamma_sigma_vs_traj}, and the proof of Proposition~\ref{prop:construction_s_curves_saturated} is an immediate consequence of Proposition~\ref{prop:traj_orth_saturated} and Corollary~\ref{corol:trajectories_real_line_sat}. Notice that $x_*$ is the point that was denoted initially by $a_L$ in Proposition~\ref{prop:traj_orth_saturated}, see also Figure~\ref{figure_traj_orth_traj_c}.

	\begin{figure}[t]
		\begin{tikzpicture}[scale=0.8]
		\draw [line width=0.8mm,lightgray] (-5.5,0)--(5.5,0) node [pos=0.98,above,black] {$\R^{(1)}$};
		%
		%
		\fill (1,3) circle[radius=4pt] node [right,shift={(1pt,5pt)}] (cstar) {};
		%
		\draw [line width=0.6mm,dashed,blue] (5,4.5) to[out=180,in=50,edge node={node [pos=0.5,above] {$\tau_*$}}] (1,3); 
		\draw [line width=0.6mm,dashed,red] (1,3) to[out=50-180,in=80,edge node={node [pos=0.8,left] {$\Delta_3$}}] (-1,0);
		\draw [thick,lightgray] (-0.5,6) to[out=270,in=110,edge node={node [pos=0.5,above] { }}] (1,3); 
			\draw [thick,lightgray] (1,3) to[out=110-180,in=260-180,edge node={node [pos=0.7,right] {}}] (1,-3); 
		\draw [thick,lightgray] (-5,4) to[out=0,in=170] (1,3); 
		\draw [thick,lightgray] (1,3) to[out=170-180,in=90] (4,0);  
		%
		\fill (1,3) circle[radius=4pt];
		%
		%
		%
		\begin{scope}[yscale=-1]
		%
		%
		\fill (1,3) circle[radius=4pt] node [right,shift={(1pt,5pt)}] (cstar) {};
		%
		%
		\draw [line width=0.6mm,dashed,blue] (5,4.5) to[out=180,in=50,edge node={node [pos=0.5,above] {$\tau_*$}}] (1,3); 
		\draw [line width=0.6mm,dashed,red] (1,3) to[out=50-180,in=80,edge node={node [pos=0.5,left] {}}] (-1,0);
		\draw [thick,lightgray] (-0.5,6) to[out=270,in=110,edge node={node [pos=0.5,above] { }}] (1,3); 
		%
		\draw [thick,lightgray] (-5,4) to[out=0,in=170] (1,3); 
		\draw [thick,lightgray] (1,3) to[out=170-180,in=90] (4,0);  
		\fill (1,3) circle[radius=4pt]; 
		%
		%
		\end{scope}
		\fill (-1,0) circle[radius=4pt] node [above right,shift={(0pt,-2pt)}] {$x_*$};
	 	\fill (1.45,0) circle[radius=4pt] node [below right,shift={(-1pt,0pt)}] {$c_*$};
		
		\end{tikzpicture}
		\caption{Saturated regime:  contour $\Gamma_*=\tau_*\cup\Delta_3$ (dashed), as described in Proposition~\ref{prop:construction_s_curves_saturated}. For convenience, we kept the remaining trajectories from $y_*^{(1)}$ and $\overline  y_*^{(1)}$ in light gray, compare with Figure~\ref{figure_traj_orth_traj_c}.}
		\label{fig:gamma_sigma_vs_traj}
	\end{figure}
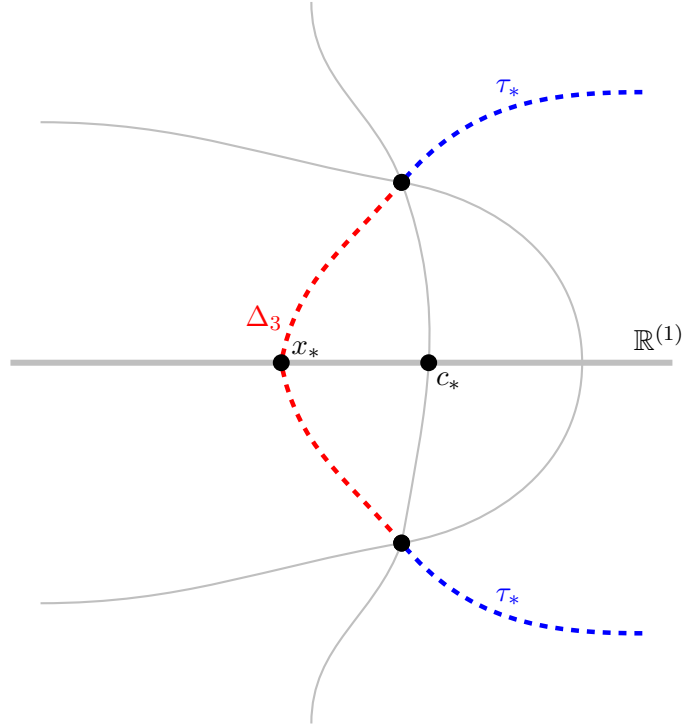

\subsubsection{Unsaturated regime}

\

The construction of $\Gamma_*$ in Proposition~\ref{prop:construction_s_curves_saturated} is only meaningful in the saturated regime. Now we address the unsaturated regime, for which  we again start by  describing the structure of the half-plane domain $\mathcal H$. Recall that now all the zeros of $\varpi$ on the first sheet $\mathcal R_1$ are actually on the real line $\R^{(1)}$.
 
\begin{prop}\label{prop:unsaturated_half_plane}
Suppose we are in the unsaturated regime. Then exactly one of the following has to occur.

\begin{enumerate}[(I)]
\item $\partial \mathcal H\cap \R^{(1)}$ consists of exactly one point $c_*^{(1)}$, which is not 
an endpoint of $\supp\lambda^{(1)}$ and it is necessarily a zero of $\varpi$. In this case $\partial \mathcal H$ is a union of two critical orthogonal trajectories $\tau_1$ and $\tau_2$ intersecting at $c_*^{(1)}$ and extending to $\infty^{(1)}$ with angles $0$ and $\pi$, respectively. And finally, $c_*\in (a_1,b_l)=\mathrm{conv} \left( \supp \lambda \right)$ and both $\tau_1$ and $\tau_2$ are analytic arcs contained in $\HH_+^{(1)}$ (except, obviously,  their endpoint $c_*^{(1)}$), see Figure~\ref{fig:orthogonal_half_plane_unsaturated_a}. In particular, in this case $\mathcal H\subset \HH^{(1)}_+$.

\item $\partial \mathcal H$ contains  exactly two distinct zeros of $\varpi$, say $c_1^{(1)}, c_2^{(1)}\in \R^{(1)}$, $c_2<c_1 $, see Figure~\ref{fig:orthogonal_half_plane_unsaturated_2}. Then $\partial \mathcal H$ is a union of two critical orthogonal trajectories $\tau_1$ and $\tau_2$ on $\HH_+^{(1)}$, with behavior as in (I), and the real interval $(c_2^{(1)},c_{1}^{(1)})$. The points $c_1^{(1)}$ and $c_{2}^{(1)}$ are the unique zeros of $\varpi$ on $\partial \mathcal H$, and 
$$
(c_2,c_1)\subset \text{\rm conv} (\supp\lambda)\setminus \supp\lambda.
$$ 
In particular, in this case $\mathcal H\subset \HH^{(1)}_+$ as well.

\item $\partial \mathcal H$ contains an unbounded interval. In this case, there are two further possibilities:
\begin{enumerate}[(a)]
\item This interval is of the form $\tau_2=(-\infty^{(1)},c_*^{(1)})$ with $c_*\leq a_1$, and $\partial \mathcal H=\tau_1\cup\tau_2$, with $\tau_1$ is as in the case (I), see Figure~\ref{fig:orthogonal_half_plane_unsaturated_bc_B}.
\item This interval is of the form $\tau_1=(c_*^{(1)},\infty^{(1)})$ with $c_*\geq b_l$, $\partial \mathcal H=\tau_1\cup\tau_2$, with $\tau_2$ is as in the case (I), see Figure~\ref{fig:orthogonal_half_plane_unsaturated_bc_C}.
\end{enumerate}
\end{enumerate}
\end{prop}

We will see below in Proposition~\ref{prop:exclusion_caseIII} that, in fact, the situation (III) would correspond to a reducible spectral curve, happening only when $a=0$. So under our conditions, the situation (III) is irrelevant.
 
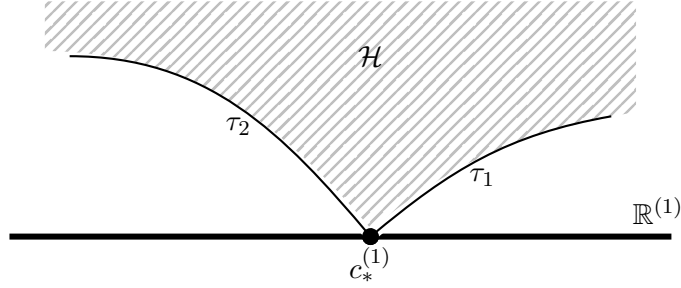
\begin{figure}
\begin{tikzpicture}[scale=0.8]
 \draw [line width=0.8mm,black] (-5.5,0)--(5.5,0) node [pos=0.98,above] {$\R^{(1)}$};
 %
 %
 \fill (0.5,0) circle[radius=4pt] node [below,shift={(0pt,0pt)}] (cstar) {$c_*^{(1)}$};
 %
 %
   \draw [thick,white,line width=4pt,pattern=north east lines wide,pattern color=lightgray] (-5,3) 
   						to [out=0,in=130] (0.5,0)
 						to [out=40,in=190] (5,2)
 						to [out=90,in=-90] (5,4)
 						to [out=180,in=0] (-5,4) 
 						to [out=-90,in=90] (-5,3);
 
 %
 %
  \draw [thick] (-4.5,3) to [out=0,in=130,edge node={node [pos=0.5,below] {$\tau_2$}}] (0.5,0)
 						to [out=40,in=190,edge node={node [pos=0.5,below] {$\tau_1$}}] (4.5,2);
 \fill (0.5,0) circle[radius=4pt];
 \node at (0.5,3) {$\mathcal H$};
 
\end{tikzpicture}
\caption{Unsaturated regime: the orthogonal trajectories emanating from $x_*^{(1)}$ (solid lines) associated orthogonal half-plane domain $\mathcal H$ (shaded area) in the situation (I) of Proposition~\ref{prop:unsaturated_half_plane}.}
\label{fig:orthogonal_half_plane_unsaturated_a}
\end{figure}

\begin{figure}[t]
\begin{tikzpicture}[scale=0.8]
 %
 %
 %
   \draw [thick,white,line width=4pt,pattern=north east lines wide,pattern color=lightgray] (-6,3) 
   						to [out=0,in=130] (-2,0)
   						to [out=0,in=180] (2,0)
 						to [out=40,in=190] (6,2)
 						to [out=90,in=-90] (6,4)
 						to [out=180,in=0] (-6,4) 
 						to [out=-90,in=90] (-6,3);
 %
 %
  \path [draw=black,thick] (-6,3) to [out=0,in=130,edge node={node [pos=0.4,below] {$\tau_2$}}] (-2,0);
  \path [draw=black,thick] (2,0) to [out=40,in=190,edge node={node [pos=0.6,below] {$\tau_1$}}] (6,2);
  \draw [draw=black,thick] (-2,0) to (2,0);
 \node at (-2,0) (c2) {};
 \fill (c2) circle[radius=4pt] node [below,shift={(0pt,-3pt)}] {$c_2$};
 \node at (2,0) (c1) {};
 \fill (c1) circle[radius=4pt] node [below,shift={(0pt,-3pt)}] {$c_1$};
 \node at (0.5,3) {$\mathcal H$};
 \node at (0,0) (c) {};
 \draw [line width=0.8mm,black] (-6,0)--(c2);
 \draw [line width=0.8mm,black] (c1)--(6,0) node [pos=0.95,above] {$\R^{(1)}$};

\end{tikzpicture}
\caption{Unsaturated regime: the orthogonal half-plane domain $\mathcal H$, which is the shaded region with boundary $\tau_1\cup[c_1,c_2]\cup\tau_2$ as described by the situation (II) of Proposition \ref{prop:unsaturated_half_plane}. We emphasize that in this case $[c_2,c_1]$ does not intersect $\supp\lambda$. }
\label{fig:orthogonal_half_plane_unsaturated_2}
\end{figure}
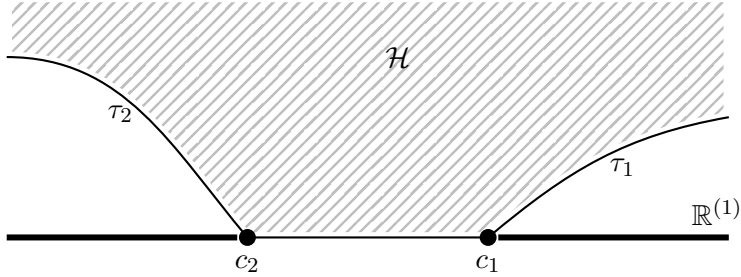

\captionsetup[subfigure]{labelformat=simple} 
\renewcommand\thesubfigure{ }

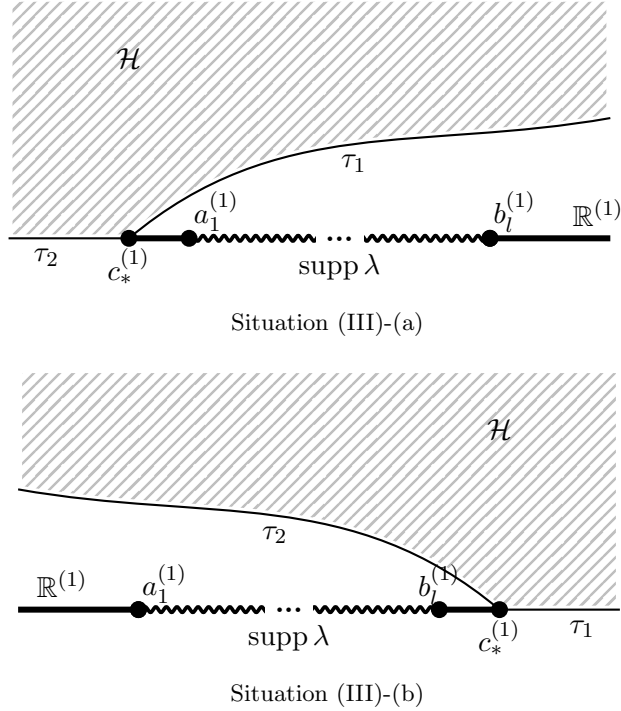
\begin{figure}[t]
\begin{subfigure}{1\linewidth}\centering
\begin{tikzpicture}[scale=0.8]
 %
 %
 %
 %
   \draw [thick,white,line width=3pt,pattern=north east lines wide,pattern color=lightgray] (-5,0) 
   						to [out=0,in=180] (-3,0)
 						to [out=40,in=190] (5,2)
 						to [out=90,in=-90] (5,4)
 						to [out=180,in=0] (-5,4) 
 						to [out=-90,in=90] (-5,0);
 %
 %
 %
  \draw [thick] (-3,0) to [out=40,in=190,edge node={node [pos=0.5,below] {$\tau_1$}}] (5,2);
  \draw [thick] (-5,0) to [out=0,in=180,edge node={node [pos=0.3,below] {$\tau_2$}}] (-3,0);
 \fill (-3,0) circle[radius=4pt] node [below,shift={(0pt,0pt)}] (cstar) {$c_*^{(1)}$};
 %
 %
 \node at (-3,3) {$\mathcal H$};
 \draw [line width=0.8mm,black] (-3,0)--(-2,0) node [above right,shift={(-3pt,-2pt)}] {$a_1^{(1)}$};
 \fill (-2,0) circle[radius=4pt];
 %
 %
 \draw [line width=0.5mm,decorate, decoration={snake, segment length=4pt, amplitude=1pt}] (-2,0)--(3,0) node [rectangle, fill=white,pos=0.5] {\bf ...} node [above right,shift={(-3pt,-2pt)}] {$b_l^{(1)}$};
 \fill (3,0) circle[radius=4pt];
 \node at (0.5,0) [below,shift={(0pt,-2pt)}] {$\supp\lambda $};
 \fill (3,0) circle[radius=4pt];
 \draw [line width=0.8mm,black] (3,0)--(5,0) node [pos=0.9,above] {$\R^{(1)}$};

\end{tikzpicture}
\subcaption{Situation (III)-(a)} \label{fig:orthogonal_half_plane_unsaturated_bc_B}
\end{subfigure}
\par\bigskip
\begin{subfigure}{1\linewidth}\centering
\begin{tikzpicture}[scale=0.8]
 %
 %
 %
 %
 \begin{scope}[xscale=-1]
   \draw [thick,white,line width=3pt,pattern=north east lines wide,pattern color=lightgray] (-5,0) 
   						to [out=0,in=180] (-3,0)
 						to [out=40,in=190] (5,2)
 						to [out=90,in=-90] (5,4)
 						to [out=180,in=0] (-5,4) 
 						to [out=-90,in=90] (-5,0);
 %
 %
 %
  \draw [thick] (-3,0) to [out=40,in=190,edge node={node [pos=0.5,below] {$\tau_2$}}] (5,2);
  \draw [thick] (-5,0) to [out=0,in=180,edge node={node [pos=0.3,below] {$\tau_1$}}] (-3,0);
 \fill (-3,0) circle[radius=4pt] node [below,shift={(0pt,0pt)}] (cstar) {$c_*^{(1)}$};
 %
 %
 \node at (-3,3) {$\mathcal H$};
 \draw [line width=0.8mm,black] (-3,0)--(-2,0) node [above,shift={(0pt,-3pt)}] {$b_l^{(1)}$};
 \fill (-2,0) circle[radius=4pt];
 %
 %
 \draw [line width=0.5mm,decorate, decoration={snake, segment length=4pt, amplitude=1pt}] (-2,0)--(3,0) node [rectangle, fill=white,pos=0.5] {\bf ...} node [above right,shift={(-3pt,-2pt)}] {$a_1^{(1)}$};
 \fill (3,0) circle[radius=4pt];
 \node at (0.5,0) [below,shift={(0pt,-2pt)}] {$\supp\lambda $};
 \fill (3,0) circle[radius=4pt];
 \draw [line width=0.8mm,black] (3,0)--(5,0) node [pos=0.95,above right] {$\R^{(1)}$};
 \end{scope}
\end{tikzpicture}
\subcaption{Situation (III)-(b)}\label{fig:orthogonal_half_plane_unsaturated_bc_C}
\end{subfigure}
\caption{Unsaturated regime: the orthogonal half-plane domain $\mathcal H$ and the orthogonal critical trajectories $\tau_1$ and $\tau_2$, in the situations (III)-(a)   and (III)-(b)  described in Proposition \ref{prop:unsaturated_half_plane}. The wiggled segments represent the bounded interval of the real axis containing $\supp\lambda $.}
\label{fig:orthogonal_half_plane_unsaturated_bc}
\end{figure}

\begin{proof}
Assume first that $\partial \mathcal H$ does not contain unbounded intervals of $\R^{(1)}$. We will show that under this situation either (I) or (II) occurs.

Our initial assumption means that the boundary of $\mathcal H$ near infinity consists of critical orthogonal trajectories $\tau_1$ and $\tau_2$ on $\HH_+^{(1)}$ emerging from $\infty^{(1)}$ at angles $0$ and $\pi$ respectively, and those trajectories are not real intervals.  Because $\partial \mathcal H$ has to contain at least one zero (see property $\mathcal H 1)$ above),  and in the unsaturated regime there are no zeros on $\HH_+^{(1)}$, the orthogonal trajectories $\tau_1$ and $\tau_2$ have to meet the real axis. Denote by $c_1^{(1)}$ and $c_2^{(1)}$ the first points of intersection of $\tau_1$ and $\tau_2$ with $\R^{(1)}$, respectively, and by $\widetilde \tau_1$ and $\widetilde \tau_2$ the subarcs of $\tau_1$ and $\tau_2$ that go from $\infty^{(1)}$ to $c_1^{(1)}$ and $c_2^{(1)}$. 

If $c_1=c_2=:c_*$, then we are immediately in case (I), as long as we also verify that $c_*\in (a_1,b_l)$. To see that the former is indeed true, suppose to the contrary that $c_*$ belongs, say, to $(-\infty,a_1]$, the case when $c_*\in [b_l^{(1)},\infty^{(1)})$ is analogous. Because the interval $(-\infty^{(1)},a_1^{(1)}]$ is also a union of arcs of orthogonal trajectories (see Proposition~\ref{prop_orthog_traj_gaps}), then the domain on $\mathcal H^{(1)}$ bounded by $(-\infty^{(1)},a_1^{(1)}]$ and $\tau_2$ violates the Principle {\bf P3}.

Suppose now that $c_1\neq c_2$, so that $\mathcal H\cap \HH_+^{(1)}$ is as displayed in Figure~\ref{fig:foliation_trajectories_0}. The argument just presented above excludes $c_j\in (-\infty,a_1]$ or $c_j\in [b_l,\infty)$. We will now show that $(c_2,c_1)\subset\R\setminus \supp\lambda$, so that (II) is taking place.

\begin{figure}[t]
\begin{tikzpicture}[scale=0.8]
 %
 %
 %
   \draw [thick,white,line width=4pt,pattern=north east lines wide,pattern color=lightgray] (-6.5,3) 
   						to [out=0,in=130] (-2,0)
   						to [out=0,in=180] (2,0)
 						to [out=40,in=190] (6.5,2)
 						to [out=90,in=-90] (6.5,4)
 						to [out=180,in=0] (-6.5,4) 
 						to [out=-90,in=90] (-6.5,3);
 %
 %
  \path [draw=red,thick,postaction={mid arrow={red,scale=2}}] (-6,3) to [out=0,in=130,edge node={node [pos=0.4,below] {$\widetilde\tau_2$}}] (-2,0);
  \path [draw=red,thick,postaction={mid arrow={red,scale=2}}] (2,0) to [out=40,in=190,edge node={node [pos=0.6,below] {$\widetilde\tau_1$}}] (6,2);
 \fill (-2,0) circle[radius=4pt] node [below,shift={(0pt,-3pt)}] {$c_2$};
 \fill (2,0) circle[radius=4pt] node [below,shift={(0pt,-3pt)}] {$c_1$};
 \node at (0.5,3) {$\mathcal H$};
 \node at (0,0) (c) {\bf ...};
 \draw [line width=0.8mm,black] (-7,0)--(c);
 \draw [line width=0.8mm,black] (c)--(7,0) node [pos=0.95,above] {$\R^{(1)}$};

\end{tikzpicture}
\caption{The structure of $\mathcal H$ on $\HH_+^{(1)}$ used in the proof of Proposition~\ref{prop:unsaturated_half_plane}, in the situation when $c_1\neq c_2$.}
\label{fig:foliation_trajectories_0}
\end{figure}
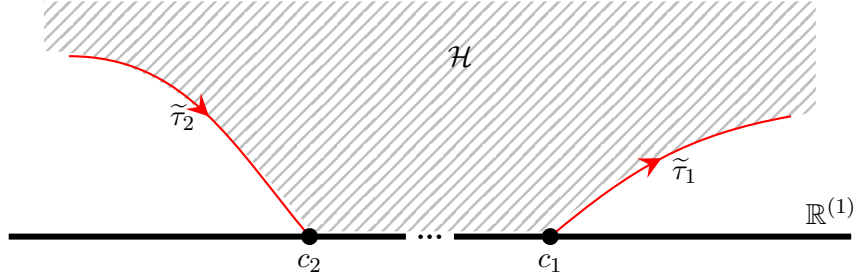

As a consequence of Lemma~\ref{lem_fundamental_orthogonal_paths_1}, we obtain that in this situation $\partial \mathcal H\cap \HH^{(1)}_+=\widetilde \tau_1\cup\widetilde \tau_2$.  Indeed, if there were a point in the intersection $\partial \mathcal H\cap \HH^{(1)}_+ \notin \widetilde \tau_1\cup\widetilde \tau_2$, then there would be an arc $\widetilde \tau$ in $\partial \mathcal H$ that does not intersect $\widetilde \tau_1$ and $\widetilde \tau_2$. Having in mind that the whole boundary $\partial \mathcal H$ is a connected and piecewise analytic curve, this would mean that $\widetilde \tau$ would have to extend to the real line on its two ends, contradicting Lemma~\ref{lem_fundamental_orthogonal_paths_1}.

The orthogonal trajectories on a half-plane domain are analytic loops starting and ending at $\infty^{(1)}$ that provide a foliation of this domain. If we assume that
$$
 (c_2, c_1) \cap \supp \lambda\neq \emptyset,
$$
then there must exist orthogonal trajectories inside $\mathcal H$ that move to another sheet across one of the cuts $[a_j,b_j]$, as well as trajectories on the upper half plane that stay completely on $\HH_+^{(1)}$, check the dashed lines in Figure~\ref{fig:foliation_trajectories}. 

\begin{figure}[t]
\begin{tikzpicture}[scale=0.8]
%
%
\draw [line width=0.5mm,decorate, decoration={snake, segment length=4pt, amplitude=1.5pt}] (-3,0)--(2.5,0);
\fill (-3,0) circle[radius=4pt] node [below,shift={(0pt,-2pt)}] {$a_j$} node [left] (a) {\bf ...};
\fill (2.5,0) circle[radius=4pt] node [below,shift={(0pt,-2pt)}] {$b_j$} node [right] (b) {\bf ...};
%
%
\path [draw=red,thick,postaction={mid arrow={red,scale=2}}] (-7,1) to [out=-20,in=160,edge node={node [pos=0.2,below] {$\widetilde\tau_2$}}] (-5,0);
\foreach \i in {1,...,6}{
\draw [thick,dashed,lightgray] ([xshift=0pt,yshift=4*\i pt]-7,1) to [out=-20,in=160-6*\i] ([xshift=5*\i pt,yshift=0pt]-5,0);
}
\path [draw=red,thick,postaction={mid arrow={red,scale=2}}] (5,0) to [out=75,in=180,edge node={node [pos=0.5,right,shift={(2pt,-1pt)}] {$\widetilde\tau_1$}}] (7,2);
\foreach \i in {1,...,6}{
\draw [thick,dashed,lightgray] ([xshift=-5*\i pt,yshift=0 pt]5,0) to [out=75,in=180] ([xshift=0 pt,yshift=4*\i pt]7,2);
}
\foreach \i in {1,...,10}{
\draw [thick,dashed,lightgray] ([xshift=0 pt,yshift=4*\i pt]-7,5) to [out=0,in=180] ([xshift=0 pt,yshift=4*\i pt]1,4)
									 to [out=0,in=180] ([xshift=0 pt,yshift=4*\i pt]7,5);
}
%
%
\path [draw=blue,line width=0.6mm,postaction={mid arrow={blue,scale=1.5}}] (-7,3) to [out=-5,in=120,edge node={node [pos=0.2,above] {$\gamma_2$}}] (-1,0);
\draw [thick,blue,dashed,line width=0.6mm] (-1,0) to [out=120-180,in=250,dashed,edge node={node [pos=0.5,below] {$\widetilde\gamma$}}] (1.5,0);
\path [draw=blue,line width=0.6mm,postaction={mid arrow={blue,scale=1.5}}] (1.5,0) to [out=250-180,in=180,edge node={node [pos=0.8,above] {$\gamma_1$}}] (7,4);
%
%
\filldraw (-1,0) circle[radius=4pt] node [below,shift={(0pt,-2pt)}] {$p$};
\filldraw (1.5,0) circle[radius=4pt] node [below,shift={(0pt,-2pt)}] {$q$};

\draw [line width=0.8mm,black] (-7,0)--(a);
\draw [line width=0.8mm,black] (b)--(7,0) node [pos=0.9,below] {$\R^{(1)}$};

\end{tikzpicture}
\caption{Part of the orthogonal half plane $\mathcal H$ on $\HH_+^{(1)}$, along with its boundary curves $\widetilde\tau_1$ and $\widetilde\tau_2$ (in solid gray), and part of the foliation by its orthogonal trajectories (dashed gray). On $\R$, the wiggled line represents the interval $[a_j^{(1)},b_{j}^{(1)}]$ on which the orthogonal trajectory $\gamma=\gamma_1\cup\widetilde \gamma\cup\gamma_3$ intersects $\R^{(1)}$. The subarc $\widetilde \gamma$ (thick dashed) is the part of $\gamma$ which is on another sheet (when labeling the points, we have dropped the upper index $^{(1)}$ for convenience).}
\label{fig:foliation_trajectories}
\end{figure}
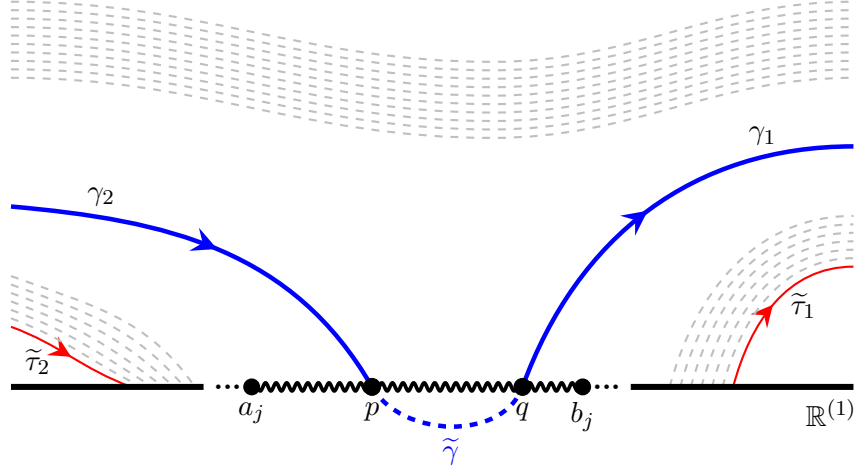

Let us orient the orthogonal trajectories on $\mathcal H$ from $-\infty$ to $+\infty$ and  travel along such curves with this orientation. The sketched behavior of trajectories implies in particular that we can choose some orthogonal trajectory $\gamma$ that leaves $\HH_+^{(1)}$ only for an arbitrarily small and short period, in the sense that the length of its subarcs that are not in $\HH_+^{(1)}$ can be assumed to be arbitrarily small.  This means that this $\gamma$  can be written in the form $\gamma=\gamma_2\cup\widetilde\gamma\cup\gamma_1$, where $\widetilde \gamma$ is the piece of $\gamma$ not in $\HH_+^{(1)}$. As it was mentioned, it can be assumed that $\widetilde \gamma$ has arbitrarily small length, and when we walk along $\gamma$ we observe the following properties:
\begin{itemize}
\item Starting from $\infty^{(1)}$ with an angle $\pi$, we meet the real axis at a first point $p^{(1)}$. This determines $\gamma_2$. 

\item We then move to another sheet, so $p^{(1)}\in [a_j^{(1)},b_j^{(1)}]$ for some $j$. Because $\gamma$ can be chosen to make the length of $\widetilde \gamma$ as small as needed, we return to $\HH_+^{(1)}$ at a point $q^{(1)}\in [a_j^{(1)},b_j^{(1)}]$ for the same $j$. The walk from $p^{(1)}$ to $q^{(1)}$ determines $\widetilde \gamma$.

\item From $q^{(1)}$, we move to $\infty^{(1)}$ with angle $0$, determining then $\gamma_1$ and finishing our walk along $\gamma$.
\end{itemize}
Such a trajectory $\gamma$ is displayed in Figure~\ref{fig:foliation_trajectories}. From the definition of an orthogonal trajectory, we know that
$$
\im \int_{\widetilde \gamma}\sqrt{-\varpi} =0. 
$$
Because $\widetilde \gamma$ can be chosen arbitrarily small, we can actually deform the path $\widetilde \gamma$ in the integration above without crossing any critical points of $\varpi$ nor branch cuts of $\mathcal R$ other than $[a_j^{(1)},b_j^{(1)}]$, and then conclude that
$$
\int_{p}^q(\xi_{2+}(x)-\xi_{3+}(x))dx\in \R,
$$
which is in contradiction with \eqref{prep_cycles}. Hence, $(c_2, c_1)\cap \supp \lambda=\emptyset$. Since $\partial \mathcal H$ has to be connected, there is a bounded arc $\widetilde \tau \subset \partial \mathcal H $ connecting $c_1^{(1)}$ and $c_2^{(1)}$ and, as a consequence of Lemma~\ref{lem_fundamental_orthogonal_paths_1}, $\widetilde \tau=[c_2^{(1)},c_1^{(1)}]$, thus concluding that (II) takes place, as explained above. The fact that the interval $[c_2^{(1)},c_1^{(1)}]$ does not contain other zeros is true because, if there were any zero there, then there would necessarily be an orthogonal trajectory emerging from it on the upper half place, contradicting the property $\mathcal H 3)$ above.

In summary, assuming that $\partial \mathcal H$ does not have unbounded real intervals we concluded that either $c_1=c_2(=x_*)$, corresponding to the case (I), or $c_2<c_1$, corresponding to case (II)

Finally, suppose now that $\partial \mathcal H$ does contain an unbounded real interval. Arguments similar to the ones carried out above in the case $c_1\neq c_2$ also exclude the following situations:
\begin{enumerate}[(1)]
\item Both intervals $(-\infty^{(1)},a_1^{(1)})$ and $(b_l^{(1)},\infty^{(1)})$ belong to $\partial \mathcal H$.
\item The interval $(-\infty^{(1)},a_1^{(1)})$ is contained in $\partial \mathcal H$, and the remaining arc from $\partial \mathcal H$ emerges from $a_1^{(1)}$ in another sheet. Same goes for $(b_l^{(1)},\infty^{(1)})$.
\end{enumerate}

To conclude the proof, let us assume that $\partial \mathcal H$ contains a maximal subinterval $(-\infty^{(1)},c_*^{(1)})$ of $(-\infty^{(1)},a_1^{(1)})$. As we are assuming this interval is maximal, if $c_*<a_1$ there should be an arc of orthogonal trajectory $\tau_1\subset \partial \mathcal H$ emerging from $c_*$ on $\HH_+^{(1)}$, whereas if $c_*=a_1$ then such $\tau_1$ also has to exist but now by (2) above. Either way, by Lemma~\ref{lem_fundamental_orthogonal_paths_1} this arc $\tau_1$ cannot return to $\R^{(1)}$ anymore, so it has to extend to $\infty^{(1)}$ (recall that we are under the unsaturated regime, so no zeros on $\HH_+^{(1)}$), and thus $\partial \mathcal H=(-\infty^{(1)},c_*^{(1)}]\cup \tau_1=\tau_2\cup \tau_1$ as claimed by situation (III)--(a). 

The case that $\partial \mathcal H$ contains a subinterval of $(b_l^{(1)},\infty^{(1)})$ implies (III)--(b) follows along exactly the same lines, we skip the details.
\end{proof}

The following is an analogue of Corollary \ref{corol:trajectories_real_line_sat}, now in the unsaturated case:
\begin{cor}\label{corol:trajectories_real_line_uhp}
Suppose we are in the unsaturated regime. If $p^{(1)}\in \R^{(1)}$ is any zero of $\varpi$, different from the point $c_*^{(1)}$ of the case (I) of Proposition~\ref{prop:unsaturated_half_plane}, then from $p^{(1)}$ emanates at most one orthogonal trajectory on $\HH_+^{(1)}$. In particular, any such zero on $\R^{(1)}\setminus (\supp \lambda )^{(1)}$ must be of order exactly $2$.
\end{cor}

\begin{proof}
A simple inspection in the description of $\partial \mathcal H$ shows that other than $c_*$, from no other zero on the boundary emanates two trajectories. Thus, suppose that $p^{(1)}$ is a zero on $\R^{(1)}\setminus \partial\mathcal H$ from which emanates at least two orthogonal trajectories, say $\gamma_1$ and $\gamma_2$. 

Without loss of generality, assume $p>c_1$ (in the case (II)) or $p>c_*$ (in the case (I)), the other cases follow analogously. This means that we are in the situations covered in one of Figures~\ref{fig:orthogonal_half_plane_unsaturated_a} or \ref{fig:orthogonal_half_plane_unsaturated_2}, and as a consequence of Lemma~\ref{lem_fundamental_orthogonal_paths_1} we see that $\gamma_1$ and $\gamma_2$ have to extend to $\infty^{(1)}$ along the same connected component of $\HH_+^{(1)}\setminus \overline{\mathcal H}$ (this is the white region between $\partial \mathcal H$ and $\R^{(1)}$ in Figures~\ref{fig:orthogonal_half_plane_unsaturated_a},  \ref{fig:orthogonal_half_plane_unsaturated_2} or \ref{fig:orthogonal_half_plane_unsaturated_bc_B}). In such a way, the domain delimited by $\gamma_1\cup\gamma_2$ then violates Principle {\bf P3}.

In the case (II), if $p\in [c_2, c_1]$ we get the same contradiction with Principle {\bf P3} by considering now the domain  delimited by, for instance, $\gamma_1 \cup [p, c_1]\cup \tau_1$.

The assertion about the order of $p^{(1)}$ as a zero of $\varpi$ is proved exactly as in Corollary \ref{corol:trajectories_real_line_sat}.
\end{proof}

We are finally ready to construct the contour  $\Gamma_*\in \mathcal T$   for the unsaturated regime.
\begin{prop}\label{prop:construction_trajectories_s_contour_unsaturated}
Suppose we are in the unsaturated regime, in one of the cases (I) or (II) from Proposition~\ref{prop:unsaturated_half_plane}. Then there exists a contour  $\Gamma_*\in \mathcal T$  that intersects the real axis at a single point, that we denote by $c_*$, and for which
\begin{equation} \label{prop317_1}
(\xi_{2}(z)-\xi_3(z))dz\in \R \quad \mbox{along } \Gamma_*\setminus \{c_*\}.
\end{equation}

Moreover, the difference $\xi_2-\xi_3$ does not vanish along $\Gamma_* $  away from the real line. 
\end{prop}

\begin{remark}
As we mentioned before, in Proposition~\ref{prop:exclusion_caseIII} we exclude the occurrence of the case (III) from Proposition~\ref{prop:unsaturated_half_plane}, so the contour $\Gamma_*$ in fact always exists.
\end{remark}

The contour  $\Gamma_*$  in the unsaturated regime is depicted in Figure~\ref{fig:gamma_sigma_vs_traj_unsaturated}, corresponding to  configuration (I) of Proposition \ref{prop:unsaturated_half_plane}, in which case $c_*$ is the point for which we used the same notation there. In the situation (II) its geometry is analogous, and in such a case $c_*=c_1$. At least at the formal level, it is natural to expect that $\Gamma_*$ in the unsaturated regime can be obtained from the corresponding contour in the saturated regime (see Figure~\ref{fig:gamma_sigma_vs_traj}) by continuous deformation of the parameter data $(V,a,\alpha)$.

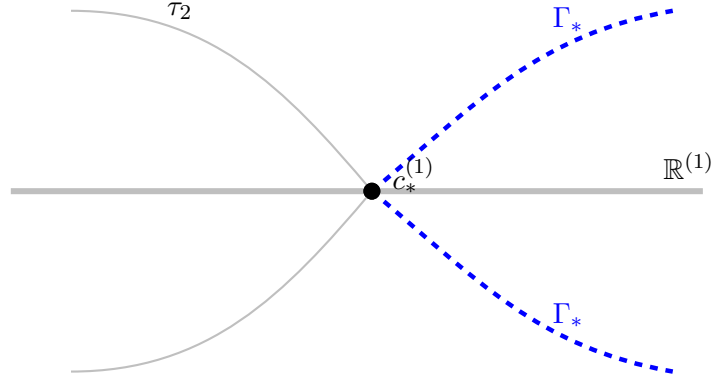
\begin{figure}[t]
	\begin{tikzpicture}[scale=0.8]
	\draw [line width=0.8mm,lightgray] (-5.5,0)--(6,0) node [pos=0.98,above,black] {$\R^{(1)}$};
	%
	%
	%
	%
	%
	%
	\draw [thick,lightgray] (-4.5,3) to [out=0,in=130,edge node={node [pos=0.3,above,black] {$\tau_2$}}] (0.5,0);
	\draw [line width=0.6mm,dashed,blue] (0.5,0) to [out=40,in=190,edge node={node [pos=0.7,above] {$\Gamma_*$}}] (5.5,3);
	%
	%
	\fill (0.5,0) circle[radius=4pt];
	
	%
	\begin{scope}[yscale=-1]
	\draw [thick,lightgray] (-4.5,3) to [out=0,in=130] (0.5,0);
	\draw [line width=0.6mm,dashed,blue] (0.5,0) to [out=40,in=190,edge node={node [pos=0.7,above] {$\Gamma_*$}}] (5.5,3);
	%
	%
	%
	\end{scope}
	
	\fill (0.5,0) circle[radius=4pt] node [right,shift={(4pt,6pt)}] (cstar) {$c_*^{(1)}$}; 
	
	\end{tikzpicture}
	\caption{Unsaturated regime:  contour $\Gamma_*$ (dashed)  as described in Proposition~\ref{prop:construction_trajectories_s_contour_unsaturated} and corresponding to the situation (I)  of Proposition \ref{prop:unsaturated_half_plane}. Arcs of $\Gamma_*$ correspond to the orthogonal trajectory  $\tau_1$ and it reflection with respect to $\R$, see Figure~\ref{fig:orthogonal_half_plane_unsaturated_a}. Orthogonal trajectory  $\tau_2$ is indicated in gray. }
	\label{fig:gamma_sigma_vs_traj_unsaturated}
\end{figure}

\begin{proof}
For the structure of $\partial \mathcal H$ covered by either one of the Figures~\ref{fig:orthogonal_half_plane_unsaturated_a} or \ref{fig:orthogonal_half_plane_unsaturated_2}, we  take
$$
\Gamma_*=\pi\left(\tau_1\cup \tau_1^*\right), 
$$
and from the definition of the orthogonal trajectory $\tau_1$ it follows that \eqref{prop317_1} is satisfied. 

The fact that $\Gamma_*\in \mathcal T$ follows from the geometry of the orthogonal trajectories involved in the construction, whereas  properties of the zeros of $\xi_2-\xi_3$  are consequences of the definition of an orthogonal trajectory, see \eqref{eq:definition_orthogonal_trajectory}, and the properties of the zeros of $\varpi$ on $\partial\mathcal H$, see for instance Proposition~\ref{prop:unsaturated_half_plane}.
\end{proof}

\subsection{Local behavior of the density of \texorpdfstring{$\lambda$}{lambda}}\label{sec:local_behavior}

\

The results in the previous section describe the behavior of certain trajectories of $\varpi$. Although some of the main findings therein, namely Propositions~\ref{prop:construction_s_curves_saturated} and \ref{prop:construction_trajectories_s_contour_unsaturated}, contain, in principle, everything needed for the proof of Theorem~\ref{thm_existence_critical_measure}, there are several almost immediate consequences that not only provide more accurate information about the density $\lambda $ but also give a more concrete geometric picture of the sheets of $\mathcal R$, and as such, facilitate  the proof of the main results. So without further ado, let us analyze the key consequences.

\begin{prop}\label{proposition_interior_support_branching}
If $\frac{d\lambda }{dx}(p)>0$ for some $p\in \R$, then none of the functions $\xi_1$, $\xi_2$ and $\xi_3$ has a branch point at $p$.
\end{prop}
\begin{proof}
If $\lambda'(p)>0$ then by \eqref{Sokhotsky-Plemelj}, $\xi_{1\pm}(p)\notin \R$ and $\xi_{1+}(p)\neq \xi_{1-}(p)$, so that $p$ cannot be a branch point of $\xi_1$. But then by \eqref{lem_orthogonal_paths_eq_0}, one of the values $\xi_{2\pm}(p)$, $\xi_{3\pm}(p)$ is real and the other not, so that again $\xi_2$ and $\xi_3$ cannot be simultaneously branched at $p$.
\end{proof}

\begin{prop}\label{proposition_endpoint_support_branching}
If $p$ is an endpoint of one of the intervals comprising $\supp\lambda $, then $p$ cannot be a branch point of both $\xi_2$ and $\xi_3$.
\end{prop}

\begin{proof}
We carry out the proof assuming that $p\in \{b_1,\hdots,b_l\}$. The case when $p\in\{a_1,\hdots,a_l\}$ is similar.

To get to a contradiction, assume that $\xi_2$ and $\xi_3$ are branched at $p$. Because $p$ is an endpoint of the support of $\lambda $, we get from \eqref{solution_cauchy_transform} that $\xi_1$ is branched at $p$ as well, so $p$ is a common branch point to the three functions $\xi_1$, $\xi_2$ and $\xi_3$. Hence, for some $\varepsilon>0$ the functions $\xi_2$ and $\xi_3$ admit an analytic continuation to $U:=D_\varepsilon(p)\setminus (p-\varepsilon,p]$, and in this neighborhood we can expand 
$$
\xi_j(z)=\widehat \xi(z) + \beta \omega_j (z-p)^{\nu/3}+\Boh((z-p)^{(\nu+1)/3}), \quad z\in U, \quad j=1,2,3,
$$
for some function $\widehat \xi$ analytic near $p$ and some integer $\nu>0$ not a multiple of $3$,  $\beta>0$, and $|w_j|=1$, with the $\omega_j$'s pairwise distinct and satisfying $\omega_1^3=\omega_2^3=\omega_3^3$. The branch of cubic root in the expansion above is the principal one, that is, it is real-valued on $(p,+\infty)$ and with a branch cut on $(-\infty,p)$. 

The function $\xi_1$ is real-valued on $(p,p+\varepsilon)$ (see \eqref{solution_cauchy_transform}), and thus we must have $\omega_1=1$, and consequently $\xi_{2+}(x),\xi_{3+}(x)$ must be nonreal on $(p,p+\varepsilon)$. Because the equation \eqref{spectral_curve} has real coefficients, we must then have $\xi_{2+}(x)=\overline{\xi_{3+}(x)}$, meaning that $\xi_{2+}(x)-\xi_{3+}(x)\in i\R\setminus \{0\}$ for $x\in (p,p+\varepsilon)$. But this is in contradiction to Proposition~\ref{prop_orthog_traj_gaps}.
\end{proof}

\begin{prop}\label{prop:interior_point_unified}
	If $p$ is an interior point of $\supp\lambda $ with $\frac{d\lambda }{dx}(p)=0$, then $p$ is either 
	\begin{enumerate}[a)]
		\item a regular point of all three branches $\xi_1, \xi_2$ and $\xi_3$, or
		\item a branch point  of all three branches $\xi_1, \xi_2$ and $\xi_3$ simultaneously. Then,  for $j=1, 2, 3$,
	\begin{equation}\label{local_behavior_branch_point_pearcey}
	\xi_j(z)=\frac{V'(p)}{3}  +\beta\omega_j(z-p)^{\nu/3}+\Boh(z-p)^{(\nu+1)/3},\quad z\to p, \quad \im z>0,
	\end{equation}
	where the cubic root  has a branch cut along $(-\infty,p)$. Here $\beta>0$, $ \omega_1^3  = \omega_2^3=\omega_3^3$ are three distinct values, and either one of the following possibilities takes place:
	\begin{itemize}
		\item $\nu=1$ and $ \omega_1  = e^{\pi i/3}$, or
		\item $\nu=5$ and $ \omega_1  = e^{2\pi i/3}$.
	\end{itemize}
	Moreover, this situation can occur only in the unsaturated regime, and only at  $p_*^{(1)}=c_*^{(1)}$ described in Proposition~\ref{prop:unsaturated_half_plane}, case (I). 
	 	\end{enumerate}
\end{prop}
\begin{proof}
	 By Lemma~\ref{lem_imaginaryperiods}, because $\frac{d\lambda}{dx}(p)=0$ we obtain $\xi_{1+}(p)=\xi_{1-}(p)$. Also by the same Lemma, for any $x\neq p$ in a small interval $(p-\varepsilon, p+\varepsilon)$ there exists a permutation $j, k$ of $\{2, 3 \}$ (depending on $x$)  such that    
$$
	\xi_{1\pm}(x)=\overline{\xi_{1\mp}(x)}=\xi_{j \mp}(x)=\overline{\xi_{j\pm}(x)}, \quad     \xi_{k\pm }(x)\in \R.
$$

If $\xi_1$ is not branched at $p$ then without loss of generality we can assume that it is not branched on $(p-\varepsilon, p+\varepsilon)$ either. Using continuity this means that in the identities above the indices $j$ and $k$ are independent of $x$. For instance, let $j=2$, so that
$$
\xi_{1\pm}(x)=\overline{\xi_{1\mp}(x)}=\xi_{2 \mp}(x)=\overline{\xi_{2\pm}(x)}, \quad     \xi_{3\pm }(x)\in \R, 
$$
	which yields that all three functions have a trivial monodromy at $p$. In other words, functions
	$$
	\begin{cases}
\xi_1(z), & \text{if } \Im z>0, \\
\xi_2(z), & \text{if } \Im z<0,
	\end{cases} \quad \text{and} \quad \begin{cases}
	\xi_3(z), & \text{if } \Im z>0, \\
	\xi_3(z), & \text{if } \Im z<0,
	\end{cases}
	$$
	have an analytic and single-valued continuation to a small neighborhood of $p$, so none of them could be branched at $p$. The other case, when $j=3$, is analyzed identically.

	This shows that we are in the option a). 
	
	Assume now that all three functions $\xi_1,\xi_2$ and $\xi_3$ are branched at a point $p$ in the interior of $\supp\lambda $. 	
	From the general theory, the Puiseux expansion of the solutions to \eqref{spectral_curve} at a branch point $p$ of multiplicity three is of the form
	\begin{equation}\label{eq_puiseux_triple}
	\xi_j(z)=\frac{V'(p)}{3} +\beta\omega_j(z-p)^{\nu/3}+\Boh(z-p)^{(\nu+1)/3},\quad z\to p, \quad \im z>0,
	\end{equation}
	where the constant term comes from \eqref{solution_cauchy_transform}, $\beta>0$ and $\omega_j$'s satisfy $ \omega_1^3  = \omega_2^3=\omega_3^3$ and $| \omega_1|=1$.  Let us analyze the possibilities for $\nu \in \N$.  For the initial screening, we will use the consequence of  Corollaries~\ref{corol:trajectories_real_line_sat} and \ref{corol:trajectories_real_line_uhp}, according to which from any zero $p^{(1)} \in \R^{(1)}$ of $\varpi$ emanates at most one orthogonal trajectory on the upper half plane $\HH_+^{(1)}$, unless we are in the unsaturated regime and  $p_*^{(1)}=c_*^{(1)}$ (see Proposition~\ref{prop:unsaturated_half_plane}, case (I)), when there are exactly two.

With our assumptions, when $p$ is a triple branch point, the local parameter at $p$ is
	\begin{equation}\label{localparameter}
	z= p+u^3,  \quad dz^2=u^4 du^2,
	\end{equation}
	and
	\begin{equation}\label{eq_local_behavior_qd_triple_branchpoint}
	\varpi = -(\xi_j-\xi_k)^2dz^2=\const \times u^{2\nu+4} du^2.
	\end{equation}
	This means that $\varpi$ has a zero of order $2\nu+4$ at this point,  and there are $2\nu + 6$ orthogonal trajectories emanating from $p$ on $\mathcal R$, with the identical angle (in the local parameter $u$)
	$$
	\frac{2\pi}{2\nu + 6}=\frac{\pi}{\nu + 3}
	$$
	between two consecutive trajectories. By \eqref{localparameter}, this angle  projects on the $z$ plane as  $3\pi /(\nu + 3)$. 
	Hence, if 
	$$
	\frac{3\pi}{\nu + 3} < \frac{\pi}{3}, \qquad \mbox{which is equivalent to} \qquad \nu > 6,
	$$
	then a winding counting implies that we must have at least three critical orthogonal trajectories emanating from $p^{(1)}$ on the upper half plane $\HH_+^{(1)}$, which is impossible.
	
	Since by  assumption $\nu \not \equiv 0 \mod 3$, we are only left with the possibilities $\nu \in \{ 1, 2, 4, 5\}$.
	
We now can discriminate further looking at the sign of $\im\xi_{1+}(x) $ in a neighborhood of $p$. Indeed, by \eqref{Sokhotsky-Plemelj},  
	$$
 \im\xi_{1+}(x)  =	\pi \frac{d\lambda }{dx}(x)     >0 \quad \text{on $(p-\varepsilon, p+\varepsilon)\setminus \{ p\}$,}
	$$
	so that $\xi_{1+}(x) \notin \R$ there, and by \eqref{lem_orthogonal_paths_eq_0}, at least one of the solutions has to be real in each of the intervals $(p-\varepsilon,p)$ and $(p,p+\varepsilon)$. Thus, necessarily either $\omega_2$ or $\omega_3$ are real (or equivalently, either $\omega_2 \in \{-1,1 \}$ or $\omega_3\in \{-1,1 \}$). Furthermore, by \eqref{eq_puiseux_triple},	
		\begin{equation*}
 \im\xi_{1+}(x) =
\begin{cases} 
 \beta \im ( \omega_1 )   |x-p|^{\nu/3} +\Boh(|x-p|^{\frac{\nu+1}{3}})					& \quad \mbox{ on } (p, p+\varepsilon), \\
 \beta \im ( \omega_1 e^{\pi i \nu/3})   |x-p|^{\nu/3} +\Boh (|x-p|^{\frac{\nu+1}{3}} )	&	 \quad \mbox{ on } ( p-\varepsilon, p).  
\end{cases}
	\end{equation*}
Hence, by assumption, both $\im ( \omega_1  )\geq 0 $ and $\im ( \omega_1 e^{\pi i \nu/3}) \geq 0$. Moreover, since  either $\xi_{2+}(x) \in \R$  or $\xi_{3+}(x) \in \R $ in  $(p-\varepsilon,p)$, we also need that either $   \omega_2 e^{\pi i \nu/3} \in \R$ or $  \omega_3 e^{\pi i \nu/3} \in \R$.

	Taking into account that $\omega_1^3 = \omega_2^3 =\omega_3^3$ and they are pairwise distinct, this leaves us only two possibilities: 
	\begin{enumerate}
		\item[a)] if, without loss of generality,   $\omega_3=1$ then $\omega_1= e^{2\pi i/3}= \overline{\omega_2}$.  A priori, $\im ( \omega_1 e^{\pi i \nu/3}) \geq 0$ implies that $\nu  = 1, 4 \text{ or } 5$. But  recalling that either $   \omega_2 e^{\pi i \nu/3} \in \R$ or $  \omega_3 e^{\pi i \nu/3} \in \R$, we conclude that the only possibility is $\nu = 5 $.
			\item[b)] if, without loss of generality,   $\omega_3=-1$ then $\omega_1= e^{\pi i/3}= \overline{\omega_2}$.  A priori, $\im ( \omega_1 e^{\pi i \nu/3}) \geq 0$ implies that $\nu  = 1, 2 \text{ or } 5$. But  recalling that either $   \omega_2 e^{\pi i \nu/3} \in \R$ or $  \omega_3 e^{\pi i \nu/3} \in \R$, we conclude that the only possibility is $\nu = 1 $.
	\end{enumerate}
 
This leaves us only with the two possibilities indicated in the case b) of the Proposition. In order to establish the last statement, we return to the local parameter \eqref{localparameter} and count the actual number of orthogonal trajectory emanating from $p^{(1)}$ on the upper half plane $\HH_+^{(1)}$. 

If, again without loss of generality, $\omega_2= \pm 1$, then $(p, p+\varepsilon)^{(2)}$ is an arc of trajectory of $\varpi$. According to the local structure of trajectories and orthogonal trajectories, there is an orthogonal trajectory emanating from $p^{(2)}$ on $\mathcal R^{(2)}$ in the direction $e^{i\theta_0}$,  $\theta_0= 3 \pi /(2(\nu + 3)) $, and the rest are equally spaced, emanating from $p$ (on the corresponding sheet) with the directions $e^{i\theta_k}$, where 
$$
\theta_k=  \dfrac{3 \pi }{2(\nu + 3)} + k\, \dfrac{3 \pi }{ \nu + 3 }   ,  \quad k=0, 1, \dots, 2\nu+5.
$$
A simple analysis shows that in both cases there are exactly two orthogonal trajectories in $\HH_+^{(1)}$: for $\nu=1$ these correspond to the directions $\theta_5, \theta_6$ (see Figure~\ref{fig:local_trajec_triple_branch}), and for $\nu=5$, to the directions $\theta_{11}, \theta_{12}$.

\begin{center}
	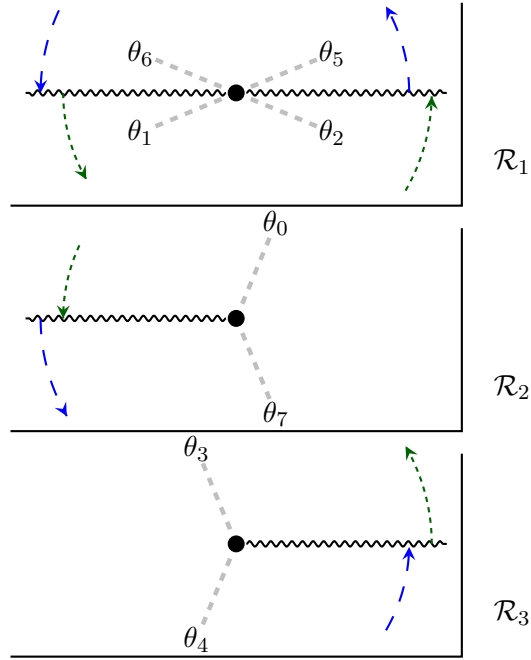
\begin{figure}[t]
		\begin{tikzpicture}
		%
		%
		\path[draw,thick] (-7,3.5) to [edge node={node [pos=1,above right,shift={(8pt,8pt)}] {$\mathcal R_1$}}] (-1,3.5) to (-1,6.2);
		\path[draw,thick] (-7,0.5) to [edge node={node [pos=1,above right,shift={(8pt,8pt)}] {$\mathcal R_2$}}] (-1,0.5) to (-1,3.2);
		\path[draw,thick] (-7,-2.5) to [edge node={node [pos=1,above right,shift={(8pt,8pt)}] {$\mathcal R_3$}}] (-1,-2.5) to (-1,0.2);
		%
		%
		\node at (-4,5) (p1) {};
		\node at (-4,2) (p2) {};
		\node at (-4,-1) (p3) {};
		\draw[thick,decorate, decoration={snake, segment length=4pt, amplitude=1pt}] (p1) -- ++(180:2.8) node (end1left) {};
		\draw[thick,decorate, decoration={snake, segment length=4pt, amplitude=1pt}] (p1) -- ++(0:2.8) node (end1right) {};
		\draw[thick,decorate, decoration={snake, segment length=4pt, amplitude=1pt}] (p2) -- ++(180:2.8) node (end2) {};
		\draw[thick,decorate, decoration={snake, segment length=4pt, amplitude=1pt}] (p3) -- ++(0:2.8) node (end3) {};
		%
		%
		\pgfmathsetmacro{\rd}{1.2}
		\pgfmathsetmacro{\th}{135}
		\pgfmathsetmacro{\tho}{67.5}
		\pgfmathsetmacro{\lw}{0.6}
		\foreach \k [remember=\k as \lastindex] in {1,2,5,6}{
			\draw [thick,dashed,line width=\lw mm,lightgray](-4,5)-- ++({\rd*cos(\th*\k+\tho)},{\rd*sin(\th*\k+\tho)}) node [black,shift={(\th*\k+\tho:5pt)}] {$\theta_{\k}$}; 
		}
		\foreach \k [remember=\k as \lastindex] in {0,7}{
			\draw [thick,dashed,line width=\lw mm,lightgray](-4,2)-- ++({\rd*cos(\th*\k+\tho)},{\rd*sin(\th*\k+\tho)}) node [black,shift={(\th*\k+\tho:5pt)}] {$\theta_{\k}$}; 
		}
		\foreach \k [remember=\k as \lastindex] in {3,4}{
			\draw [thick,dashed,line width=\lw mm,lightgray](-4,-1)-- ++({\rd*cos(\th*\k+\tho)},{\rd*sin(\th*\k+\tho)}) node [black,shift={(\th*\k+\tho:5pt)}] {$\theta_{\k}$}; 
		}
		\filldraw (p1) circle[radius=3pt]; 
		\filldraw (p2) circle[radius=3pt]; 
		\filldraw (p3) circle[radius=3pt]; 
		%
		%
		%
		\begin{scope}[shift={(p1)}]
		\draw [blue,thick,dash pattern={on 6pt off 6pt on 6pt off 6pt},postaction={end arrow={scale=1.2}},domain=155:180] plot ({2.6*cos(\x)}, {2.6*sin(\x)});
		\draw [green,thick,dash pattern={on 2pt off 2pt on 2pt off 2pt},postaction={end arrow={scale=1.2}},domain=180:210] plot ({2.3*cos(\x)}, {2.3*sin(\x)});
		\draw [green,thick,dash pattern={on 2pt off 2pt on 2pt off 2pt},postaction={end arrow={scale=1.2}},domain=-30:-1] plot ({2.6*cos(\x)}, {2.6*sin(\x)});
		\draw [blue,thick,dash pattern={on 6pt off 6pt on 6pt off 6pt},postaction={end arrow={scale=1.2}},domain=0:30] plot ({2.3*cos(\x)}, {2.3*sin(\x)});
		\end{scope}
		\begin{scope}[shift={(p2)}]
		\draw [green,thick,dash pattern={on 2pt off 2pt on 2pt off 2pt},postaction={end arrow={scale=1.2}},domain=155:180] plot ({2.3*cos(\x)}, {2.3*sin(\x)});
		\draw [blue,thick,dash pattern={on 6pt off 6pt on 6pt off 6pt},postaction={end arrow={scale=1.2}},domain=180:210] plot ({2.6*cos(\x)}, {2.6*sin(\x)});
		\end{scope}
		\begin{scope}[shift={(p3)}]
		\draw [green,thick,dash pattern={on 2pt off 2pt on 2pt off 2pt},postaction={end arrow={scale=1.2}},domain=0:30] plot ({2.6*cos(\x)}, {2.6*sin(\x)});
		\draw [blue,thick,dash pattern={on 6pt off 6pt on 6pt off 6pt},postaction={end arrow={scale=1.2}},domain=-30:-1] plot ({2.3*cos(\x)}, {2.3*sin(\x)});
		\end{scope}
		
		\end{tikzpicture}
		\caption{The sheet structure of $\mathcal R$ and behavior of local trajectories near $p=p^{(1)}$ in the case when the expansion \eqref{local_behavior_branch_point_pearcey} is taking place with $\nu=1$. The wiggled lines correspond to the branch cuts connecting the three different sheets at $p$, with the arrow lines (long dashed and short dashed) indicating the orientation of the branch cut connections. The dashed gray lines correspond to the angles $\theta_0, \hdots, \theta_7$ indicating the directions in which  the orthogonal trajectories  emanate from $p$. }
		\label{fig:local_trajec_triple_branch}
	\end{figure}
\end{center}

It remains to apply Corollaries~\ref{corol:trajectories_real_line_sat} and \ref{corol:trajectories_real_line_uhp} and Proposition~\ref{prop:unsaturated_half_plane} to conclude the proof. 

\end{proof}

\begin{remark}
Looking at the possible deformations of our problem, we can describe, rather informally, the dynamical origin of both types of triple branch points stated in Proposition~\ref{prop:interior_point_unified}. The transition from the saturated to unsaturated regime takes place  when both non-real zeros of $\varpi$, $y_*^{(1)}$ and $\overline{y_*}^{(1)}$ (square root type branch points), merge on the real line.  If this happens in the bulk of $\supp \lambda$, the triple branching with $\nu=1$ occurs. Otherwise, the zero formed from the coalescence of $y_*^{(1)}$ and $\overline{y_*}^{(1)}$ on $\R^{(1)}$ is ``trapped'' between the colliding endpoints of $\supp \mu_1$ and $\supp \mu_2$ (also square root type branch points), giving origin to $\nu=5$ in \eqref{local_behavior_branch_point_pearcey}. In random matrix theory language, the latter would correspond to the closure of a gap happening at the place where the support of the third measure $\supp\mu_3^*$ (which will be constructed later) degenerates to a single point.
\end{remark}

We finish this section with a definition that is motivated by the conclusions of Proposition~\ref{prop:construction_s_curves_saturated}:
\begin{definition}\label{def:unsaturated_regular_singular}
We say that we are in the {\it regular unsaturated regime} if we are in the saturated regime but there is no point for which the expansion \eqref{local_behavior_branch_point_pearcey} is valid. Otherwise,  we say that we are in the {\it singular unsaturated regime}.
\end{definition}

The words ``regular'' and ``singular'' in the definition above come from their consequence on the local behavior of the density of $\lambda $ claimed by Theorem~\ref{thm_existence_critical_measure}. As we will see in a moment, the Pearcey-type singular behavior exists precisely when we are in the singular unsaturated regime.

\subsection{Riemann surface for the spectral curve: take II}\label{sec:riemann_surface_II}

\

At the beginning of Section \ref{section_qd} we mentioned the three-sheeted Riemann surface $\mathcal R=\mathcal R_1\cup \mathcal R_2\cup \mathcal R_3$ associated to the spectral curve \eqref{spectral_curve}. In that moment, we chose $\mathcal R_1=\overline{\C}\setminus \supp\lambda $, but the remaining sheets $\mathcal R_2$ and $\mathcal R_3$, apart from minor considerations, were taken rather arbitrarily. 

Now we can use the information from the previous sections to construct $\mathcal R_2$ and $\mathcal R_3$ in a more precise manner.

Recall that the saturated and unsaturated regimes are described in Definition~\ref{definition_saturated_unsaturated}, and the unsaturated regime was divided further into two subregimes in Definition~\ref{def:unsaturated_regular_singular}.
\begin{thm}\label{thm:branchpoints}
The branch points of $\xi_1,\xi_2$ and $\xi_3$ are as follows.
\begin{enumerate}[(a)]
\item For any $p\in \{a_1,b_1,\hdots,a_l,b_l\}$, the function $\xi_1$ and exactly one of the other two solutions $\xi_2$ or $\xi_3$ are branched (``square root'' branching).

\item In the singular unsaturated regime, the three functions $\xi_1,\xi_2$ and $\xi_3$ are branched at $p=c_*$, which is a point in the interior of $\supp\lambda$. Furthermore, this is the unique point where the expansion \eqref{local_behavior_branch_point_pearcey} is valid.

\item In the saturated regime, the functions $\xi_2$ and $\xi_3$ are branched also at $y_*\in \mathbb H_+$ and $\overline y_*$ and $(\xi_2-\xi_3)^2$ has a simple zero at these points.
\end{enumerate}

The points listed above are the only branch points of these functions.
\end{thm}


\begin{proof}

From the definition of $\xi_1$ in \eqref{solution_cauchy_transform}, it follows that it has to be branched at every point $p\in \{a_1,b_1,\hdots,a_l,b_l\}$. By Proposition~\ref{proposition_endpoint_support_branching} exactly one among $\xi_2$ and $\xi_3$ is also branched at this point.

In the singular unsaturated regime, the point $p=c_*$ is a branch point of the three functions by Definition~\ref{def:unsaturated_regular_singular} and Proposition~\ref{prop:interior_point_unified}. Uniqueness of such $c_*$ follows also from Proposition~\ref{prop:interior_point_unified}.

In the saturated regime, the quadratic differential $\varpi$ has a zero at $y_*^{(1)}$. By the construction of the sheet $\mathcal R_1$ in \eqref{construction_sheet_1} we know that the difference $(\xi_2-\xi_3)^2$ has to have a simple zero at this point, and consequently $\xi_2$ and $\xi_3$ have to be branched with square root behavior for their difference. Symmetry under conjugation gives us the same conclusion for $\overline y_*$.

To see that these are the only possible branch points, fix $p\in \C$. We proceed case by case
\begin{itemize}
\item If $p\in \C\setminus \R$, then from \eqref{solution_cauchy_transform} we know that $\xi_1$ cannot be branched at this point. As for the other two solutions $\xi_2$ and $\xi_3$, proceeding as above and using Proposition~\ref{prop_unique_zero} we see that they can only vanish at this $p$ if we are in the saturated regime and $p=y_*$ or $p=\overline y_*$, that is, if we are in case (c).

\item If $p\in \R\setminus \supp\lambda $, then by the representation \eqref{solution_cauchy_transform} the function $\xi_1$ cannot be branched at $p$. Also, from Corollary~\ref{cor:_gaps_orthogonal_trajectories} we learn that if $\xi_2-\xi_3$ vanishes at this point, it has to do so with integer order, so these functions are not branched at $p$ either.

\item The case when $p$ is an endpoint of $\supp\lambda $ falls into (a), so it was already dealt with.

\item In the case when $p$ is in the interior of $\supp\lambda $ with $\frac{d\lambda }{dx}(p)>0$, then none of the functions $\xi_1,\xi_2$ and $\xi_3$ is branched, as proved in Proposition~\ref{proposition_interior_support_branching}.

\item Finally, as a last case suppose that $p$ is in the interior of $\supp\lambda $ but $\frac{d\lambda }{dx}(p)=0$ and $p$ is a branch point for one of the solutions. Then by Proposition~\ref{prop:interior_point_unified}  it has to be a branch point of all three solutions. By Proposition~\ref{prop:interior_point_unified} and Definition~\ref{def:unsaturated_regular_singular}, we are  in the case (b).
\end{itemize}
\end{proof}

The previous proposition takes care of the book-keeping of branch points that will be necessary to construct the sheet structure of $\mathcal R$ in a very explicit manner. Our next step is to construct the sets that will support the measures $\mu_1^*$ and $\mu_2^*$ in Theorem \ref{thm_existence_critical_measure}, which will also provide a natural branch cut structure for $\mathcal R$.

In the unsaturated regime, an inspection of the possible geometries of $\partial \mathcal H$ in Figures~\ref{fig:orthogonal_half_plane_unsaturated_a}, \ref{fig:orthogonal_half_plane_unsaturated_2}  and \ref{fig:orthogonal_half_plane_unsaturated_bc}  shows that $\HH^{(1)}_+\setminus \overline{\mathcal H}$ has at most two connected components, which are necessarily unbounded. Likewise, in the saturated regime an inspection of Figure~\ref{fig:orthogonal_half_plane} shows that $\HH^{(1)}_+\setminus (\overline {\mathcal H}\cup \tau_0)$ also has two unbounded components.

In either case, we denote these components by $\mathcal G_1$ and $\mathcal G_2$, with the convention that they are possibly empty if the situation (III) of Proposition~\ref{prop:unsaturated_half_plane} takes place, and labeled in such a way that $\mathcal G_1$ is unbounded near $\infty$ for large and positive values of $\re z$, whereas $\mathcal G_2$ remains unbounded for large negative values of $\re z$. These sets are displayed in Figures~\ref{fig:sets_G} and \ref{fig:sets_G_unsaturated}, in the saturated and unsaturated regimes, respectively. 

The solution $\xi_1$ of the spectral curve is, from the very beginning, fixed on the whole plane by \eqref{solution_cauchy_transform}, while the other two solutions are well defined analytic functions in any simply connected domain $D$ of $\C$ not containing branch points. However, their labeling is fixed only by their asymptotic behavior \eqref{asymptotics_xi}, and is an issue to solve in the case of a bounded $D$. Up to this point, this distinction of labels was not relevant simply because we only cared about the {\it unsigned} difference of solutions   $\xi_2$ and $\xi_3$. 

The sets $\mathcal G_1$ and $\mathcal G_2$ are simply connected and do not contain branch points of the spectral curve. Thus, we can {\it uniquely} define $\xi_2$ and $\xi_3$ as analytic functions on $\mathcal G_j$, $j=1,2$, by performing analytic continuation from $\infty$ throughout the whole set $\mathcal G_j$ with the help of the asymptotics \eqref{asymptotics_xi}. With this convention in mind, for the next result, for points $z,w\in  \mathcal G_k$ set
$$
\Upsilon_k(z,w):= \int_{z}^{w}(\xi_2(s)-\xi_3(s))ds,\quad z,w\in \mathcal G_k, 
$$
where the path of integration is contained in $\mathcal G_k$; we extend $\Upsilon_k$ for finite points on the boundary of $\mathcal G_k$ by continuity.

\begin{lem}\label{lem:sign_im_part}
For $k=1,2$ and any points $x,y\in \R\cap \partial \mathcal G_k$, $x>y$,
\begin{equation}\label{signUpsilon}
(-1)^k\im \Upsilon_k(x,y)\geq 0.
\end{equation}
\end{lem}

\begin{proof}
Fix $x>y$ as in the proposition. Clearly the path of integration defining $\Upsilon_k$ is irrelevant, as long as it stays within $\overline{\mathcal G_k}$. Also, denoting $x_\varepsilon=x+i\varepsilon$, continuity tells us that it is enough to verify that $(-1)^k \im \Upsilon_k(x_\varepsilon,y_\varepsilon)\geq 0$ for any $\varepsilon>0$ sufficiently small.

Since $y_\varepsilon^{(1)}$ is a regular point of $\varpi$, there is a unique orthogonal trajectory passing through $y_\varepsilon^{(1)}$; in other words, there are two rays of orthogonal trajectories emanating from $y_\varepsilon^{(1)}$ in opposite directions. From Remark~\ref{rem_regular_orth_traj} we see that one of these trajectories, say $\tau$, has to extend to $\infty^{(1)}$ without intersecting $\R^{(1)}$. We orient $\tau$ from $y_\varepsilon^{(1)}$ to $\infty$. In a similar manner, there is an orthogonal trajectory $\gamma$ from $x_\varepsilon^{(1)}$ to $\infty$. 

Either $\gamma$ or $\tau$ contain both $x_\varepsilon^{(1)}$ and $y_\varepsilon^{(1)}$ (and then $\im \Upsilon_k(x_\varepsilon,y_\varepsilon)= 0$, which is consistent with \eqref{signUpsilon}), or $\gamma$ and $\tau$ are disjoint. 

In the latter case, both contours $\tau$ and $\gamma$ diverge to $\infty$ within $\mathcal G_k$ in the same asymptotic direction (with $\arg(\cdot) =0$ for $k=1$ and $\arg(\cdot)=\pi$ for $k=2$).

Thus, for any $R$ sufficiently large, both trajectories $\tau$ and $\gamma$ intersect $\overline{\mathcal G_k}\cap \{|z|=R\}$, say at points $u$ and $v$ respectively. We can write
$$
\Upsilon_k(x_\varepsilon,y_\varepsilon)=\int_{y_\varepsilon}^u( \xi_2(s)-\xi_3(s))ds + \int_{u}^v( \xi_2(s)-\xi_3(s))ds+\int_{v}^{x_\varepsilon}( \xi_2(s)-\xi_3(s))ds.
$$
The first and the last integrals above can be computed along $\tau$ and $\gamma$, which are orthogonal trajectories, so
$$
\im\Upsilon_k(x_\varepsilon,y_\varepsilon)=\im  \int_{u}^v( \xi_2(s)-\xi_3(s))ds.
$$
From the geometry of $\mathcal G_k$ and the fact that $\gamma$ and $\tau$ cannot intersect, we can take $R$ sufficiently large and make sure that $\im u>\im v$ when $k=1$ and $\im u<\im v$ when $k=2$. We then use the expansion \eqref{asymptotics_xi} to get that
$$
\int_{u}^v( \xi_2(s)-\xi_3(s))ds=2a(v-u)(1+\Boh(1)),
$$
and this yields the result.
\end{proof}

\begin{remark}
The claim in Lemma~\ref{lem:sign_im_part} has a natural interpretation in terms of the quadratic differential $\varpi$. Each set $\mathcal G_k$ can be decomposed as $\mathcal G_k=\HH^{(1)}_+\cap \left(\cup \mathcal S\right)$, where the union is over a finite number of sets $\mathcal S$, each one being a strip domain of the critical graph of $\varpi$. That is, the boundary of $\partial \mathcal S$ contains exactly two poles of $\varpi$ and finitely many zeros, and fixing a finite point $q\in \mathcal S$ the function $z\mapsto \Upsilon_k(q,z)$ is a conformal map from $\mathcal S$ to a domain bounded by two vertical lines on $\C$ - hence the name {\it strip domain}.

Lemma~\ref{lem:sign_im_part} is then saying that the (signed) height (or moduli parameter) associated to each $\mathcal S$ is the same, and as such one can extend $\Upsilon_k(q,\cdot)$ conformally to the whole union $\cup \mathcal S$ by gluing together some subarcs of the boundaries of these domains. 
\end{remark}

In the unsaturated regime, Theorem~\ref{thm:branchpoints} assures us that $\xi_2$ and $\xi_3$ can be analytically continued from $\infty$, using the expansion \eqref{asymptotics_xi}, to the upper half plane $\HH_+$. 

In the saturated regime, however, the upper half plane contains the branch point $y_*$ of $\xi_2$ and $\xi_3$. Thus, in this situation from now on we speak about the analytic continuation of these functions $\xi_2$ and $\xi_3$ from $\infty$ using \eqref{asymptotics_xi} to the simply connected domain $\HH_+\setminus \Delta_3$, where $\Delta_3$ is the contour defined in Proposition~\ref{prop:construction_s_curves_saturated}. The only branch point of these functions on $\HH_+$ is an endpoint of $\Delta_3$, so this analytic continuation is well defined. 

Observe also that this definition does not agree with the convention for $\xi_2$ and $\xi_3$ we used for Lemma~\ref{lem:sign_im_part}. To help in the discussion regarding this distinction, for simplicity let us denote by $\widehat \xi_2$ and $\widehat \xi_3$ the branches of these analytic functions which are obtained with analytic continuation from $\infty$ with \eqref{asymptotics_xi} to the domain $\HH_+\setminus \tau_0$, where $\tau_0=\overline{\mathcal G_1}\cap \overline{\mathcal G_2}$, as displayed in Figure~\ref{fig:sets_G}, and we continue denoting by $\xi_2,\xi_3$ the analytic continuations to $\HH_+\setminus \Delta_3$. Also, denote by $\mathcal D$ the domain on $\mathcal G_2$ bounded by $\tau_0$ and $\Delta_3$.

The branches $\widehat \xi_2$ and $\widehat \xi_3$ are the ones for which \eqref{lem:sign_im_part} is applicable, so for them this Lemma implies that
$$
\im(\widehat \xi_{2+}(x)-\widehat \xi_{3+}(x))dx<0,\quad x<c_*, \quad \mbox{and} \quad \im(\widehat \xi_{2+}(x)-\widehat \xi_{3+}(x))dx>0,\quad x>c_*.
$$
On the sets $\mathcal G_1$ and $\mathcal G_2\setminus \overline{\mathcal D}$, the branches $\xi_2$ and $\widehat \xi_2$ coincide, as well as the branches $\xi_3$ and $\widehat \xi_3$. This means that the inequalities above immediately transfer to 
\begin{equation*}
\im(\xi_{2+}(x)- \xi_{3+}(x))dx<0,\quad x<x_*, \quad \mbox{and} \quad \im( \xi_{2+}(x)- \xi_{3+}(x))dx>0,\quad x>c_*.
\end{equation*}
On the interval $(x_*,c_*)$, however, there is a change. In the domain $\mathcal D$ we actually have $\widehat \xi_3=\xi_2$ and $\widehat \xi_2=\xi_3$, and the previous inequality updates to
\begin{equation}\label{eq:inequalities_arrangement_supports_1}
\im(\xi_{2+}(x)- \xi_{3+}(x))dx<0,\quad x<x_*, \quad \mbox{and} \quad \im( \xi_{2+}(x)- \xi_{3+}(x))dx>0,\quad x>x_*.
\end{equation}

These calculations were valid in the saturated regime, but in the unsaturated regime one can simply take $\mathcal D=\emptyset$ and $x_*=c_*$ in the calculations above, and the inequalities \eqref{eq:inequalities_arrangement_supports_1} are still valid.

\begin{center}
	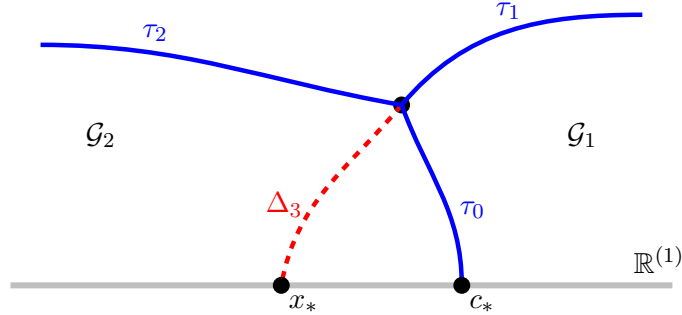
\begin{figure}[t]
		\begin{tikzpicture}[scale=0.8]
		\draw [line width=0.8mm,lightgray] (-5.5,0)--(5.5,0) node [pos=0.98,above,black] {$\R^{(1)}$};
		%
		%
		\fill (1,3) circle[radius=4pt] node [right,shift={(1pt,5pt)}] (cstar) {};
		%
		\draw [line width=0.6mm,blue] (5,4.5) to[out=180,in=50,edge node={node [pos=0.5,above] {$\tau_1$}}] (1,3); 
		\draw [line width=0.6mm,dashed,red] (1,3) to[out=50-180,in=80,edge node={node [pos=0.6,left] {$\Delta_3$}}] (-1,0);
			\draw [line width=0.6mm,blue] (1,3) to[out=110-180,in=90,edge node={node [pos=0.6,right] {$\tau_0$}}] (2,0); 
		\draw [line width=0.6mm,blue] (-5,4) to[out=0,in=170,edge node={node [pos=0.3,above] {$\tau_2$}}] (1,3); 
		%
		\fill (-1,0) circle[radius=4pt] node [below right,shift={(-1pt,0pt)}] {$x_*$};
	 	\fill (2,0) circle[radius=4pt] node [below right,shift={(-1pt,0pt)}] {$c_*$};
		%
		%
		%
		\node at (4,2.5) {$\mathcal G_1$};
		\node at (-4,2.5) {$\mathcal G_2$};
		\end{tikzpicture}
\caption{Saturated regime: the partition of the half plane $\HH_+^{(1)}$ induced by the orthogonal trajectories $\tau_0,\tau_1$ and $\tau_2$, the corresponding sets $\mathcal G_1$ and $\mathcal G_2$ and also the set $\Delta_3$, which is a vertical trajectory. This is an updated version of previous figures, compare for instance with Figures~\ref{fig:orthogonal_half_plane} and \ref{fig:gamma_sigma_vs_traj}.} 
		\label{fig:sets_G}
	\end{figure}
\end{center}

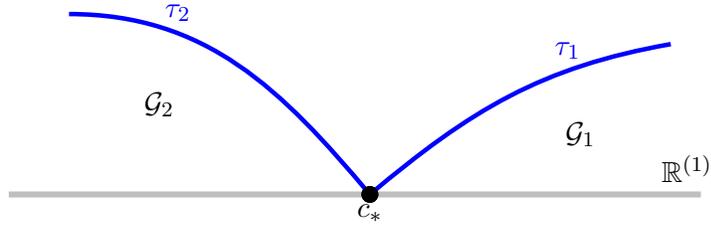
\begin{figure}[t]
	\begin{tikzpicture}[scale=0.8]
	\draw [line width=0.8mm,lightgray] (-5.5,0)--(6,0) node [pos=0.98,above,black] {$\R^{(1)}$};
	%
	%
	%
	%
	%
	%
	\draw [blue,line width=0.6mm] (-4.5,3) to [out=0,in=130,edge node={node [pos=0.3,above] {$\tau_2$}}] (0.5,0);
	\draw [line width=0.6mm,blue] (0.5,0) to [out=40,in=190,edge node={node [pos=0.7,above] {$\tau_1$}}] (5.5,2.5);
	%
	%
	\node at (4,1) {$\mathcal G_1$};
	\node at (-3,1.5) {$\mathcal G_2$};
	\fill (0.5,0) circle[radius=4pt];
	
	\fill (0.5,0) circle[radius=4pt] node [below,shift={(0pt,0pt)}] (cstar) {$c_*$}; 
	
	\end{tikzpicture}
	\caption{Unsaturated regime:  the partition of the half plane $\HH_+^{(1)}$ induced by the orthogonal trajectories $\tau_1$ and $\tau_2$ and the corresponding sets $\mathcal G_1$ and $\mathcal G_2$. This is an updated version of previous figures, compare for instance with Figures~\ref{fig:orthogonal_half_plane_unsaturated_a} and\ref{fig:gamma_sigma_vs_traj_unsaturated}, and corresponds only to the situation in Proposition~\ref{prop:unsaturated_half_plane}--(I), but the geometry of $\mathcal G_1$ and $\mathcal G_2$ in the situations (II)--(III) is analogous. }
	\label{fig:sets_G_unsaturated}
\end{figure}

Inspired  by the boundary value behavior \eqref{prep_cycles}, let us define the closed subsets $\Delta_1$ and $\Delta _2$ of $\R$ as follows:
\begin{equation} \label{defDeltas}
\Delta_j:= \overline{ \{ x\in \R:\,  \text{sign} (	\Im \left( \xi_2-\xi_3\right)_+(x))=(-1)^j \}} , \quad j=1,2.   
\end{equation}
Each $\Delta_j$ is a finite union of real intervals; by \eqref{lem_orthogonal_paths_eq_0} from Lemma~\ref{lem_imaginaryperiods}, $\xi_{1\pm}(x)=\xi_{(j+1)\mp}(x) $ on $\Delta_j$, 
\begin{equation}\label{eq:partition_support_1}
\supp\lambda =\Delta_1\cup \Delta_2,
\end{equation}
and $\Delta_1\cap \Delta_2$ a priori can have at most a finite number of points. Recall also that in the saturated regime we defined the real-symmetric and piece-wise analytic curve $\Delta_3$,  see Proposition~\ref{prop:construction_s_curves_saturated}. Our next step is to understand the relative position between $\Delta_1,\Delta_2$ and $\Delta_3$.

The inequalities \eqref{eq:inequalities_arrangement_supports_1} immediately imply the next result.

\begin{prop}\label{prop:structureDelta}
The sets $\Delta_1$ and $\Delta_2$ satisfy
$$
\Delta_1\subset [x_*,+\infty) \quad \mbox{and}\quad \Delta_2\subset (-\infty,x_*].
$$
Thus, either $\Delta_1\cap \Delta_2 = \emptyset$ or $\Delta_1\cap \Delta_2 = \{x_*\}$, and these situations can be further classified as follows.
\begin{enumerate}[(a)]
\item In the regular unsaturated regime, we always have $\Delta_1\cap \Delta_2 = \emptyset$, and we can take $x_*=c_*$ to be the unique critical point on the boundary of $\mathcal H$.
\item In the singular unsaturated regime, we always have $\Delta_1\cap \Delta_2 = \{x_*\}$ and $x_*=c_*$ is the unique point where the expansion \eqref{local_behavior_branch_point_pearcey} holds (see Proposition~\ref{thm:branchpoints}). In particular, $x_*$ is in the interior of $\supp\lambda$.
\item In the saturated regime, $\Delta_1\cap \Delta_2 = \emptyset$ or $\Delta_1\cap \Delta_2 = \{x_*\}$; in case the latter takes place, $x_*$ is the point of intersection of $\supp \lambda$ with $\Delta_3$ (see Proposition~\ref{prop:construction_s_curves_saturated}).
\end{enumerate}
\end{prop}

After all this work, we can finally eliminate the case (III) that appeared in Proposition~\ref{prop:unsaturated_half_plane}.

\begin{prop}\label{prop:exclusion_caseIII}
Suppose that we are in the unsaturated regime and case (III) of Proposition~\ref{prop:unsaturated_half_plane} takes place. Then the spectral curve \eqref{spectral_curve} is reducible and $\alpha\in \{0,1\}$.
\end{prop}
\begin{proof}
Suppose, for instance, that situation (III)--(a) is taking place. In this case $c_*=x_*\leq a_1$ and the construction just carried out assures that $\Delta_2\subset (-\infty,a_1]$. However, in this interval the spectral curve has no branch points (see Theorem~\ref{thm:branchpoints}), and consequently $\Delta_2=\emptyset$. That is, we have
$$
\im(\xi_{2+}(x)-\xi_{3+}(x))\geq 0,\quad x\in \R.
$$
But a combination of \eqref{lem_orthogonal_paths_eq_0} and  \eqref{prep_cycles} then tells us that $\xi_1$ and $\xi_2$ never share a branch point. Once again recalling Theorem~\ref{thm:branchpoints}, we also know that in the unsaturated regime $\xi_2$ and $\xi_3$ do not share branch points. This  means that $\xi_2$ does not have branch points at all, and thus it is entire. So the algebraic equation is reducible and  we can compute from \eqref{asymptotics_xi} that $\alpha=0$ and $\xi_2\equiv a$.

The case (III)--(b) leads to reducibility and $\alpha=1$ in a similar way.
\end{proof}

\begin{remark}\label{rmk:irreducible}
Actually, the argument outlined above says that if we were to allow $\alpha=[0,1]$ in Definition~\ref{definition_spectral_curve}, then the spectral curve \eqref{spectral_curve} would be reducible if, and only if, $\alpha\in \{0,1\}$, and in such a situation the unsaturated regime, case (III), would be taking place.
\end{remark}

With the help of Theorem~\ref{thm:branchpoints} and Proposition~\ref{prop:structureDelta} we are now ready to explicitly construct the Riemann surface $\mathcal R=\mathcal R_1\cup\mathcal R_2\cup \mathcal R_3$. The construction will be based on whether we are in the regular unsaturated, singular unsaturated or saturated regimes. In either case, as a final result we will find that each $\xi_j$ is analytic in $\mathcal R_j$, and the global meromorphic solution to \eqref{spectral_curve} is defined on the whole surface $\mathcal R$ as in \eqref{eq:global_solution}.

\subsection*{Riemann surface in the regular unsaturated regime}

By Theorem~\ref{thm:branchpoints}, in this case the collection of branch points of $\xi_1,\xi_2$ and $\xi_3$ is $\{ a_1,b_1,\hdots,a_l,b_l\}$, and by Proposition~\ref{prop:structureDelta}, 
$$
\Delta_1\cap \Delta_2=\emptyset.
$$
We take $\mathcal R_1$ as in \eqref{construction_sheet_1}, 
$$
\mathcal R_2:=\overline \C\setminus \Delta_1,\quad \mathcal R_3:=\overline \C\setminus \Delta_2,
$$
and $\mathcal R$ is realized as the three-sheeted branched cover
$$
\mathcal R=\mathcal R_1\cup\mathcal R_2\cup\mathcal R_3
$$
of $\overline \C$, with $\mathcal R_1$ glued with $\mathcal R_j$ along $\Delta_{j-1}$ in the usual crosswise manner, $j=2,3$. In particular, in this case there is no direct connection between the sheets $\mathcal R_2$ and $\mathcal R_3$. This construction is illustrated in Figure~\ref{fig:sheet_config_unsaturated}, left.

By the Riemann-Hurwitz formula, the genus of $\mathcal R$ in this case is
$$
g=l-2,
$$
where $l$ is the number of disjoint components of the support $\supp\lambda$.

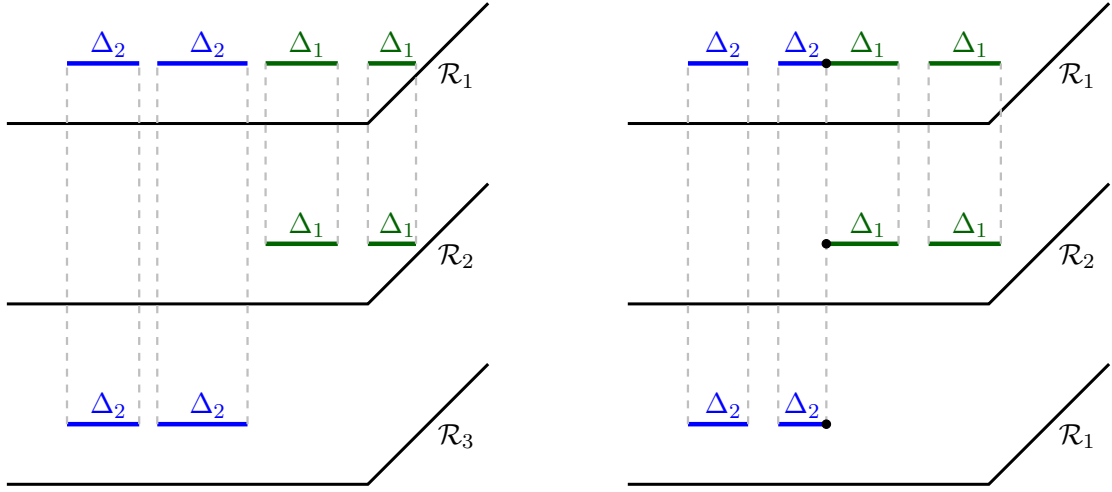
\begin{figure}
\begin{center}
\begin{subfigure}{0.5\textwidth}\centering
\begin{tikzpicture}[x  = {(1cm,0cm)},
                    y  = {(1cm,1cm)},
                    z  = {(0cm,1cm)},
                    scale = 0.8,
                    color = {black}]
%
%
%
%
%

%
%
\begin{scope}[canvas is xy plane at z=6]
\path[draw,line width=0.4mm] (-3,-1) to (3,-1) to [edge node={node [pos=0.4,right,shift={(4pt,0pt)}] {$\mathcal R_1$}}] (3,1);
%
%
\path [draw,line width=0.6mm,blue] (-3,0) to [edge node={node [pos=0.5,above,shift={(2pt,-2pt)}] {$\Delta_2$}}] (-1.8,0);
\path [draw,line width=0.6mm,blue] (-1.5,0) to [edge node={node [pos=0.5,above,shift={(2pt,-2pt)}] {$\Delta_2$}}] (0,0);
\path [draw,line width=0.6mm,green] (2,0) to [edge node={node [pos=0.5,above,shift={(2pt,-2pt)}] {$\Delta_1$}}] (2.8,0);
\path [draw,line width=0.6mm,green] (0.3,0) to [edge node={node [pos=0.5,above,shift={(2pt,-2pt)}] {$\Delta_1$}}] (1.5,0);
\end{scope}
%
%

%
%
\begin{scope}[canvas is xy plane at z=3]
\path[draw,line width=0.4mm] (-3,-1) to (3,-1) to [edge node={node [pos=0.4,right,shift={(4pt,0pt)}] {$\mathcal R_2$}}] (3,1);
%
%
%
\path [draw,line width=0.6mm,green] (0.3,0) to [edge node={node [pos=0.5,above,shift={(2pt,-2pt)}] {$\Delta_1$}}] (1.5,0);
\path [draw,line width=0.6mm,green] (2,0) to [edge node={node [pos=0.5,above,shift={(2pt,-2pt)}] {$\Delta_1$}}] (2.8,0);
\end{scope}
%
%
%
\begin{scope}[canvas is xy plane at z=0]
\path[draw,line width=0.4mm] (-3,-1) to (3,-1) to [edge node={node [pos=0.4,right,shift={(4pt,0pt)}] {$\mathcal R_3$}}] (3,1);
%
%
\path [draw,line width=0.6mm,blue] (-1.5,0) to [edge node={node [pos=0.5,above,shift={(2pt,-2pt)}] {$\Delta_2$}}] (0,0);
\path [draw,line width=0.6mm,blue] (-3,0) to [edge node={node [pos=0.5,above,shift={(2pt,-2pt)}] {$\Delta_2$}}] (-1.8,0);
\end{scope}

%
%
\draw[line width=0.3mm,dashed,lightgray] (0.3,0,3)--(0.3,0,6);
\draw[line width=0.3mm,dashed,lightgray] (1.5,0,3)--(1.5,0,6);
\draw[line width=0.3mm,dashed,lightgray] (-1.5,0,0)--(-1.5,0,6);
\draw[line width=0.3mm,dashed,lightgray] (0,0,0)--(0,0,6);
\draw[line width=0.3mm,dashed,lightgray] (-3,0,0)--(-3,0,6);
\draw[line width=0.3mm,dashed,lightgray] (-1.8,0,0)--(-1.8,0,6);
\draw[line width=0.3mm,dashed,lightgray] (2,0,3)--(2,0,6);
\draw[line width=0.3mm,dashed,lightgray] (2.8,0,3)--(2.8,0,6);
\end{tikzpicture}
\end{subfigure}%
\begin{subfigure}{0.5\textwidth}\centering
\begin{tikzpicture}[x  = {(1cm,0cm)},
                    y  = {(1cm,1cm)},
                    z  = {(0cm,1cm)},
                    scale = 0.8,
                    color = {black}]
%
%
%
%
%
%
%
\begin{scope}[canvas is xy plane at z=6]
\path[draw,line width=0.4mm] (-3,-1) to (3,-1) to [edge node={node [pos=0.4,right,shift={(4pt,0pt)}] {$\mathcal R_1$}}] (3,1);
%
%
%
%
\path [draw,line width=0.6mm,blue] (-1.5,0) to [edge node={node [pos=0.5,above,shift={(0pt,-2pt)}] {$\Delta_2$}}] (-0.7,0);
\path [draw,line width=0.6mm,blue] (-3,0) to [edge node={node [pos=0.5,above,shift={(2pt,-2pt)}] {$\Delta_2$}}] (-2,0);
\path [draw,line width=0.6mm,green] (-0.7,0) to [edge node={node [pos=0.5,above,shift={(2pt,-2pt)}] {$\Delta_1$}}] (0.5,0);
\path [draw,line width=0.6mm,green] (1,0) to [edge node={node [pos=0.5,above,shift={(2pt,-2pt)}] {$\Delta_1$}}] (2.2,0);

\end{scope}
%
%
%
\begin{scope}[canvas is xy plane at z=3]
\path[draw,line width=0.4mm] (-3,-1) to (3,-1) to [edge node={node [pos=0.4,right,shift={(4pt,0pt)}] {$\mathcal R_2$}}] (3,1);
%
%
\path [draw,line width=0.6mm,green] (-0.7,0) to [edge node={node [pos=0.5,above,shift={(2pt,-2pt)}] {$\Delta_1$}}] (0.5,0);
\path [draw,line width=0.6mm,green] (1,0) to [edge node={node [pos=0.5,above,shift={(2pt,-2pt)}] {$\Delta_1$}}] (2.2,0);
\end{scope}
%
%
\begin{scope}[canvas is xy plane at z=0]
\path[draw,line width=0.4mm] (-3,-1) to (3,-1) to [edge node={node [pos=0.4,right,shift={(4pt,0pt)}] {$\mathcal R_1$}}] (3,1);
%
%
\path [draw,line width=0.6mm,blue] (-1.5,0) to [edge node={node [pos=0.5,above,shift={(0pt,-2pt)}] {$\Delta_2$}}] (-0.7,0);
\path [draw,line width=0.6mm,blue] (-3,0) to [edge node={node [pos=0.5,above,shift={(2pt,-2pt)}] {$\Delta_2$}}] (-2,0);
\end{scope}

%
%
\draw[line width=0.3mm,dashed,lightgray] (-1.5,0,0) -- (-1.5,0,6);
\draw[line width=0.3mm,dashed,lightgray] (1,0,3) -- (1,0,6);
\draw[line width=0.3mm,dashed,lightgray] (2.2,0,3) -- (2.2,0,6);
\draw[line width=0.3mm,dashed,lightgray] (0.5,0,3) -- (0.5,0,6);
\draw[line width=0.3mm,dashed,lightgray] (-3,0,0) -- (-3,0,6);
\draw[line width=0.3mm,dashed,lightgray] (-2,0,0) -- (-2,0,6);
\draw[line width=0.3mm,dashed,lightgray] (-0.7,0,0) -- (-0.7,0,6);
%
%
\filldraw (-0.7,0,0) circle[radius=2pt]; 
\filldraw (-0.7,0,3) circle[radius=2pt]; 
\filldraw (-0.7,0,6) circle[radius=2pt];
\end{tikzpicture}
\end{subfigure}
\caption{The sheet configurations in the regular unsaturated regime (left) and singular unsaturated regime (right). In both figures, the solid lines correspond to the cuts and the dashed gray lines indicate the common branch points to the sheets. In the singular unsaturated regime there is a branch point common to the three sheets, this point is indicated by the black dot on the three sheets.}\label{fig:sheet_config_unsaturated}
\end{center}
\end{figure}

\subsection*{Riemann surface in the singular unsaturated regime}

By Theorem~\ref{thm:branchpoints}, the branch points are $ a_1,b_1,\hdots,a_l,b_l$ and $x_*$, and $\Delta_1\cap \Delta_2 = \{x_*\}$, where $x_*$ is described by the option b) of Proposition~\ref{prop:structureDelta}. 

As in the previous case, 
$\mathcal R = \mathcal R_1\cup \mathcal R_2\cup \mathcal R_3$, observing that $x_*$ is common to the three sheets, and it is also the only common point to $\mathcal R_2$ and $\mathcal R_3$. This construction is shown in Figure~\ref{fig:sheet_config_unsaturated}, right.

The genus of $\mathcal R$ in this case is still
$$
g=l-2,
$$
where $l$ is the number of disjoint components of the support $\supp\lambda$. Notice that in the transition from the regular to the singular unsaturated regime the number of components decreases in $1$, but the genus of $\mathcal R$ remains unchanged.

\subsection*{Riemann surface in the saturated regime}

Once again by Theorem~\ref{thm:branchpoints}, the branch points are $ a_1,b_1,\hdots,a_l,b_l$, and also $y_*$ and $\overline y_*$, the latter two being common branch points to $\xi_2$ and $\xi_3$.

In any case, by Proposition~\ref{prop:structureDelta} either $\Delta_1\cap \Delta_2 = \{x_*\}$ or $\Delta_1\cap \Delta_2=\emptyset$ (when $\Delta_3\cap \R$ is not in $\supp\lambda $).

We  set $\mathcal R_1$ as in \eqref{construction_sheet_1},
$$
\mathcal R_2:=\overline \C \setminus \left( \Delta_1\cup \Delta_3 \right),\quad \mathcal R_3:=\overline \C \setminus \left( \Delta_2\cup \Delta_3 \right),
$$
and glue these sheets together in
$$
\mathcal R=\mathcal R_1\cup\mathcal R_2\cup\mathcal R_3,
$$
with $\mathcal R_1$ connected to $\mathcal R_j$ along $\Delta_{j-1}$, $j=2,3$, and $\mathcal R_2$ and $\mathcal R_3$ connected along $\Delta_3$, always in the usual crosswise manner. This sheet configuration is illustrated in Figure~\ref{fig:sheet_config_saturated}.

Again the Riemann-Hurwitz formula gives us that now the genus of $\mathcal R$ is
$$
g=l-1,
$$
where $l$ is the number of disjoint components of the support $\supp\lambda$.

\begin{figure}
\begin{center}
\begin{subfigure}{0.5\textwidth}\centering
\begin{tikzpicture}[x  = {(1cm,0cm)},
                    y  = {(1cm,1cm)},
                    z  = {(0cm,1cm)},
                    scale = 0.8,
                    color = {black}]
%
%
%
%
%

%
%
\begin{scope}[canvas is xy plane at z=6]
\path[draw,line width=0.4mm] (-3,-1) to (3,-1) to [edge node={node [pos=0.4,right,shift={(4pt,0pt)}] {$\mathcal R_1$}}] (3,1);
%
%
\path [draw,line width=0.6mm,blue] (-1.5,0) to [edge node={node [pos=0.5,above,shift={(2pt,-2pt)}] {$\Delta_2$}}] (0,0);
\path [draw,line width=0.6mm,blue] (-3,0) to [edge node={node [pos=0.5,above,shift={(2pt,-2pt)}] {$\Delta_2$}}] (-1.8,0);
\path [draw,line width=0.6mm,green] (1,0) to [edge node={node [pos=0.5,above,shift={(2pt,-2pt)}] {$\Delta_1$}}] (1.8,0);
\path [draw,line width=0.6mm,green] (2.1,0) to [edge node={node [pos=0.5,above,shift={(2pt,-2pt)}] {$\Delta_1$}}] (2.85,0);
\end{scope}
%
%
%
\begin{scope}[canvas is xy plane at z=3]
\path[draw,line width=0.4mm] (-3,-1) to (3,-1) to [edge node={node [pos=0.4,right,shift={(4pt,0pt)}] {$\mathcal R_2$}}] (3,1);
%
%
\draw [line width=0.6mm,red] (0.8,0.5) .. controls (0.3,0.2) and (0.3,-0.2) .. (0.8,-0.5) node [pos=0.2,above] {$\Delta_3$};
%
%
\path [draw,line width=0.6mm,green] (2.1,0) to [edge node={node [pos=0.45,above,shift={(2pt,-2pt)}] {$\Delta_1$}}] (2.85,0);
\path [draw,line width=0.6mm,green] (1,0) to [edge node={node [pos=0.75,below,shift={(-2pt,2pt)}] {$\Delta_1$}}] (1.8,0);
\end{scope}

%
%
\begin{scope}[canvas is xy plane at z=0]
\path[draw,line width=0.4mm] (-3,-1) to (3,-1) to [edge node={node [pos=0.4,right,shift={(4pt,0pt)}] {$\mathcal R_3$}}] (3,1);
%
%
\draw [line width=0.6mm,red] (0.8,0.5) .. controls (0.3,0.2) and (0.3,-0.2) .. (0.8,-0.5) node [pos=0.2,above] {$\Delta_3$};
%
%
\path [draw,line width=0.6mm,blue] (-1.5,0) to [edge node={node [pos=0.5,above,shift={(2pt,-2pt)}] {$\Delta_2$}}] (0,0);
\path [draw,line width=0.6mm,blue] (-3,0) to [edge node={node [pos=0.5,above,shift={(2pt,-2pt)}] {$\Delta_2$}}] (-1.8,0);
\end{scope}

%
%
\draw[line width=0.3mm,dashed,lightgray] (0.8,0.5,3) -- (0.8,0.5,0);
\draw[line width=0.3mm,dashed,lightgray] (0.8,-0.5,3) -- (0.8,-0.5,0);
\draw[line width=0.3mm,dashed,lightgray] (1,0,3)--(1,0,6);
\draw[line width=0.3mm,dashed,lightgray] (1.8,0,3)--(1.8,0,6);
\draw[line width=0.3mm,dashed,lightgray] (-1.5,0,0)--(-1.5,0,6);
\draw[line width=0.3mm,dashed,lightgray] (0,0,0)--(0,0,6);
\draw[line width=0.3mm,dashed,lightgray] (2.1,0,3)--(2.1,0,6);
\draw[line width=0.3mm,dashed,lightgray] (2.85,0,3)--(2.85,0,6);
\draw[line width=0.3mm,dashed,lightgray] (-3,0,0)--(-3,0,6);
\draw[line width=0.3mm,dashed,lightgray] (-1.8,0,0)--(-1.8,0,6);
\end{tikzpicture}
\end{subfigure}%
\begin{subfigure}{0.5\textwidth}\centering
\begin{tikzpicture}[x  = {(1cm,0cm)},
                    y  = {(1cm,1cm)},
                    z  = {(0cm,1cm)},
                    scale = 0.8,
                    color = {black}]
%
%
\begin{scope}[canvas is xy plane at z=6]
\path[draw,line width=0.4mm] (-3,-1) to (3,-1) to [edge node={node [pos=0.4,right,shift={(4pt,0pt)}] {$\mathcal R_1$}}] (3,1);
%
%
%
\path [draw,line width=0.6mm,blue] (-1.5,0) to [edge node={node [pos=0.5,above,shift={(2pt,-2pt)}] {$\Delta_1$}}] (-0.7,0);

\path [draw,line width=0.6mm,blue] (-3,0) to [edge node={node [pos=0.5,above,shift={(2pt,-2pt)}] {$\Delta_1$}}] (-2,0);
\path [draw,line width=0.6mm,green] (1,0) to [edge node={node [pos=0.5,above,shift={(2pt,-2pt)}] {$\Delta_2$}}] (2.2,0);
\path [draw,line width=0.6mm,green] (-0.7,0) to [edge node={node [pos=0.5,above,shift={(2pt,-2pt)}] {$\Delta_2$}}] (0.5,0);

\end{scope}
%
%
%
%
\begin{scope}[canvas is xy plane at z=3]
\path[draw,line width=0.4mm] (-3,-1) to (3,-1) to [edge node={node [pos=0.4,right,shift={(4pt,0pt)}] {$\mathcal R_2$}}] (3,1);
%
%
\path [draw,line width=0.6mm,green] (1,0) to [edge node={node [pos=0.5,below,shift={(2pt,2pt)}] {$\Delta_1$}}] (2.2,0);
\path [draw,line width=0.6mm,green] (-0.7,0) to [edge node={node [pos=0.5,below,shift={(2pt,2pt)}] {$\Delta_1$}}] (0.5,0);
%
%
\draw [line width=0.6mm,red] (0,0.7) .. controls (-0.3,0.5) and (-0.6,0.3) .. (-0.7,0) node [pos=0.5,above] {$\Delta_3$};
\draw [line width=0.6mm,red] (0,-0.7) .. controls (-0.3,-0.5) and (-0.6,-0.3) .. (-0.7,0);
\end{scope}
%
%
%
\begin{scope}[canvas is xy plane at z=0]
\path[draw,line width=0.4mm] (-3,-1) to (3,-1) to [edge node={node [pos=0.4,right,shift={(4pt,0pt)}] {$\mathcal R_3$}}] (3,1);
%
%
\path [draw,line width=0.6mm,blue] (-1.5,0) to [edge node={node [pos=0.5,above,shift={(0pt,-2pt)}] {$\Delta_2$}}] (-0.7,0);
\path [draw,line width=0.6mm,blue] (-3,0) to [edge node={node [pos=0.5,above,shift={(2pt,-2pt)}] {$\Delta_2$}}] (-2,0);
%
%
\draw [line width=0.6mm,red] (0,0.7) .. controls (-0.3,0.5) and (-0.6,0.3) .. (-0.7,0) node [pos=0.5,above] {$\Delta_3$};
\draw [line width=0.6mm,red] (0,-0.7) .. controls (-0.3,-0.5) and (-0.6,-0.3) .. (-0.7,0);
\end{scope}
%
%
%
\draw[line width=0.3mm,dashed,lightgray] (-0.7,0,3) -- (-0.7,0,6);
\draw[line width=0.3mm,dashed,lightgray] (0,-0.7,3) -- (0,-0.7,0);
\draw[line width=0.3mm,dashed,lightgray] (0,0.7,3) -- (0,0.7,0);
\draw[line width=0.3mm,dashed,lightgray] (-1.5,0,0) -- (-1.5,0,6);
\draw[line width=0.3mm,dashed,lightgray] (1,0,3) -- (1,0,6);
\draw[line width=0.3mm,dashed,lightgray] (2.2,0,3) -- (2.2,0,6);
\draw[line width=0.3mm,dashed,lightgray] (0.5,0,3) -- (0.5,0,6);
\draw[line width=0.3mm,dashed,lightgray] (-3,0,0) -- (-3,0,6);
\draw[line width=0.3mm,dashed,lightgray] (-2,0,0) -- (-2,0,6);
\end{tikzpicture}
\end{subfigure}
\caption{The sheet configurations in the saturated regime. On the left, the case when $\Delta_3\cap \supp\lambda =\emptyset$ and on the right when $\Delta_3\cap \supp\lambda = \{x_*\}$. In both figures, the solid lines correspond to the cuts and the dashed gray lines indicate the common branch points to the sheets. }\label{fig:sheet_config_saturated}
\end{center}
\end{figure}
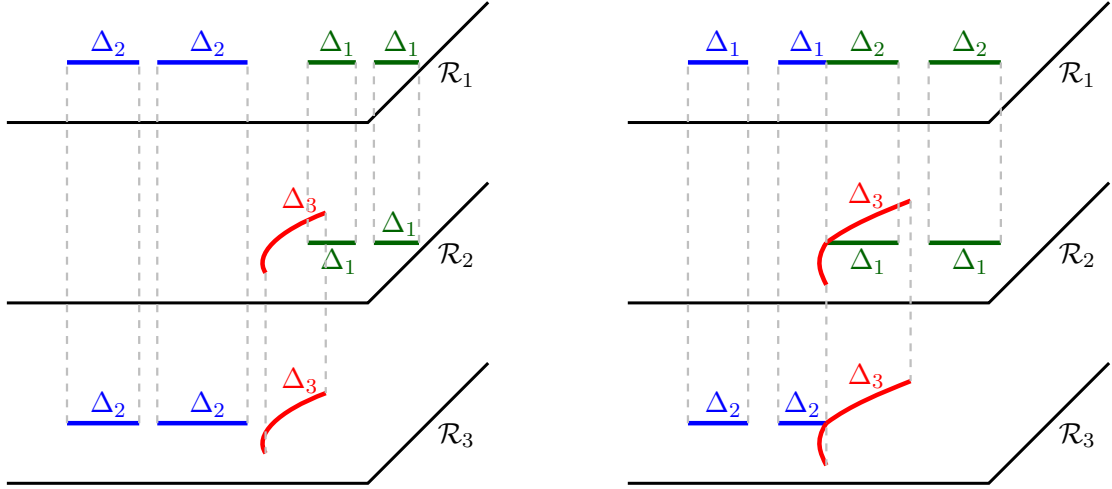

\subsection{From the Riemann surface to critical measures} \label{sec:6}

\

We are finally in the position to construct the measures $\mu_1^*,\mu_2^*$ and $\mu_3^*$ for Theorem~\ref{thm_existence_critical_measure}. To do so, we will make use of the sets $\Delta_j$ that were already introduced in the previous section. Recall also the set $\Gamma_*$ formally defined in Proposition~\ref{prop:construction_s_curves_saturated} (in the saturated regime) and Proposition~\ref{prop:construction_trajectories_s_contour_unsaturated} (in the unsaturated regime). We orient $\Gamma_*$ from the lower to the upper half plane, and in the saturated regime the set $\Delta_3\subset \Gamma_*$ inherits the orientation from $\Gamma_*$. In the unsaturated regime we set $\Delta_3=\emptyset$.

Define measures $\mu_1^*,\mu_2^*$ and $\mu_3^*$ through their densities as
\begin{equation}\label{def:target_critical_measures}
\begin{aligned}
& d\mu_1^*(s)=\restr{d\lambda (s)}{\Delta_1}=\frac{1}{2\pi i}(\xi_{1+}(s)-\xi_{1-}(s))ds =  \frac{1}{2\pi i}(\xi_{1}-\xi_{2})_+ (s)ds, \; s\in \Delta_1, \\
& d\mu_2^*(s)=\restr{d\lambda (s)}{\Delta_2}=\frac{1}{2\pi i}(\xi_{1+}(s)-\xi_{1-}(s))ds=  \frac{1}{2\pi i}(\xi_{1}-\xi_{3})_+ (s)ds, \; s\in \Delta_2, \\
& d\mu_3^*(s)=\frac{1}{2\pi i}(\xi_{2+}(s)-\xi_{3+}(s))ds, \; s\in \Delta_3,
\end{aligned}
\end{equation}
with the convention that $\mu_3^*=0$ if $\Delta_3=\emptyset$, that is, if we are in the unsaturated regime. Notice that using \eqref{prep_cycles} we can alternatively express all three $\mu_j^*$'s in terms of the boundary values of $\xi_{2+}-\xi_{3+}$ on $\Delta_j$'s.

\begin{prop}\label{prop:relation_cauchy_transf_xi}
The measures $\mu_1^*$, $\mu_2^*$ and $\mu_3^*$ defined in \eqref{def:target_critical_measures} are positive and their Cauchy transforms satisfy the identities
\begin{equation}\label{eq:relations_cauchy_transf_xi}
\begin{aligned}
& \xi_1(z)=C^{\mu_1^*}(z)+C^{\mu_2^*}(z)+V'(z), \\
& \xi_2(z)=C^{\mu_3^*}(z)-C^{\mu_1^*}(z)+a,  \\
& \xi_3(z)=-C^{\mu_2^*}(z)-C^{\mu_3^*}(z)-a, 
\end{aligned}
\quad z\in \C\setminus (\Delta_1\cup\Delta_2\cup\Delta_3).
\end{equation}
In particular, the following relations between their total masses holds:
$$
|\mu_1^*|+|\mu_2^*|=1,\quad |\mu_1^*|-|\mu_3^*|=\alpha,\quad |\mu_2^*|+|\mu_3^*|=1-\alpha.
$$
Moreover, if $ \Delta$ is a connected component of $  \Delta_1 \cup \Delta_2 \cup \Delta_3 $  then
\begin{equation} \label{cycles}
	 (-\mu_1^*+\mu_2^*+2\mu_3^*)(\Delta) 
	 = \frac{1}{2\pi i} \ointclockwise_{ \Delta^{(1)}} \sqrt{-\varpi} =   \frac{1}{2\pi i} \ointclockwise_{ \Delta } (\xi_{2}(s)-\xi_{3}(s))ds,
\end{equation}
where the contour integral is understood as a small loop encircling $\Delta$ in the clockwise direction, and separating it from the other components.
\end{prop}
\begin{proof}
Measures $\mu_1^*$ and $\mu_2^*$ are positive because they are restrictions of a positive measure $\lambda$ (see Equation \eqref{def:target_critical_measures}). Also, from Proposition~\ref{prop:construction_s_curves_saturated} we learn that the density of $\mu_3^*$, whenever $\mu_3^*\neq 0$, is real and does not change sign on $\Delta_3\cap \HH_+$, so in order to prove that $\mu_3^*$ is positive on $\Delta_3\cap \HH_+$ it is enough to show that this density is positive near $y_*$. To see the latter, denote
$$
\Phi(z)=\int_{y_*}^{z}(\xi_{2}(s)-\xi_{3}(s))ds,\quad z\in \HH_+\setminus \Delta_3,
$$
This function $\Phi$ is holomorphic in its domain of definition, and according to Proposition~\ref{prop:construction_s_curves_saturated} it satisfies
\begin{equation}\label{aux_equation_12}
\Phi(z)\in \R \text{ when } z\in \tau_*\cap \HH_+,\quad \text{and} \quad  \Phi(z)\in i\R \text{ when } z\in \Delta_3\cap \HH_+.
\end{equation}
Now, the contour $\tau_*$, by its definition, extends to $\infty$ with angle $0$ (see Figure~\ref{fig:gamma_sigma_vs_traj}). From the asymptotics \eqref{asymptotics_xi} we also get that
$$
\Phi(z)=2az(1+\boh(1)), \quad z\to \infty,
$$
which shows that actually $\Phi(z)>0$ for $z\in \tau_*\cap \HH_+$ in a neighborhood of $\infty$. Because $\xi_2-\xi_3$ has no zeros along $\tau_*\cap\HH_+$ (check again Proposition~\ref{prop:construction_s_curves_saturated}), we actually get that
$$
\Phi(\tau_*\cap \HH_*)=(0,+\infty).
$$
At $z=y_*$, the difference $\xi_2-\xi_3$ vanishes as square root, so $\Phi$ vanishes with power $3/2$ at $z=y_*$. To rotate from $\tau_*$ to $\Delta_3$ at $y_*$ in the counter-clockwise direction, we change the angle by $\pi$, so in the image of $\Phi$ this becomes a rotation by $3\pi/2$ in the counter-clockwise direction. Combining with \eqref{aux_equation_12}, this means that for some $\delta>0$,
$$
\Phi_+(z)\in i\R_-,\quad z\in D_\delta(y_*)\cap \Delta_3.
$$
This is equivalent to saying that $d\mu_3^*$ is positive along $D_{\delta}(y_*)\cap \HH_+$, and by the previous observation we consequently conclude that $\mu_3^*$ is positive everywhere on $\Delta_3\cap \HH_+$. 

The proof that $\mu_3^*$ is positive on $\Delta_3\cap \HH_-$ follows along the same lines, or alternatively by symmetry under complex conjugation. We omit the details.

The first equation in \eqref{eq:relations_cauchy_transf_xi} is easily derived once we observe that $\mu_1^*+\mu_2^*=\lambda $ and recall \eqref{solution_cauchy_transform}. To get the second one, fix $x\in \C\setminus (\Delta_1\cup\Delta_2\cup\Delta_3)$. For a contour $\gamma$ enclosing $\Delta_1\cup\Delta_2\cup\Delta_3$ in the counter-clockwise orientation and with $z$ in its exterior, deformation of contours shows that
\begin{align*}
C^{\mu_1^*}(z)-C^{\mu_3^*}(z) & = \frac{1}{2\pi i}\int_{\Delta_1}\frac{\xi_{1+}(s)-\xi_{1-}(s)}{s-z}ds - \frac{1}{2\pi i}\int_{\Delta_3} \frac{\xi_{2+}(s)-\xi_{3+}(s)}{s-z}ds \\
						 & = \frac{1}{2\pi i}\int_{\Delta_1}\frac{\xi_{2-}(s)-\xi_{2+}(s)}{s-z}ds + \frac{1}{2\pi i}\int_{\Delta_3} \frac{\xi_{2-}(s)-\xi_{2+}(s)}{s-z}ds \\
						 & = \frac{1}{2\pi i}\oint_{\gamma}\frac{\xi_2(s)}{s-z}ds,
\end{align*}
and by the residues theorem and the expansion \eqref{asymptotics_xi},
$$
C^{\mu_1^*}(z)-C^{\mu_3^*}(z)=-\res \left(\frac{\xi_2(s)}{s-z},s=z\right)-\res \left(\frac{\xi_2(s)}{s-z},s=\infty\right)=-\xi_{2}(z)+a,
$$
which is equivalent to the second equation in \eqref{eq:relations_cauchy_transf_xi}. For the third equation, simply add the first two  and recall that $\xi_1+\xi_2=V'-\xi_3$, see \eqref{alg_relations}.

The relations between the masses then follow from \eqref{eq:relations_cauchy_transf_xi} and the expansion \eqref{asymptotics_xi}, having in mind that for any finite complex measure $\sigma$ with compact support, the asymptotics
$$
C^{\sigma}(z)=-\frac{\sigma(\C)}{z}+\Boh(z^{-2}),\quad z\to\infty,
$$
holds.

Finally, \eqref{cycles} is obtained using the explicit expressions \eqref{eq:relations_cauchy_transf_xi} and interchanging the order of integrals in the right hand side.
\end{proof}

Set
$$
\vec\mu_* :=(\mu_1^*,\mu_2^*,\mu_3^*)
$$
with components being the measures in \eqref{def:target_critical_measures}, and let $\Gamma_*$ be as in Propositions~\ref{prop:construction_s_curves_saturated} and \ref{prop:construction_trajectories_s_contour_unsaturated}. The fact that $\vec \mu \in \mathcal M(\Gamma_*)$ is an immediate consequence of Proposition~\ref{prop:relation_cauchy_transf_xi}.

We now show that $\vec \mu_*$ is critical. The transformation $\xi\mapsto \zeta=\xi-\frac{V'}{3}$ transforms \eqref{spectral_curve} in an equation of the form
$$
\zeta^3-R(z)\zeta+D(z)=0,
$$
for some polynomials $R$ and $D$. The solutions to this equation are
$$
\begin{aligned}
& \zeta_1(z)=\xi_1(z)-\frac{V'(z)}{3}=C^{\mu_1}(z)+C^{\mu_2}(z)+Q_1'(z), && \quad Q_1(z):=\frac{2V(z)}{3}, \\
& \zeta_2(z)=\xi_3(z)-\frac{V'(z)}{3}=-C^{\mu_2}(z)-C^{\mu_3}(z)+Q_2'(z),&& \quad Q_2(z):=-az-\frac{V(z)}{3},\\
& \zeta_3(z)=\xi_2(z)-\frac{V'(z)}{3}=C^{\mu_3}(z)-C^{\mu_1}(z)+Q_3'(z), && \quad Q_3(z):=az-\frac{V(z)}{3}.
\end{aligned}
$$
By \cite[Theorem~1.3]{martinez_silva_critical_measures} the vector of measures $(\mu_2^*,\mu_1^*,\mu_3^*)$ is critical for the energy with interaction matrix
$$
\begin{pmatrix}
1 & 1/2 & 1/2 \\
1/2 & 1 & -1/2 \\
1/2 & -1/2 & 1
\end{pmatrix}
$$
and potentials
$$
\begin{aligned}
& \Phi_1(z)=Q_1(z)-Q_2(z)= V(z)+az=V_2(z), \\
& \Phi_2(z)=Q_1(z)-Q_3(z)=V(z)-az=V_1(z), \\
& \Phi_3(z)=Q_3(z)-Q_2(z)=2az=V_3(z).
\end{aligned}
$$
Interchanging the roles of $\mu_1^*$ and $\mu_2^*$, we thus get that $\vec\mu_*=( \mu_1^*,\mu_2^*,\mu_3^*)$ is critical for our energy $E_S(\cdot)$.

The $S$-property \eqref{vector_s_property} is an immediate consequence of the fact that $\vec\mu_*$ is critical, see \cite[Theorem~1.8]{martinez_silva_critical_measures}.

Finally, the classification of possible singular behaviors of $\lambda $ follows from the sheet structure for $\xi_1$ and \eqref{eq:relations_cauchy_transf_xi}.

For a Borel measure $\mu$ on $\C$ its logarithmic potential $U^\mu$ and its Cauchy transform $C^\mu$ are related through
\begin{equation}\label{relation_cauchy_transform_potential}
\frac{\partial U^\mu}{\partial z}(z)=\frac{1}{2}C^{\mu}(z),\quad z\in \C\setminus\supp\mu;
\end{equation}
under suitable assumptions on $\mu$ and its support, this identity can be further extended to boundary points of $\supp\mu$ through appropriate limits, as well as to $\supp\mu$ in the distributional sense.  Using this relation, an immediate consequence of \eqref{eq:relations_cauchy_transf_xi} are the following expressions:
\begin{lem}\label{lem:identity_potential_xi}
Let $y\in \Delta_1\cup\Delta_2\cup\Delta_3$. There are real constants $l_1,l_2$ and $l_3$, possibly depending on $y$, such that for $z\in \C \setminus (\R \cup \Delta_3)$, the logarithmic potentials of the measures $\mu_1^*$, $\mu_2^*$ and $\mu_3^*$ satisfy
\begin{align*}
& U^{\mu_1^*}(z)+U^{\mu_2^*}(z)+\re V(z)=\re \int_{y}^z \xi_1(s)ds +l_1,\\
& U^{\mu_3^*}(z)-U^{\mu_1^*}(z)+a\re z=\re \int_{y}^z \xi_2(s)ds +l_2, 
\\
& U^{\mu_2^*}(z)+U^{\mu_3^*}(z)+a\re z=-\re \int_{y}^z \xi_3(s)ds +l_3.
\end{align*}
\end{lem}

We can finally establish the following result, which completes the proof of Theorem~\ref{thm_existence_critical_measure}.

\begin{prop} \label{prop:constants}
The components of the vector critical measure $\vec \mu_*$ satisfy the variational identities \eqref{eq:variational_identitiesThm} from Theorem~\ref{thm_existence_critical_measure}. In particular, there exists a  real constant $\ell_3$ such that in the saturated case,
\begin{equation}\label{eq:variational_identities1}
2U^{\mu_3^*}(z)-U^{\mu_1^*}(z)+U^{\mu_2^*}(z)+\phi_3(z)=\ell_3, \quad z\in \Delta_3.
\end{equation}
Additionally,
\begin{enumerate}[i)]
\item if for any closed loop $\gamma$ on the Riemann surface $\mathcal R$, visiting the sheets $\mathcal R_1$ and $\mathcal R_2$,
\begin{equation}\label{Boutroux1}
\re \oint_\gamma \sqrt{-\varpi} =0,
\end{equation}
then there exists a  real constant $\ell_1$ such that
\begin{equation}\label{eq:variational_identities2}
2U^{\mu_1^*}(z)+U^{\mu_2^*}(z)-U^{\mu_3^*}(z)+\phi_1(z)=\ell_1, \quad z\in \Delta_1.
\end{equation}
\item if for any closed loop $\gamma$ on the Riemann surface $\mathcal R$, visiting the sheets $\mathcal R_1$ and $\mathcal R_3$,
\begin{equation}\label{Boutroux2}
\re \oint_\gamma \sqrt{-\varpi} =0,
\end{equation}
then there exists a  real constant $\ell_2$ such that
\begin{equation}\label{eq:variational_identities3}
2U^{\mu_2^*}(z)+U^{\mu_1^*}(z)+U^{\mu_3^*}(z)+\phi_2(z)=\ell_2, \quad z\in \Delta_2.
\end{equation}
\end{enumerate}
Finally, if both conditions i) and ii) are satisfied, then $\ell_1=\ell_2$. 
\end{prop}

Recall that 
$$
\phi_1(z)= \re \left( V(z)-a z\right), \quad \phi_2(z)= \re \left( V(z)+a z\right), \quad \phi_3(z)= 2a \re \left( z\right).
$$
\begin{proof}
Let $p\in \Delta_1$ be any point on $\Delta_1$. Then Lemma~\ref{lem:identity_potential_xi}, there exist real constants $l_1$, $l_2$ such that for $z\in \HH_+ \setminus   \Delta_3$,
$$
2U^{\mu_1^*}(z)+U^{\mu_2^*}-U^{\mu_3^*}(z)+\phi_1(z)=\re \int_{p}^z (\xi_1(s)-\xi_2(s))ds +l_1-l_2.
$$
The assertion that the right hand side remains constant along the connected component of $\Delta_1$ that contains $p$ is a direct consequence of \eqref{def:target_critical_measures}. Furthermore, if $q\in \Delta_1$ is a point on a different connected component of $\Delta_1$, then
$$
\left( 2U^{\mu_1^*} +U^{\mu_2^*}-U^{\mu_3^*} +\phi_1\right)(q)= \left( 2U^{\mu_1^*} +U^{\mu_2^*}-U^{\mu_3^*} +\phi_1\right)(p) + \re \int_{p}^q (\xi_1(s)-\xi_2(s))ds,
$$
where we integrate along any path in $\HH_+ \setminus   \Delta_3$. But the integral in the right hand side above can be written as in \eqref{Boutroux1}, so that condition \eqref{Boutroux1} implies equality of the left hand side in \eqref{eq:variational_identities2} along the whole set $\Delta_1$ (with the same constant). 
The remaining variational equalities follow in a similar manner. We omit the details.
\end{proof}

\begin{remark} \label{boutrouxRem}
The condition that all the periods of $\xi dz$ are purely imaginary (compare with  conditions \eqref{Boutroux1}, \eqref{Boutroux2}) is commonly called as the \textit{Boutroux condition} for the spectral curve. It has appeared several times in the past as transcendental conditions to determine the associated spectral curve of a matrix model, see for instance \cite{bertola_boutroux, kuijlaars_tovbis_supercritical_normal_matrix_model, bertola_mo, bertola_tovbis_quartic_weight}.
\end{remark}

\begin{proof}[Proof of Theorem~\ref{thm:local_behavior}]
If $x_0$ is an endpoint of $\supp\lambda$, then by Theorem~\ref{thm:branchpoints} (a) only $\xi_1$ and exactly one of $\xi_2$ or $\xi_3$ can be branched. Having in mind that this density is a multiple of $\xi_{1+}-\xi_{j+}$, the expansion \eqref{painleve_type_behavior} has to be valid.

If $x_0$ is a point in the interior of $\supp\lambda$, again by Theorem~\ref{thm:branchpoints} either none of the functions $\xi_j$ is branched, leading to \eqref{sine_like_behavior} or the three are branched at the unique point $x_*$, giving \eqref{pearcey_like_behavior}.
\end{proof}

\section{From critical measures to constrained equilibrium} \label{sec:constrained}

Our goal in this section is to compare the constrained equilibrium problem of Bleher, Delvaux and Kuijlaars \cite{bleher_delvaux_kuijlaars_external_source} with our vector critical measure $\vec\mu_*=(\mu_1^*,\mu_2^*,\mu_3^*)$.

In \cite{bleher_delvaux_kuijlaars_external_source}, the authors work under the symmetry assumptions $\alpha=1/2$ and $V(z)=V(-z)$. These in turn imply additional symmetries on the spectral curve they found. These symmetries are summarized in the next definition.

\begin{definition}\label{def:symmetric_spectral_curve}
We say a spectral curve as in Definition~\ref{definition_spectral_curve} is {\it symmetric} if 
\begin{equation}\label{eq:symmetric_spectral_curve}
F(-\xi,-z)=-F(\xi,z),\quad \xi,z\in \C,
\end{equation}
or in other words, if the coefficients $p_0,p_1$ and $p_2$ satisfy $p_k(-z)=(-1)^{k+1}p_k(z)$, $k=0,1,2$.
\end{definition}

The spectral curve found in \cite{bleher_delvaux_kuijlaars_external_source} is admissible (in the sense of Definition~\ref{definition_spectral_curve}) and symmetric; although the coefficients $p_0$ and $p_1$ are not explicit, their expressions in terms of the constrained equilibrium measure are enough to establish this fact. Furthermore, the corresponding Riemann surface $\mathcal R$ for $F(\xi,z)=0$ was also described in \cite{bleher_delvaux_kuijlaars_external_source}, and the Boutroux condition, although not explicitly mentioned there, follows easily from their construction of the spectral curve from the constrained equilibrium problem.

To be a little more precise, in the aforementioned work the authors actually describe $\mathcal R$ under the additional assumption that their constrained vector equilibrium problem is regular (not to be confused with our notion of regular saturated regime!). Their notion of regularity only affects their asymptotic analysis, and consequently their proof of convergence for the limiting eigenvalue distribution is only valid in that case. But their result on the {\it existence} of the spectral curve for a symmetric potential, and the geometry of $\mathcal R$, is valid even for their singular cases.

As an outline of this section, we will study in depth the symmetric spectral curves \eqref{eq:symmetric_spectral_curve}, extracting from the corresponding critical measures $(\mu_1^*,\mu_2^*,\mu_3^*)$, whose existence is now established, a new pair of measures $(\nu_1,\nu_2)$. This construction will be based in the geometry of the Riemann surface $\mathcal R$ already established in this paper, combined with techniques from potential theory.  At the end of the section, we then verify that if the symmetric spectral curve under consideration is the one in \cite{bleher_delvaux_kuijlaars_external_source}, then this new pair $(\nu_1,\nu_2)$ reduces to the solution of the constrained equilibrium problem they considered.

\subsection{Geometry of trajectories} \label{sec:geomTrajectories}

\

For a symmetric spectral curve, $V(z)=-V(-z)$, and from the expansion of $\xi_2$ and $\xi_3$ given in \eqref{asymptotics_xi} we thus get that  $\alpha=1/2$ in this case. 
Another immediate consequence of this definition is that the critical graph of the quadratic differential $\varpi$ is symmetric with respect to the imaginary axis, and Propositions~\ref{prop:traj_orth_saturated}, \ref{prop:construction_s_curves_saturated} and \ref{prop:construction_trajectories_s_contour_unsaturated} are supplemented by the following one (see also the notation in Definition~\ref{definition_saturated_unsaturated}):
\begin{prop}\label{prop:trajectories_symmetric_case}
Suppose the spectral curve \eqref{spectral_curve} is symmetric. Then the trajectories of $\varpi$ are symmetric with respect to the imaginary axis. In particular, $\re y_*=0$. In the saturated regime, $y_*\in i\R_+$, $x_*<0$,  $(\Delta_3 \setminus \{-y_*, y_* \}) \subset \C_-$,  and
\begin{equation} \label{vertical}
(\xi_{2}(z)-\xi_3(z))dz \begin{cases}
  \in \R   & \text{along the vertical segment } (-y_*,y_*) , \\
  \in i\R  & \text{along } i\R \setminus [-y_*,y_*].
\end{cases}
\end{equation}
Moreover, in the variational identity  \eqref{eq:variational_identities1} of Proposition \ref{prop:constants},
\begin{equation} \label{l3=0}
\ell_3 =0,
\end{equation}
whenever \eqref{Boutroux1} and \eqref{Boutroux2} are met. 
\end{prop}

\begin{figure}[t]
	\begin{tikzpicture}[scale=0.8]
	\draw [line width=0.8mm,lightgray] (-5.5,0)--(5.5,0) node [pos=0.98,above,black] {$\R$};
		\draw [line width=0.4mm,lightgray] (0,0)--(0,4) node [pos=0.98,above,black] {$i\R$};
	\draw [blue,line width=0.6mm] (-4.5,3) to [out=0,in=130,edge node={node [pos=0.3,above] {$\tau_2$}}] (0,0);
		\draw [blue,line width=0.6mm] (4.5,3) to [out=180,in=180-130,edge node={node [pos=0.3,above] {$\tau_2$}}] (0,0);
	%
	%
	%
	\node at (4,1) {$\mathcal G_1$};
	\node at (-3,1.5) {$\mathcal G_2$};
	\fill (0,0) circle[radius=4pt];
	
	\fill (0,0) circle[radius=4pt] node [below,shift={(0pt,-1pt)}] (cstar) {$c_*$}; 
	
	\end{tikzpicture}
	\caption{Symmetric unsaturated regime: The orthogonal trajectories that determine the half plane domain $\mathcal H$. This is an update version, taking into account symmetry, of previous figures, compare for instance with Figure~\ref{fig:sets_G_unsaturated}.}
	\label{fig:unsaturated_symmetric}
\end{figure}
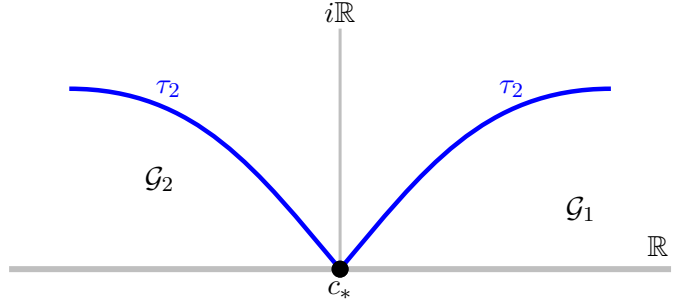

	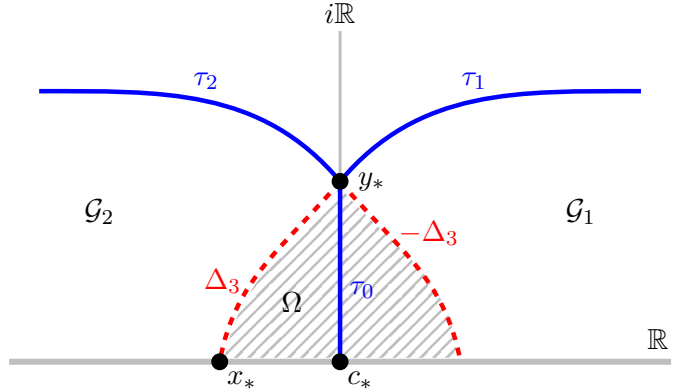
\begin{figure}[t]
	\centering
		\begin{tikzpicture}[scale=0.8]
		%
		   \draw [thick,white,line width=3pt,pattern=north east lines wide,pattern color=lightgray] (0,3) to [out=50-180,in=80] (-2,0)
		to [out=0,in=180] (2,0)
		to [out=180-90,in=-50] (0,3);
		\draw [line width=0.8mm,lightgray] (-5.5,0)--(5.5,0) node [pos=0.98,above,black] {$\R$};
		\draw [line width=0.4mm,lightgray] (0,3) to[out=90,in=270,edge node={node [pos=0.98,above, black] {$i\R$}}] (0,5.5); 
		%
		\draw [line width=0.6mm,blue] (5,4.5) to[out=180,in=50,edge node={node [pos=0.5,above] {$\tau_1$}}] (0,3); 
				\draw [line width=0.6mm,blue] (-5,4.5) to[out=0,in=180-50,edge node={node [pos=0.5,above] {$\tau_2$}}] (0,3); 
		\draw [line width=0.6mm,dashed,red] (0,3) to[out=50-180,in=80,edge node={node [pos=0.6,left] {$\Delta_3$}}] (-2,0);
		\draw [line width=0.6mm,dashed,red] (0,3) to[out=-50,in=180-80,edge node={node [pos=0.3,right] {$-\Delta_3$}}] (2,0);
			\draw [line width=0.6mm,blue] (0,3) to[edge node={node [pos=0.6,right,shift={(-1pt,0pt)}] {$\tau_0$}}] (0,0); 
		\node at (-0.8,1) {$\Omega$};
		\fill (-2,0) circle[radius=4pt] node [below right,shift={(-1pt,0pt)}] {$x_*$};
	 	\fill (0,0) circle[radius=4pt] node [below right,shift={(-1pt,0pt)}] {$c_*$};
		%
		%
		%
		\node at (4,2.5) {$\mathcal G_1$};
		\node at (-4,2.5) {$\mathcal G_2$};
		%
		%
		\fill (0,3) circle[radius=4pt] node [right,shift={(3pt,0pt)}] (cstar) {$y_*$};
		\end{tikzpicture}
		\caption{Symmetric saturated regime:  The orthogonal trajectories on $\HH_+$ that emanate from $y_*^{(1)}$ and the trajectory that we defined to be $\Delta_3$. For convenience, we also added the set $-\Delta_3$. The domain $\Omega$ is the whole shaded region, so in particular $\tau_0\subset\Omega$. This is an update version, taking into account symmetry, of previous figures, compare for instance with Figures~\ref{fig:orthogonal_half_plane}, \ref{fig:gamma_sigma_vs_traj} and \ref{fig:sets_G}. }
\label{fig:saturated_symmetric}
\end{figure}

\begin{proof}
We invite the reader to take a look at Figures~\ref{fig:unsaturated_symmetric} and \ref{fig:saturated_symmetric}, where the trajectories and orthogonal trajectories of relevance are updated for the symmetric situation considered here. The proof of almost all assertions is straightforward; we only establish \eqref{l3=0}. In what follows, it is convenient to define 
$$
y_*=0
$$
in the unsaturated regime, in order to handle both cases simultaneously.  

By Lemma~\ref{lem:identity_potential_xi}, for $z\in \mathbb H_+ \setminus \Delta_3$, 
$$
2U^{\mu_3^*}(z)-U^{\mu_1^*}(z)+U^{\mu_2^*}(z)+2 a \re z=\ell_3  + \re \int_{y_*}^z \left( \xi_2(s) -  \xi_3(s) \right)\, ds .
$$
Taking into account \eqref{vertical},  
\begin{equation} \label{l3}
2U^{\mu_3^*}(z)-U^{\mu_1^*}(z)+U^{\mu_2^*}(z) =\ell_3    , \quad z\in i\R_+ \setminus [0,y_*].
\end{equation}
Recall that for a (signed) compactly supported Borel measure $\mu$ on $\C$,
\begin{equation} \label{asympPotential}
U^\mu(z)= -|\mu| \log  |z|+o(1), \quad z \to \infty.
\end{equation}
Notice that for $\alpha=1/2$, $(2 \mu_3^* - \mu_1^* + \mu_2^*)$ is a neutral measure. Thus, making $z\to +i \infty$ in \eqref{l3} we get \eqref{l3=0}. 
\end{proof}

For a symmetric spectral curve,  $\xi(z)$ is a solution  if and only if  $-\xi(-z)$ is a solution also. Taking into account the branch cuts for $\xi_1$ we have that
\begin{equation} \label{symetryfor1}
\xi_1(-z) = - \xi_1(z), \quad z\in \C\setminus\supp \lambda.
\end{equation}
In the unsaturated case, when $[a_k, b_k]$, $k=1, \dots, l$, are the only branch cuts for $\xi_2$ and $\xi_3$, we also have
\begin{equation} \label{symetryfor23}
\xi_2(-z) = - \xi_3(z), \quad z\in \C\setminus\supp \lambda.
\end{equation}

This implies, in particular, that for the unsaturated regime the picture is totally symmetric: $x\in \Delta_1$ if and only if $-x\in \Delta_2$, so that in particular, $|\mu_1|=|\mu_2|=1/2$. In fact, as in this case $\partial \mathcal H$ has to be symmetric w.r.t. the imaginary axis, its unique boundary critical point $c_*=x_*$ has to be at the origin. This is the point of separating $\Delta_1$ and $\Delta_2$ in Proposition~\ref{prop:structureDelta}, so
$$
\Delta_1 \subset \C_+, \quad \Delta_2 \subset \C_-.
$$

In the saturated regime we have two additional branch points, $\pm y_*\in i\R$, and an additional branch cut $\Delta_3$ for $\xi_2-\xi_3$. This means that \eqref{symetryfor23} is now valid only outside of the bounded domain, say $\Omega$ (see Figure~\ref{fig:saturated_symmetric}), delimited by $\Delta_3$ and $-\Delta_3=\{-z:\, z\in \Delta_3\}$. Inside such a domain, we have
$$
\xi_2(-z) = - \xi_2(z), \quad \xi_3(-z) = - \xi_3(z), \quad z\in \C\setminus\supp \lambda.
$$
As a consequence, we do have the symmetry outside $\Omega$: $x\in \Delta_1\setminus \Omega$ if and only if $-x\in \Delta_2\setminus \Omega$. But the asymmetry of $\Delta_3$ implies that now
$$
\Delta_2\cap \C_+=\emptyset \quad \mbox{and} \quad \Delta_1\cap \C_-\subset \Omega.
$$
In addition, certainly $\Delta_1\cap \C_- \neq\emptyset$, otherwise the union of trajectories $\Delta_3\cup (-\Delta_3)$ would encircle a bounded domain, symmetric for $\R$, that violates the Principle {\bf P3}.

On a connected component $[a_k, b_k]$ of $\supp \lambda$ that intersects $\Delta_3$ we have that it carries $\mu_2$ on the left of the intersection point $x_*$, and $\mu_1$ on the right.

\subsection{Construction of the constrained measure}\label{sec:constrained_measure}
 
 \
 
In the saturated regime,   $\Delta_3\cup [-y_*,y_*]$ enclose a bounded domain that we call $\mathcal H_3$. In the unsaturated regime, we set $\mathcal H_3:=\emptyset$. Let also
$$
\mathcal H_1:=\C_-\setminus \overline{\mathcal H_3},\quad \mathcal H_2:=\C_+,
$$
see the notation in \eqref{notationCH}. 
In the saturated regime, the sets $\mathcal H_1,\mathcal H_2$ and $\mathcal H_3$ are shown in Figure~\ref{fig:domains_H}.

\begin{figure}[t]
\begin{tikzpicture}[scale=0.6]
 \draw [line width=0.6mm,dash pattern={on 2pt off 1pt on 6pt off 2pt},red] (0,3) to[out=50-180,in=80,edge node={node [left, shift={(-1pt,-1pt)}] {$\Delta_3$}}] (-2,0);
\draw [line width=0.6mm,lightgray] (0,5) to (0,3);  
\draw [line width=0.5mm] (-4,0) to (3,0) node [right,shift={(1pt,-1pt)}] (cstar) {$\R$}; 
\draw [line width=0.6mm,orange,dash pattern={on 6pt off 6pt on 6pt off 6pt}] (0,3) to (0,-3); 
\begin{scope}[yscale=-1]
\draw [line width=0.6mm,dash pattern={on 2pt off 1pt on 6pt off 2pt},red] (0,3) to[out=50-180,in=80] (-2,0);
\draw [line width=0.6mm,lightgray] (0,5) to (0,3);  
\fill (0,3) circle[radius=4pt] node [right,shift={(1pt,-1pt)}] (cstar) {$-y_*$}; 
\end{scope}
\fill (0,3) circle[radius=4pt] node [right,shift={(1pt,-1pt)}] (cstar) {$y_*$};
 \fill (-2,0) circle[radius=4pt] node [left,shift={(1pt,5pt)}] {$x_*$};
\node at (-0.9,0.5) {$\mathcal H_3$};
\node [above] at (2.5,1.5) {$\mathcal H_2$};
\node [above] at (-3.7,1.5) {$\mathcal H_1$};

\end{tikzpicture}
\caption{Saturated regime: the partition of the plane into the domains $\mathcal H_1$, $\mathcal H_2$ and $\mathcal H_3$ used in the definition of the function $H$ in \eqref{def:function_H}.} 
\label{fig:domains_H}
\end{figure}
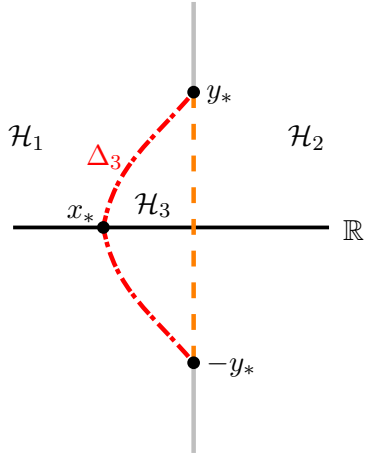

With
\begin{equation}\label{definition_functions_H}
\begin{aligned}
H_1(z) & := U^{\mu_1^*}(z)-U^{\mu_3^*}(z) , \\
H_2(z) & := U^{\mu_2^*}(z)+U^{\mu_3^*}(z),   \\
H_3(z) & := U^{\mu_2^*}(z)+U^{\mu_3^*}(z)+2 a \re z ,
\end{aligned}
\end{equation}
(in the unsaturated regime, only $H_1$ and $H_2$ are relevant),  let $  H:  \mathcal H_1\cup \mathcal H_2 \cup \mathcal H_3 \to \R$ be defined by
\begin{equation}\label{def:function_H}
\restr{H}{\mathcal H_j}:=H_j, \quad j=1, 2, 3.
\end{equation}
By the discussion above, $\Delta_1\subset \overline{\mathcal H_2\cup \mathcal H_3}$, $\Delta_2\subset \overline{\mathcal H_1}$,   so that $H$ is continuous and harmonic in $\mathcal H_1\cup \mathcal H_2 \cup \mathcal H_3$. 

\begin{lem}\label{lem:H_continuous}
The function $H$ extends continuously to $\C$.
\end{lem}
\begin{proof}
Observe that
\begin{equation}\label{eq:difference_H3_H2}
H_3(z)-H_2(z) =2a\re z=\phi_3(z),
\end{equation}
and $\re z=0$ along $[-y_*,y_*]$, so that $H$ is continuous across this segment. Furthermore, 
$$
H_3(z)-H_1(z)=2U^{\mu_3^*}(z)+U^{\mu_2^*}(z)-U^{\mu_1^*}(z)+\phi_3(z) ,
$$ 
and by \eqref{l3=0} and Proposition~\ref{prop:constants} this is zero along $\Delta_3$:  $H$ is continuous across $\Delta_3$ as well.

Finally, for $z\in i\R\setminus [-y_*,y_*]$,
\begin{equation}\label{eq:difference_H1_H2_2}
\begin{aligned}
H_2(z)-H_1(z) & =2U^{\mu_3^*}(z)+U^{\mu_2^*}(z)-U^{\mu_1^*}(z)  \\
& =2U^{\mu_3^*}(z)+U^{\mu_2^*}(z)-U^{\mu_1^*}(z)+\phi_3(z) =0,
\end{aligned}
\end{equation}
where we have used again   \eqref{l3=0} and Proposition~\ref{prop:constants}. Continuity along   $i\R\setminus [-y_*,y_*]$ is established also. 
\end{proof}

\begin{lem}\label{lem:harmonic_delta3}
Suppose that we are in the saturated regime. Then the function $H$ admits a  harmonic continuation through $\Delta_3$.
\end{lem}
\begin{proof}
	According to Lemma~\ref{lem:H_continuous}, $H$ is continuous across $\Delta_3$, so it is enough to verify that
	$$
	\frac{\partial H}{\partial n_+}(z)=-\frac{\partial H}{\partial n_-}(z), \quad z \in \Delta_3\setminus \{  \pm y_*  \},
	$$
	for $z$ along $\Delta_3$, where $n_\pm$ are the normal vectors to $\Delta_3$ at $z$. This identity follows immediately from the explicit expression for $H$ in $\mathcal H_1$ and $\mathcal H_3$, combined with the $S$-property \eqref{vector_s_property}. 
\end{proof}

Our final ingredient is the following
\begin{lem}\label{lemma_minimum_potential_cut_2}
	There exists a neighborhood $\Omega_2$ of $i\R\setminus [-y_*,y_*]$ such that
\begin{equation} \label{min1}
	\min(H_1(z),H_2(z))=H_j(z),\quad z\in \Omega_2 \cap \mathcal H_j,\quad j=1,2.
\end{equation}
	Moreover, in the saturated regime there exists a neighborhood $\Omega_1$ of $(-y_*,y_*)$ such that
\begin{equation} \label{min2}
	\min(H_2(z),H_3(z))=H_j(z),\quad z\in \Omega_1\cap \mathcal H_j,\quad j=2,3.
\end{equation}
\end{lem}
\begin{proof}
	Clearly, \eqref{min2} follows immediately from \eqref{eq:difference_H3_H2}. So, let us prove \eqref{min1}.
	
	Recalling \eqref{definition_functions_H} and that in the present case $\ell_3=0$, we proceed as  before (see Lemma~\ref{lem:identity_potential_xi}) and get that for $z$ in a neighborhood of $i\R\setminus [-y_*,y_*]$,
	\begin{align}
	H_2(z)-H_1(z)&  = 2U^{\mu_3^*}(z)-U^{\mu_1^*}(z)+U^{\mu_2^*}(z) =\re\int_{y_*}^z(\xi_2(s)-\xi_3(s))ds-\phi_3(z)  \nonumber \\
&	 =\re\int_{y_*}^z(\xi_2(s)-\xi_3(s)-2a)ds, \label{aux_equation_7}
	\end{align}
	where in the last step we used the explicit form of $\phi_3$. This shows that in order to conclude the proof we need to analyze the structure of trajectories on $\mathcal R_1$ of the following auxiliary quadratic differential, defined by
	\begin{equation}\label{def:modified_qd}
	\widetilde \varpi= 
	\begin{cases}
	-(\xi_2(z)-\xi_3(z)-2a)^2dz^2, & \mbox{ on } \mathcal R_1, \\
	-(\xi_3(z)-\xi_1(z)-2a)^2dz^2, & \mbox{ on } \mathcal R_2, \\
	-(\xi_1(z)-\xi_2(z)-2a)^2dz^2, & \mbox{ on } \mathcal R_3.
	\end{cases}
	\end{equation}
	From the expansion \eqref{asymptotics_xi} it follows that $\widetilde \varpi$ is either regular or has a zero at $\infty^{(1)}$, whose order depends on the first nonzero term in the error term in \eqref{asymptotics_xi}. Furthermore, since $a \in \R$, arcs of orthogonal trajectories of $\varpi$ and $\widetilde \varpi$ on $\R$ coincide, as well as arcs of trajectories on $i\R$. 

We claim that $\widetilde \varpi$ does not vanish on $\HH_+^{(1)}$. The proof follows the ideas already  used to establish   Proposition~\ref{prop_unique_zero}. Suppose that there is a zero, say $p^{(1)}\in \HH_+^{(1)}$, of $\widetilde \varpi$. Emanating from it, there are at least three distinct orthogonal trajectories, and because $\infty^{(1)}$ is either a regular point or a zero of $\varpi$, at most one of these orthogonal trajectories diverges to $\infty^{(1)}$, in the asymptotically vertical direction (see Principle \textbf{P5} in Section~\ref{sec:trajandorthtraj}).
	
This means that there is a path $\gamma \subset \HH_+^{(1)}$, comprised of a finite union of orthogonal critical trajectories, with endpoints $c$ and $d$ on the real axis; without loss of generality, we assume that there is no other zero of $\widetilde \varpi$ in $(c,d)^{(1)}$. If we assume that $(c,d)\cap \supp \lambda=\emptyset$, we obtain that the union of $\gamma$ and of its reflection with respect to $\R$ enclose a bounded domain on $\mathcal R_1$, contradicting Principle \textbf{P.3}. Thus, the other option to discard is when $[c,d]\subset  \supp\lambda$. Since necessarily 
	$$
	\int_c^{d}(\xi_{1-}(x)-\xi_{1+}(x)-2a)dx\in \R,
	$$
and because $a$ is real, we get that
	$$
	\int_c^{d}(\xi_{1-}(x)-\xi_{1+}(x))dx\in \R,
	$$
	which leads to a contradiction with  \eqref{Sokhotsky-Plemelj}.

The just established fact implies that 
\begin{equation} \label{gNegative}
h(z):= \im \int_{y_*}^{z} (\xi_2(s)-\xi_3(s)-2a)ds <0, \quad z \in i\R_+\setminus [0, y_*].
\end{equation}
Indeed, notice that on $i \R_+$ we could have defined equivalently
$$
h(z)=\frac{1}{i} \int_{y_*}^{z} (\xi_2(s)-\xi_3(s)-2a)ds.
$$

By definition $h(y_*)=0$, and for $x>0$
$$
  h'(ix)=-2ax+ (\xi_2-\xi_3)(ix).
$$
Since $(\xi_2-\xi_3)(y_*)=0$ (see Propositions~\ref{prop:construction_s_curves_saturated} and \ref{prop:construction_trajectories_s_contour_unsaturated}), we have that $h'(y_*)<0$, so that $h$ is strictly decreasing on $i\R_+\setminus [0, y_*]$ in a neighborhood of $y_*$. Since $\widetilde \varpi$ does not vanish on $\HH_+^{(1)}$, this means that $h'(y)\neq 0$ there, thus proving  \eqref{gNegative}.
	
We can return now to \eqref{aux_equation_7}. The real part of the integral in the right hand side is $0$ along  $[y_*,i\infty)$, while the imaginary part is $h(z)$, and we just proved that it is strictly negative there. Since conformal maps preserve orientation, we conclude that $H_2(z)-H_1(z)>0$ for $z$ immediately to the left of $[y_*,i\infty)$ and $H_2(z)-H_1(z)<0$ immediately to the right, concluding the proof of \eqref{min1}.
\end{proof}

Our main conclusion is that $H$ is actually superharmonic on $\C$, so up to an additive harmonic term, it is a logarithmic potential of a positive measure. We can say more:

\begin{thm}\label{prop_definition_lambda_2}
 There exists a positive measure $\nu_2$ with $|\nu_2|=1/2$, $\supp\nu_2\subset i\R$, for which
 $$
 U^{\nu_2}(z)=H(z),\quad z\in \C.
 $$
The measure $\nu_2$ is absolutely continuous with respect to the arc length on $i\R$, has continuous density and 
\begin{equation}\label{inequality_lambda2_sigma}
\begin{aligned}
\nu_2= \sigma, & \quad \mbox{ on } [-y_*,y_*], \\
\nu_2< \sigma, & \quad \mbox{ on } i\R\setminus [-y_*,y_*],
\end{aligned}
\end{equation}
where $\sigma$ is given in \eqref{def:constrained_measure}. 
\end{thm}
\begin{proof}
	A combination of Lemmas \ref{lem:H_continuous} and \ref{lem:harmonic_delta3} tells us that $H$ is harmonic in $\C\setminus i\R$, and by 
Lemma~\ref{lemma_minimum_potential_cut_2}, it is   superharmonic in a neighborhood of $i\R$.

From the Riesz Decomposition Theorem \cite[Theorem~II.3.1, page~100]{Saff_book} there exist a measure $\nu_2$ and a harmonic function $u:\C\to\R$ such that
$$
H(z)=u(z)+U^{\nu_2}(z),\quad z\in \C, \quad  \supp\nu_2\subset i\R.
$$
From the definition of $H$, the relations between the masses of $\mu_1^*$, $\mu_2^*$ and $\mu_3^*$ and from \eqref{asympPotential} it follows that
 $$
 u(z)=\left(\frac{1}{2}-|\nu_2|\right)\log\frac{1}{|z|}+o(1), \quad z\to \infty. 
 $$
Since $u$ is harmonic on $\C$, from the maximum principle it follows that $|\nu_2|=1/2$ and $u\equiv 0$.  

The partial derivatives of the potential $H=U^{\nu_2}$ are continuous up to the boundary of each of the domains $\mathcal H_1$, $\mathcal H_2$ and $\mathcal H_2$, so from \cite[Thm. 1.5, p. 92]{Saff_book} we get that $\nu_2$ is absolutely continuous with continuous density.

For \eqref{inequality_lambda2_sigma} we use the definition of the functions $H_j$'s in \eqref{definition_functions_H} and also \eqref{def:function_H} to express
\begin{align*}
\frac{\partial U^{\nu_2}_+}{\partial z}(z)-\frac{\partial U^{\nu_2}_-}{\partial z}(z) & = \frac{\partial H_+}{\partial z}(z) -\frac{\partial H_-}{\partial z}(z) \\  & =
\begin{dcases}
2a \frac{\partial }{\partial z}\re z,& \quad z\in (-y_*,y_*), \\
-\frac{\partial}{\partial z}\left( U^{\mu_3^*}(z)+U^{\mu_2^*}(z) -U^{\mu_1^*}(z)\right),& \quad z\in i\R\setminus [-y_*,y_*],
\end{dcases}
\end{align*}
where we oriented $i\R$ from $-i\infty$ to $i\infty$. In virtue of Sokhotsky-Plemelj's formula and \eqref{relation_cauchy_transform_potential},
\begin{equation*}
\frac{d\nu_2}{ds}(z) = \frac{1}{2\pi i}(C^{\nu_2}_+(z)-C^{\nu_2}_-(z)) 
						= \frac{1}{\pi i} \frac{\partial}{\partial z}(U^{\nu_2}_+(z)-U^{\nu_2}_-(z)) 
						= 	\dfrac{a}{\pi}, \quad z\in (-y_*,y_*).
\end{equation*}
In particular, this leads to the first equation in \eqref{inequality_lambda2_sigma}. 

Similarly,
\begin{align*}
\frac{d\nu_2}{ds}(z) & =\dfrac{1}{2\pi i } ( -2C^{\mu_3*}(z) - C^{\mu_2^*}(z)+C^{\mu_1^*}(z))\\
						& = -\frac{1}{2\pi i}\left(\xi_2(z)-\xi_3(z)-2a\right),\quad z\in i\R \setminus [-y_*,y_*].
\end{align*}
which can be rearranged to
$$
\frac{d\sigma}{ds}(z)-\frac{d\nu_2}{ds}(z)=\frac{1}{2\pi i}\left(\xi_2(z)-\xi_3(z)\right), \quad z\in i\R \setminus [-y_*,y_*].
$$
In virtue of Proposition~\ref{prop_unique_zero}, this identity means that the density of the difference $\sigma-\nu_2$ does not vanish on $i\R\setminus [-y_*,y_*]$. Also, because $\nu_2$ is finite, its density has to vanish as $z\to \pm i\infty$. This is enough to conclude that the density of $\sigma-\nu_2$ is always positive on  $i\R\setminus [-y_*,y_*]$, giving us the inequality in \eqref{inequality_lambda2_sigma} and concluding the proof.
\end{proof}

As we mentioned, in \cite{bleher_delvaux_kuijlaars_external_source} the authors proved the existence of a symmetric spectral curve for any even potential $V$ and $\alpha=1/2$. With $\vec{\nu}^*=(\nu_1^*,\nu_2^*)$ being the constrained equilibrium measure discussed at the beginning of \eqref{sec:totallySym}, the measure $\lambda$ in \eqref{solution_cauchy_transform} is 
\begin{equation}\label{eq:lambda_symmetric}
\lambda=\nu_1^*.
\end{equation}
In addition, their spectral curve also satisfies the Boutroux conditions \eqref{Boutroux1} and \eqref{Boutroux2}, and the variational conditions for the constrained equilibrium problem say that $\nu=\nu_2^*$ is the unique measure supported on $i\R$ with $|\nu|=1/2$, $\nu\leq \sigma$, and for which
\begin{equation}\label{eq:variational_cond_constrained}
2U^{\nu}(z)-U^{\nu_1^*}(z)=0,\quad z\in \supp(\sigma-\nu),\qquad 2U^{\nu}(z)-U^{\nu_1^*}(z)\leq 0,\quad z\in i\R.
\end{equation}

We now compare their constrained equilibrium measure $\vec \nu^*$ with our construction in the symmetric case, that is, we now prove Theorem~\ref{thm_s_property_constrained_problem}.

\begin{proof}[Proof of Theorem~\ref{thm_s_property_constrained_problem}]
So from now on, let us assume we are working with the symmetric spectral curve given in \cite{bleher_delvaux_kuijlaars_external_source}, denoting their constrained equilibrium measure by $\vec\nu^*=(\nu_1^*,\nu_2^*)$, and our pair of measures by $\nu_1=\mu_1^*+\mu_2^*=\lambda$ and $\nu_2$, the latter  given in Theorem~\ref{prop_definition_lambda_2}. From \eqref{eq:lambda_symmetric}, we immediately get
$$
\nu_1^*=\lambda=\mu_1^*+\mu_2^*=\nu_1,
$$
and our goal is thus to show that $\nu_2^*=\nu_2$.

To do so, let us start by summarizing our findings in the symmetric situation so far: we have a measure $\nu_2$ on $i\R$, satisfying the constraint \eqref{inequality_lambda2_sigma}, and such that by its construction and \eqref{definition_functions_H}, 
\begin{equation}\label{propertiesH}
\begin{aligned}
U^{\nu_2}(z) & = U^{\mu_1^*}(z)-U^{\mu_3^*}(z) , \quad z\in \mathcal H_1, \\
U^{\nu_2}(z) & = U^{\mu_2^*}(z)+U^{\mu_3^*}(z) , \quad z\in \mathcal H_2,   \\
U^{\nu_2}(z) & = U^{\mu_2^*}(z)+U^{\mu_3*}(z)+2 a \re z , \quad z\in \mathcal H_3.
\end{aligned}
\end{equation}
In particular, recalling that $\Delta_2\subset \C_-$, variational identities from Propositions~\ref{prop:constants} and \ref{prop:trajectories_symmetric_case} imply that
$$
\begin{aligned}
	& 2U^{\mu_1^*+\mu_2^*}(x)-U^{\nu_2}(x)+V(x)-a \re x=\ell_1, \quad x\in \Delta_1\cap \C_+, \\
	& 2U^{\mu_1^*+\mu_2^*}(x)-U^{\nu_2}(x)+V(x)+a \re x  =\ell_2, \quad x\in \Delta_2\subset \C_-, \\
	& 2U^{\mu_3^*}(z)-U^{\mu_1^*}(z)+U^{\mu_2^*}(z)+\phi_3(z)=0, \quad z\in \Delta_3. \\
\end{aligned} 
$$
In principle, the constants above may depend on the connected components of the sets $\Delta_1,\Delta_2$ and $\Delta_3$. However, as we are now working under the assumption of the symmetric spectral curve in \cite{bleher_delvaux_kuijlaars_external_source}, \eqref{Boutroux1} and \eqref{Boutroux2} hold true and the constant $\ell_j$ is the same for any connected component of $\Delta_j$.

As for $ x\in \Delta_1\cap \C_-$, we have that
$$
2U^{\mu_1^*+\mu_2^*}(x)-U^{\mu_2^*+\mu_3^*}(x)+V(x)-a \re x=\ell_1,
$$
and by the third identity in \eqref{propertiesH}, it can be written as 
$$
2U^{\mu_1^*+\mu_2^*}(x)- U^{\nu_2}(x)+V(x)+a \re x  =\ell_1,
$$
which shows that $\ell_1=\ell_2=:\ell$. Recalling that $\nu_1^*=\mu_1^*+\mu_2^*$, we summarize these identities as
$$
2U^{\nu_1^*}(x)-U^{\nu_2}(x)+V(x)-a |\re x|=\ell, \quad x\in \supp \lambda=\Delta_1 \cap \Delta_2.
$$
On the other hand, adding the first two equalities in \eqref{propertiesH} we get that 
\begin{equation}\label{constraint_variational1}
2U^{\nu_2}(z)= U^{\nu_1^*}(z), \quad z\in i\R\setminus[-y_*, y_*],
\end{equation}
and 
$$
2U^{\nu_2}(z)= 2U^{\nu_1^*}(z),  \quad z\in [-y_*, y_*].
$$
By symmetry, $U^{\nu_1^*}(0)=0$. Also, having in mind that $\supp\nu_1^*\subset \R$,
$$
\frac{d}{dy}\left(2U^{\nu_1^*}(iy)\right)=\frac{d}{dy}\int \log\frac{1}{x^2+y^2}\, d\nu_1^*(x)=-2y\int \frac{1}{x^2+y^2}\, d\nu_1^*(x),
$$
showing that $U^{\nu_1^*}(z)$ is strictly decreasing as $z$ traverses $i\R_+$ upwards. Thus, $U^{\nu_1^*}(z)<0$ for $z\in i\R\setminus \{0\}$, and 
\begin{equation}\label{constraint_variational2}
2U^{\nu_2}(z)= 2U^{\nu_1^*}(z) \leq U^{\nu_1^*}(z),  \quad z\in [-y_*, y_*].
\end{equation}

Equations~\eqref{constraint_variational1}--\eqref{constraint_variational2} show that $\nu_2$ satisfies \eqref{eq:variational_cond_constrained}, and consequently $\nu_2=\nu_2^*$, finishing the proof of Theorem~\ref{thm_s_property_constrained_problem}. 
\end{proof}

\section{Multiple orthogonal polynomials}  \label{sec5}

From this section on, we will also tacitly assume that $V$ is a polynomial of even degree and positive leading coefficient, so that the measures $e^{-N(V(x)\pm a x)}dx$ have all moments finite, and the matrix model \eqref{external_source_model} is well defined.

As it was mentioned in Section~\ref{sec:general}, we can embed the  average characteristic polynomial 
$ \pi_{n_1, n_2}=\pi_{\vec n}$, defined in \eqref{averageCharPoly}, in a family of multiple orthogonal polynomials (MOP) $P_{\vec k}^{(N)}$, satisfying \eqref{defPNmultiple}, in such a way that $\pi_{\vec n}=P^{(N)}_{\vec n}$, $N=n_1+n_2$. However, in order to prove Theorem~\ref{thm:as_convergence} we will need to treat the indices $N$ and $\vec k$ as independent parameters.

\subsection{Zeros of MOP} \label{sec:zerosMOP}

\

It is well known that the zeros of standard orthogonal polynomials on $\R$ are real and simple, and that the zeros of two consecutive polynomials  interlace. Such a result  for multiple orthogonal polynomials is not immediate; however, in our case it still holds, and we prove it taking advantage of a certain feature of our weights. 

Namely, a system of functions $\{w_1,\hdots, w_s\}$ forms an \textbf{AT system} on a bounded or unbounded interval or the real axis if any linear combination $a_1w_1+\hdots+a_sw_s$ has at most $s-1$ zeros on this interval, counted with account of multiplicity (see e.g.~\cite[Chapter 4]{nikishin_sorokin_book}). 
\begin{prop}\label{prop:rAT}
	For any non-negative integers $k_1, k_2$ the system
	$$
\left\{e^{Na x},\hdots, x^{k_1-1}e^{Na x},e^{-Na x},\hdots, x^{k_2-1}e^{-Na x} \right\}
	$$
	is an AT-system on $\R$.
\end{prop}
\begin{proof} Suppose to the contrary that a function of the form
$$
f(x)=p_{1}(x)e^{Nax}+p_2(x)e^{-Nax},
$$
where $p_1$ and $p_2$ are polynomials of degree up to $k_1-1$ and $k_2-1$, respectively, has at least $k_1+k_2$ zeros on $\R$. Then the same is true for the function
$$
g(x)= p_1(x)+p_2(x)e^{-2Nax},
$$
and thus, the $k_1$-th derivative $g^{(k_1)}$ of $g$ has to have at least $k_2$ zeros. However, $g^{(k_1)}(x)=e^{-2Nax}\times$(polynomial degree $\leq k_2-1$), which trivially can have at most $k_2-1$ zeros.
\end{proof}

The previous proposition establishes that 
	$$
\left\{  e^{-NV_1(x)},\hdots, x^{k_1-1}e^{-NV_1(x)},e^{-NV_2(x)},\hdots, x^{k_2-1}e^{-NV_2(x)}  \right\}
$$
is an AT-system on $\R$. By \cite[Theorem~2.1]{haneczok_vanassche_2012}, the following is true:
\begin{prop}\label{prop:realzeros}
	For any index $\vec k$, the zeros of $P^{(N)}_{\vec k}$ are real and simple; furthermore, for $j=1,2$, the zeros of $P^{(N)}_{\vec k}$ and $P^{(N)}_{\vec k-\vec e_j}$ interlace.
\end{prop}
%

Interlacing has straightforward consequences on the limits of the zero-counting measures $\mu(P^{(N)}_{\vec n})$ (recall the definition in \eqref{counting_measure_definition}):
\begin{prop}
	Suppose that for some path $(\vec n_j)$ and a sequence of integers $(N_j)$ we have the convergence
	$$
	\mu(P^{(N_j)}_{\vec n_j})\stackrel{*}{\to} \lambda \quad \mbox{as }j \to \infty.
	$$
	for some compactly supported measure $\lambda$. Then for any fixed vector $\vec v\in \Z^2$  we have the convergence
	$$
	\mu(P^{(N_j)}_{\vec n_j+\vec v}) \stackrel{*}{\to} \lambda \quad \mbox{as }j \to \infty,
	$$
	for the same measure $\lambda$.
\end{prop}
\begin{proof}
	Clearly, it is sufficient to prove the assertion for $P^{(N_j)}_{\vec n_j\pm \vec e_k}$, $k=1,2$. By the interlacing property, for any interval $[a,b]\subset \R$,
	$$
	\left| \int_a^b  d \mu(P^{(N_j)}_{\vec n_j})   - \frac{ |\vec n_j\pm \vec e_k|}{|\vec n_j| }\, \int_a^b  d \mu(P^{(N_j)}_{\vec n_j\pm \vec e_k})   \right| \le \frac{2}{|\vec n_j| },
	$$
	and the result follows. 
\end{proof}

\subsection{Recurrence coefficients of MOP}\label{sec:recurrMOP}

\

Closely related to the multiple orthogonal polynomials  $(P^{(N)}_{\vec k})$ in \eqref{defPNmultiple}   is the sequence of associated biorthogonal functions $(Q_{\vec k}^{(N)})$, also known as \textit{type I multiple orthogonal polynomials}, of the form
\begin{equation}\label{secondkindQ}
Q_{\vec k}^{(N)}(x)=A^{(N,1)}_{\vec k}(x)e^{-NV_1(x)}+A^{(N,2)}_{\vec k}(x)e^{-NV_2(x)},
\end{equation}
where $A_{\vec k}^{(N,1)}$ and $A_{\vec k}^{(N,2)}$ are polynomials with 
\begin{equation}\label{eq:normalization_typeI}
A_{\vec k}^{(N,j)}(x) =\gamma_{\vec k}^{(N,j)}x^{k_j-1}+\mbox{ lower order terms},
\end{equation}
uniquely determined by the orthogonality conditions
$$
\int x^k Q_{\vec j}^{(N)}(x)dx=
\begin{cases}
0, & \mbox{if } k\le |\vec j|-2, \\
1, & \mbox{if } k= |\vec j|-1,
\end{cases}
$$
(see e.g.~\cite[Section 23.1.3]{ismail_book} or \cite{vanassche_recurrence_mop}). Taking into account \eqref{defPNmultiple}  we see that $(P^{(N)}_{\vec k})$ and $(Q_{\vec k}^{(N)})$  satisfy the 
biorthogonality relations
\begin{equation}\label{eq:biothogonality}
\int P_{\vec k}^{(N)}(x) Q_{\vec j}^{(N)}(x)dx=
\begin{cases}
1, & \mbox{if } |\vec k|=|\vec j|-1, \\
0, & \mbox{if } j_1\leq k_1 \mbox{ and } j_2\leq k_2, \mbox{ or } |\vec k|\leq |\vec j| -2,
\end{cases}
\end{equation} 
and also that
\begin{equation}\label{eq:norming_constants}
\frac{1}{\gamma^{(N,j)}_{\vec k}}=\int x^{k_j-1}P_{\vec k-\vec e_j}^{(N)}(x)e^{-NV_j(x)}dx, \quad j=1,2.
\end{equation}

It is well known \cite[Equation~(1.7)]{vanassche_recurrence_mop} that multiple orthogonal polynomials $(P^{(N)}_{\vec k})$ satisfy a system of linear recurrence relations only involving nearest neighbor multi-indices $\vec k \pm \vec e_j$,
\begin{equation}\label{eq:recurrence_relations}
xP_{\vec k}^{(N)}(x)=P^{(N)}_{\vec k+\vec e_j}(x)+b^{(j)}_{\vec k}P^{(N)}_{\vec k}(x)+a^{(1)}_{\vec k}P^{(N)}_{\vec k -\vec e_1}(x)+a^{(2)}_{\vec k}P^{(N)}_{\vec k -\vec e_2}(x),\quad j=1,2,
\end{equation}
where $a^{(j)}_{\vec k}=a^{(N,j)}_{\vec k}$ and $b^{(j)}_{\vec k}=b^{(N,j)}_{\vec k}$ vary with $N$. Additionally, the coefficients $a^{(N,j)}_{\vec k}$ have the following expression in terms of  $\gamma_{\vec k}^{(N,j)}$ in \eqref{eq:normalization_typeI},
\begin{equation}\label{eq:recurrence_coeff_norming_constants}
a^{(N,j)}_{\vec k}=\frac{\gamma_{\vec k}^{(N,j)}}{\gamma_{\vec k+\vec e_j}^{(N,j)}} 
,\quad j=1,2.
\end{equation}

The sequence of type I MOP's also satisfy a similar recurrence, although it is not as widely explored. It takes the form \cite[Equation~(1.10)]{vanassche_recurrence_mop}
\begin{equation}\label{eq:rec_relation_type_I}
xQ^{(N)}_{\vec k}(x)=Q^{(N)}_{\vec k-\vec e_j}+b^{(j)}_{\vec k-\vec e_j}Q^{(N)}_{\vec k}+a^{(1)}_{\vec k}Q^{(N)}_{\vec k +\vec e_1}(x)+a^{(2)}_{\vec k}Q^{(N)}_{\vec k +\vec e_2}(x),
\end{equation}
where the coefficients are the same as the ones that appear in \eqref{eq:recurrence_relations}, although here the $b$'s come with shifted index.

Theorem~\ref{thm:as_convergence} will be proved under the assumption of boundedness of zeros of MOP's slightly off a given up-right path. This control of zeros of polynomials off the path, in essence, can be traced back to the different recurrence relations satisfied by the MOP's: even though they satisfy a four term (nearest neighbor) recurrence relation, under general conditions one cannot assure that the coefficients in this recurrence are bounded (see for instance \cite[Theorem~2.3.1]{denisov_yattselev_2020}). On the other hand, under the assumption of boundedness of zeros, the recurrence relation along a given path has bounded coefficients, but has the drawback of having more terms, as shown in the next three results, moving us further away from the path.

\begin{prop}
Let  $(\vec n_k)$ be an up-right path as in \eqref{eq:up_right}, with $d$ defined in \eqref{eq:up_right} and $N>0$ given. There exist coefficients $\theta^{(N,0)}_{\vec n_k},\hdots, \theta^{(N,d)}_{\vec n_k}$ such that the sequence $(P^{(N)}_{\vec n_k})_{k\geq 0}$ satisfies a $d$-term recurrence relation of the form
\begin{equation}\label{eq:step_recurrence}
x\, P^{(N)}_{\vec n_k}(x)=P^{(N)}_{\vec n_{k+1}}(x)+\sum_{\ell=0}^d \theta^{(N,\ell)}_{\vec n_k} P^{(N)}_{\vec n_{k-\ell}}(x).
\end{equation}
Furthermore, these recurrence coefficients $\theta^{(N,0)}_{\vec n_k},\hdots \theta^{(N,d)}_{\vec n_k}$ are given by polynomial expressions in the variables $a_{\vec n_{k-j}}^{(N,1)}$, $a_{\vec n_{k-j}}^{(N,2)}$, $b_{\vec n_{k-j}}^{(N,1)}$ and $b_{\vec n_{k-j}}^{(N,2)}$, $j=0,\hdots, d$, and with coefficients and degrees that remain uniformly bounded as $N,k\to\infty$.
\end{prop}

\begin{proof}
We follow standard arguments, see e.g.~\cite[Lemma~3.9]{hardy_2015} and also \cite[Proposition~2.1]{DuitsFahsKozhan2019}. Let us write
$$
m_k=j, \quad \text{with } j \text{ determined by } \vec n_{k}-\vec n_{k-1}=\vec e_j.
$$
Evaluating \eqref{eq:recurrence_relations} with $j=1$ and $j=2$ and taking the difference, we get the compatibility condition
$$
P^{(N)}_{\vec n_k + \vec e_1}(x)=P^{(N)}_{\vec n_k+\vec e_2}+(b^{(N,2)}_{\vec n_k}-b^{(N,1)}_{\vec n_k})P^{(N)}_{\vec n_k}(x).
$$
We apply this to \eqref{eq:recurrence_relations} recursively: whenever $m_\ell=1$ we use the above to get rid of $P^{(N)}_{\vec n_\ell-\vec e_2}$, whereas when $m_\ell=2$ we eliminate $P^{(N)}_{\vec n_\ell-\vec e_1}$. As a result, \eqref{eq:recurrence_relations} becomes
$$
xP^{(N)}_{\vec n_k}=P^{(N)}_{\vec n_{k+1}}(x)+\sum_{\ell=0}^k \theta^{(N,\ell)}_{\vec n_k} P^{(N)}_{\vec n_{k-\ell}}(x),
$$
where
$$
\theta^{(N,0)}_{\vec n_k}=b^{(N,m_k)}_{\vec n_k},\quad \theta^{(1)}_{\vec n_k}=a^{(N,1)}_{\vec n_k} + a^{(N,2)}_{\vec n_k}
$$
and
$$
\theta^{(N,\ell)}_{\vec n_k} =a_{\vec n_k}^{(N,1)}\prod_{j=k-\ell+1}^{k} (b^{(N,1)}_{\vec n_j-\vec e_1}-b^{(N,m_j )}_{\vec n_j-\vec e_1})+a_{\vec n_k}^{(N,2)}\prod_{j=k-\ell+1}^{k} (b^{(N,2)}_{\vec n_j-\vec e_2}-b^{(N,m_j )}_{\vec n_j-\vec e_2}), \quad 2\leq \ell \leq k.
$$
From the last condition in \eqref{eq:up_right}, we see that when $\ell>d$, both products above are zero, concluding the proof.
\end{proof}

\begin{cor} \label{cor:generalizedRecurr}
	Let  $(\vec n_k)$ be an up-right path as in \eqref{eq:up_right}, and $N>0$ fixed and set
$$
\Theta^{(N)}_{\vec n_k}:=(\theta^{(N,0)}_{\vec n_k},\hdots,\theta^{(N,d)}_{\vec n_k})
$$	
to be the vector of coefficients of the recurrence relation \eqref{eq:step_recurrence}.  For any given $\ell>0$, there exist polynomials $F^{(\ell)}_{-d(\ell-1)}, F^{(\ell)}_{-d\ell +1},\hdots, F^{(\ell)}_{\ell-1}$ in $(\ell+1)(d+1)$-variables such that the sequence $(P^{(N)}_{\vec n_k})_{k\geq 0}$ satisfies the recurrence relation
	\begin{equation}\label{eq:gener_recurr}
	x^\ell \, P^{(N)}_{\vec n_k}(x)=P^{(N)}_{\vec n_{k+\ell }}(x)+ \sum_{j=-d(\ell-1)}^{\ell-1} F^{(\ell)}_{j}\left( \Theta^{(N)}_{\vec n_{k-(\ell-1)d}},\Theta^{(N)}_{\vec n_{k-(\ell-1)d+1}},\hdots, \Theta^{(N)}_{\vec n_{k+\ell-1}}  \right) P^{(N)}_{\vec n_{k+j}}(x) , \quad k\geq d\ell,
	\end{equation}
	%
	%
	The polynomials $F_j^{(\ell)}$'s are universal, in the sense that they are independent of the measures of multiple orthogonality and independent of $N$ and of the up-right path $(\vec n_k)$. In particular, their coefficients and degrees remain bounded as $N,k\to\infty$.
\end{cor}

\begin{proof}
	For $\ell=1$ the assertion is trivial. Let $\ell\geq 2$; we can rewrite  the recurrence relation \eqref{eq:step_recurrence} as
$$
x^{\ell}P^{(N)}_{\vec n_k}(x)=x^{\ell-2}\left(xP^{(N)}_{\vec n_{k+1}}(x)+\sum_{j=0}^d \theta^{(N,j)}_{\vec n_k} xP^{(N)}_{\vec n_{k-j}}(x)\right)
$$
and use \eqref{eq:step_recurrence} again with each polynomial between parentheses. Iterating this process we get \eqref{eq:gener_recurr}. 
\end{proof}

We also need to control the growth of the recurrence coefficients appearing in the previous results.

\begin{prop} \label{propZerosbdd}
Let  $(\vec n_k)$ be an up-right path as in \eqref{eq:up_right} and $(N_k)$ be a sequence of integers with $N_k\to \infty$ as $k\to \infty$. If the zeros of $P^{(N_k)}_{\vec n_{k+1}}$ remain uniformly bounded as $k \to \infty$, then the corresponding recurrence coefficients $\theta_{\vec n_k}^{(N_k,0)},\hdots, \theta_{\vec n_k}^{(N_k,d)}$ in \eqref{eq:step_recurrence} also remain bounded as $k\to\infty$.
\end{prop}

\begin{proof}
We follow the standard scheme of proof, see e.g.~\cite[Lemma 2.2]{Aptetal2006}. For ease of notation, we omit the $k$-dependence on $N_k$ and simply write $N_k=N$.

Recall that the zeros of $P^{(N)}_{\vec n_j}$ are all real and simple, and the zeros of $P^{(N)}_{\vec n_{j}}$ and $P^{(N)}_{\vec n_{j+1}}$ interlace, for any $j$. If they are additionally contained in an interval $\Delta\subset \R$, then for any compact $K\subset \C\setminus \Delta$ there exists a constant $C(K)>0$ such that
$$
\left| \frac{P^{(N)}_{\vec n_{j-k}}(x)  }{ P^{(N)}_{\vec n_{j-i}}(x) } \right|\leq C, \quad x\in K, \; j \in \N, \; 0\leq i\leq d, \; -1\leq k< i.
$$
 We rewrite \eqref{eq:step_recurrence}  as
 $$
  x=\frac{P^{(N)}_{\vec n_{k+1}}(x)}{P^{(N)}_{\vec n_k}(x)}+\theta^{(N,0)}_{\vec n_k} +\sum_{j=1}^d \theta^{(N,j)}_{\vec n_k} \frac{P^{(N)}_{\vec n_{k-j}}(x)}{P^{(N)}_{\vec n_k}(x)}.
 $$
%
 With $r>0$ sufficiently large for which $\Delta\subset \{|z|<r \}$, Cauchy's integral formula yields
 $$
\theta^{(N,0)}_{k}=-\frac{1}{2\pi i} \oint_{|z|=r} \frac{P^{(N)}_{\vec n_{k+1}}(z)}{P^{(N)}_{\vec n_k}(z)}\frac{dz}{z},
 $$
 and this proves that $|\theta^{(N,0)}_{k} |\leq C(\{|z|\leq r \})$. 
 
 In the same vein, for any $\ell$ with $1\leq \ell \leq d$,
   $$
 x\, \frac{P^{(N)}_{\vec n_{k}}(x)}{P^{(N)}_{\vec n_{k-\ell}}(x)}  =\frac{P^{(N)}_{\vec n_{k+1}}(x)}{P^{(N)}_{\vec n_{k-\ell}}(x)}+ \theta^{(N,\ell)}_{k} +\sum_{\substack{j=1 \\ j\neq \ell}}^d \theta^{(N,j)}_k \frac{P^{(N)}_{\vec n_{k-j}}(x)}{P^{(N)}_{\vec n_{k-\ell}}(x)}
 $$
 so that
 \begin{align*}
\theta^{(N,\ell)}_{k} 
 & =  \frac{1}{2\pi i} \oint_{|z|=r} \left(  z\, \frac{P^{(N)}_{\vec n_{k}}(z)}{P^{(N)}_{\vec n_{k-\ell}}(z)} - \frac{P^{(N)}_{\vec n_{k+1}}(z)}{P^{(N)}_{\vec n_{k-\ell}}(z)} -\sum_{j=1}^{\ell-1} \theta^{(N,j)}_k \frac{P^{(N)}_{\vec n_{k-j}}(z)}{P^{(N)}_{\vec n_{k-\ell}}(z)}  \right) \frac{dz}{z}.
 \end{align*}
 Thus, assuming inductively that for some constant $M>0$ 
 $$
 | \theta^{(N,j)}_k |\leq M, \quad j=0,\hdots, \ell-1,
 $$
we obtain
 $$
 | \theta^{(N,\ell)}_k |\leq r \, C(\{|z|\leq r \})+C(\{|z|\leq r \}) + (\ell-1) M \, C(\{|z|\leq r \}),
 $$
which is enough to conclude the proof.
\end{proof}

The last results established relations between type II MOP's along a given up-right path. But we also need to relate both type I and II MOP's slightly off the given path to the ones along the path. 

For a given sequence of multi-indices $(\vec n_k)$ with an up-right path structure satisfying \eqref{eq:up_right} we define the symbols 
\begin{equation} \label{def:Symbols1}
\parallel_k, \, \perp_k \in \{ 1, 2\}
\end{equation}
by the property
\begin{equation} \label{def:Symbols2}
\vec e_{\parallel_k} := \vec n_{k+1}-\vec n_k, \quad \vec e_{\parallel_k} \perp \vec e_{\perp_k}.
\end{equation}
In other words, the vector $\vec e_{\parallel_k}$ gives us the direction from $\vec n_k$ to $\vec n_{k+1}$, while $\vec e_{\perp_k}$ is the only vector among $\vec e_1$ and $\vec e_2$ which is perpendicular to $\vec e_{\parallel_k}$.

\begin{prop}\label{prop:rec_off_path}
Fix $N$ and an up-right path $(\vec n_k)$ with corresponding value $d$ as in \eqref{eq:up_right} and let $(Q_{\vec m})=(Q_{\vec m}^{(N)})$ and $(P_{\vec m})=(P_{\vec m}^{(N)})$ be type I  and type II MOP's for the given $N$. Then there exist $\ell_k \in \{1,\hdots, d-1\} $ such that with the notation \eqref{def:Symbols1}--\eqref{def:Symbols2}, 
\begin{equation} \label{eq:RRQ1}
Q_{\vec n_k+\vec e_{\perp_k}}(x)=\sum_{i=0}^{\ell_k} q^{(i)}_{\vec n_k} Q_{\vec n_{k+i+1}}(x) , 
\end{equation}
where the coefficients $q^{(i)}_{\vec n_k}$ are defined recursively by the formula
\begin{equation} \label{coef:RRQ1}
q^{(0)}_{\vec n_k}=1, 
\qquad q^{(j+1)}_{\vec n_{k}}=q^{(j)}_{\vec n_{k}}\left(b^{(\perp_k)}_{\vec n_{k+j}+\vec e_{\parallel_k}}-b^{(\parallel_k)}_{\vec n_{k+j}+\vec e_{\perp_k}}\right), \; j=0,1, ,\hdots, \ell_k-1,
\end{equation}
and $b^{(j)}_{\vec k}$ are the recurrence coefficients in \eqref{eq:recurrence_relations}. 

Analogously,  there exist $\nu_k \in \{1,\hdots, d-1\}$ such that
\begin{equation} \label{eq:RRP1}
P_{\vec n_k-\vec e_{\perp_{k-1}}}(x)=\sum_{i=0}^{\nu_k} p^{(i)}_{\vec n_{k}} P_{\vec n_{k-i-1}} (x),
\end{equation}
where the coefficients above are defined recursively by
\begin{equation} \label{coef:RRP1}
p^{(0)}_{\vec n_k}=1, 
\quad p^{(j+1)}_{\vec n_k}=p^{(j)}_{\vec n_k}\left(b^{(\perp_{k-1})}_{\vec n_{k-j-1}-\vec e_{\perp_{k-1}}}-b^{(\parallel_{k-1})}_{\vec n_{k-j-1}-\vec e_{\perp_{k-1}}}\right), \; j=0, 1,\hdots, \nu_k-1.
\end{equation}
\end{prop}

\begin{proof}
Taking the difference of the recurrence relations \eqref{eq:rec_relation_type_I} with $j=1$ and $j=2$ gives the identity
$$
Q_{\vec n-\vec e_2}(x)-Q_{\vec n-\vec e_1}(x)=\left(b^{(1)}_{\vec n-\vec e_1}-b^{(2)}_{\vec n-\vec e_2}\right)Q_{\vec n}(x),
$$
which is valid for any $\vec n\in \N^2$. Applying it to $\vec n=\vec n_{k}+\vec e_1+\vec e_2=\vec n_{k+1}+\vec e_{\perp_k}$ we obtain the relation
\begin{align*}
Q_{\vec n_k+\vec e_{\perp_k}}(x) & =Q_{\vec n_k+\vec e_{\parallel_k}}(x)+\left(b^{(\perp_k)}_{\vec n_k+\vec e_{\parallel_k}}  -b^{(\parallel_k)}_{\vec n_k+\vec e_{\perp_k}}\right) Q_{\vec n_{k+1}+\vec e_{\perp_k}}(x)\\
 & =  Q_{\vec n_{k+1}}(x)+ q^{(1)}_{\vec n_k}\, Q_{\vec n_{k+1}+\vec e_{\perp_k}}(x).
\end{align*}
If $\vec e_{\perp_k}= \vec e_{\parallel_{k+1}}$ so that $\vec n_{k+1}+\vec e_{\perp_k}= \vec n_{k+2}$, then we get  \eqref{eq:RRQ1}-- \eqref{coef:RRQ1} with $\ell_k=1$. Otherwise,  $\perp_k= \perp_{k+1}$, $\parallel_k= \parallel_{k+1}$, and we replace $Q_{\vec n_{k+1}+\vec e_{\perp_k}}$ in the right hand side using this identity with  $k \mapsto k+1$:
\begin{align*}
Q_{\vec n_k+\vec e_{\perp_k}}(x)  & =Q_{\vec n_{k+1}}(x)+q^{(1)}_{\vec n_k}\, \left( Q_{\vec n_{k+2}}(x)+ \left(b^{(\perp_k)}_{\vec n_{k+1}+\vec e_{\parallel_{k+1}}}  -b^{(\parallel_k)}_{\vec n_{k+1}+\vec e_{\perp_{k+1}}}\right)Q_{\vec n_{k+2}+\vec e_{\perp_k}}(x)  \right)\\
& = Q_{\vec n_{k+1}}(x)+q^{(1)}_{\vec n_k}\, Q_{\vec n_{k+2}}(x)+ q^{(1)}_{\vec n_k}\,   \left(b^{(\perp_k)}_{\vec n_{k+1}+\vec e_{\parallel_k}}  -b^{(\parallel_k)}_{\vec n_{k+1}+\vec e_{\perp_k}}\right) Q_{\vec n_{k+2}+\vec e_{\perp_k}}(x) \\
& = Q_{\vec n_{k+1}}(x)+q^{(1)}_{\vec n_k}\, Q_{\vec n_{k+2}}(x)+ q^{(2)}_{\vec n_k}\,    Q_{\vec n_{k+2}+\vec e_{\perp_k}}(x)  .
\end{align*}
We can iterate this process until $Q_{\vec n_{k+j-1}+\vec e_{1} + \vec e_2} =Q_{\vec n_{k+j}+\vec e_{\perp_k}}= Q_{\vec n_{k+j+1}}$. By the last assumption in \eqref{eq:up_right}, this will happen in at most $d-1$ steps. 


The proof of \eqref{eq:RRP1} goes along the same line. From the nearest neighbor recurrence relation \eqref{eq:recurrence_relations} we obtain
$$
P_{\vec n +\vec e_2}(x)-P_{\vec n+\vec e_1}(x)=(b^{(1)}_{\vec n}-b^{(2)}_{\vec n})P_{\vec n}(x),
$$
valid for any $\vec n\in \N^2$, and apply it to $\vec n=\vec n_k-\vec e_1-\vec e_2=\vec n_{k-1}-\vec e_{\perp_{k-1}}$. This yields 
$$
P_{\vec n_{k-1}-\vec e_{\perp_{k-1}}}(x)=P_{\vec n_{k-1}}(x)+\left(b^{(\perp_{k-1})}_{\vec n_{k-1}-\vec e_{\perp_{k-1}}}-b^{(\parallel_{k-1})}_{ \vec n_{k-1}-\vec e_{\perp_{k-1}}}  \right)P_{ \vec n_{k-1}-\vec e_{\perp_{k-1}}}(x).
$$
If $\vec e_{\perp_{k-1}}=\vec  e_{\parallel_{k-2}}$ so that $\vec n_{k-1}- \vec e_{\perp_{k-1}}= \vec n_{k-2}$, then we get  \eqref{eq:RRP1}--\eqref{coef:RRP1} with $\nu_k=1$. Otherwise,  $\vec e_{\perp_{k-1}}= \vec e_{\perp_{k-2}}$, $\parallel_{k-1}= \parallel_{k-2}$, and we replace $P_{\vec n_{k-1}-\vec e_{\perp_{k-1}}}$ in the right hand side using this identity with  $k \mapsto k-1$:
\begin{align*}
P_{\vec n_{k-1}-\vec e_{\perp_{k-1}}}(x) & = P_{\vec n_{k-1}}(x)+p^{(1)}_{\vec n_k}\, \left( P_{\vec n_{k-2}}(x)+\left(b^{(\perp_{k-2})}_{\vec n_{k-2}-\vec e_{\perp_{k-2}}}-b^{(\parallel_{k-2})}_{ \vec n_{k-2}-\vec e_{\perp_{k-2}}}  \right)P_{ \vec n_{k-2}-\vec e_{\perp_{k-2}}}(x)  \right)\\
	& = P_{\vec n_{k-1}}(x)+p^{(1)}_{\vec n_k}\,   P_{\vec n_{k-2}}(x) + p^{(1)}_{\vec n_k}\,   \left( b^{(\perp_{k-1})}_{\vec n_{k-2}-\vec e_{\perp_{k-1}}}-b^{(\parallel_{k-1})}_{ \vec n_{k-2}-\vec e_{\perp_{k-1}}}   \right) P_{ \vec n_{k-2}-\vec e_{\perp_{k-2}}}(x) \\
	& =P_{\vec n_{k-1}}(x)+p^{(1)}_{\vec n_k}\,   P_{\vec n_{k-2}}(x) + p^{(2)}_{\vec n_k}\,    P_{ \vec n_{k-2}-\vec e_{\perp_{k-2}}}(x). 
\end{align*}
We can iterate this process until $P_{\vec n_{k-j}-\vec e_{\perp_{k-1}}}= P_{\vec n_{k-j-1}}$. By the last assumption in \eqref{eq:up_right}, this will happen in at most $d-1$ steps. 
%
%
\end{proof}

Proposition~\ref{propZerosbdd} allows us to control the recurrence coefficients for the type II MOP's along an up-right path. The next result is its analogue for the coefficients that appeared in Proposition~\ref{prop:rec_off_path}.

\begin{prop} \label{prop:zerosboundedsequencePN}
Consider an  up-right path $(\vec n_N)$, with corresponding sequence of type I and II MOP's with $N$-varying weights $(Q^{(N)}_{\vec n_N})$ and $(P^{(N)}_{\vec n_N})$ and coefficients $(p_{\vec n_N}^{(j)})=(p_{\vec n_N}^{(N,j)})$ and $(q_{\vec n_N}^{(j)})=(q_{\vec n_N}^{(N,j)})$ as in Proposition~\ref{prop:rec_off_path}. If the zeros of $P^{(N)}_{\vec n_{N+d}}$ remain bounded as $N\to\infty$, then the sequences $(p_{\vec n_N}^{(N,j)})$ and $(q_{\vec n_N}^{(N,j)})$ also remain bounded as $N\to\infty$.
\end{prop}

\begin{proof}
By standard arguments already used in the proof of Proposition~\ref{propZerosbdd} we obtain that  the nearest neighbor recurrence coefficients  \eqref{eq:recurrence_relations} satisfy the identity
\begin{equation}\label{eq:int_repr_rec_coeff}
b^{(N,j)}_{\vec n}=
-\frac{1}{2\pi i}\oint_{|z|=r}\frac{P^{(N)}_{\vec n+\vec e_j}(z)}{P^{(N)}_{\vec n}(z)}\frac{dz}{z}, \quad j=1, 2,
\end{equation}
where $r$ is chosen so that the disk $\{|z|=r\}$ encloses all the roots of $P^{(N)}_{\vec n}$. By Proposition~\ref{prop:realzeros}, the zeros of $P^{(N)}_{\vec n+\vec e_j}$ and $P^{(N)}_{\vec n}$ interlace, and by the same arguments used to prove \eqref{eq:recurrence_relations} we see that boundedness of the zeros of $P^{(N)}_{\vec n+\vec e_j}$ implies boundedness of $b^{(N,j)}_{\vec n}$.

In particular, taking $\vec n= \vec n_{N+k} + \vec e_i$, $i\neq j$, $i\in \{1,2\}$, $k\ge 0$, we get 
%
$$
b^{(N,j)}_{\vec n_{N+k}+\vec e_{i}}=-\frac{1}{2\pi i}\oint_{|z|=r} \frac{P^{(N)}_{\vec n_{N+k}+\vec e_1+\vec e_2}(z)}{P^{(N)}_{n_{N+k}+\vec e_i}(z)}\frac{dz}{z}.
$$
By Proposition~\ref{prop:realzeros}, the zeros of $P^{(N)}_{\vec n_{N+d}}$ bound the zeros of $P^{(N)}_{\vec n_{N+k} +\vec e_1 + \vec e_2}$, $0\le k \le \ell_N-2\le d-3$. Moreover, by construction,   $P^{(N)}_{\vec n_{N+\ell_N-1} +\vec e_1 +\vec  e_2}= P^{(N)}_{\vec n_{N+d}}$. This shows that all $b^{(N,j)}_{\vec n_{N+k}+\vec e_{i}}$ taking part in the formulas  \eqref{coef:RRQ1} are uniformly bounded.


We apply similar arguments to bound the $p$'s, except that we use \eqref{eq:int_repr_rec_coeff} with  $\vec n =\vec n_{N-k}-\vec e_j$, $j\in \{1,2\}$ and $k\ge 0$, which gives
$$
b^{(N,j)}_{\vec n_{N-k}-\vec e_j}=-\frac{1}{2\pi i}\oint_{|z|=r}\frac{P^{(N)}_{\vec n_{N-k}}(z)}{P^{(N)}_{\vec n_{N-k}-\vec e_j}(z)}\frac{dz}{z}.
$$
The argument is even simpler now, since because of the up-right path structure and the interlacing, all zeros of $P^{(N)}_{\vec n_{N-k}}$, $k\ge 0$, are  bounded by the zeros of $P^{(N)}_{\vec n_{N+d}}$.
\end{proof}

\begin{remark} \label{rem:sequenceofPN}
Observe that the assumption of Proposition~\ref{prop:zerosboundedsequencePN} that all the zeros of $P^{(N)}_{\vec n_{N+d}}$ remain bounded as $N\to\infty$ is equivalent to assume boundedness of zeros of the sequence $P^{(N-d)}_{\vec n_{N}}$. This will be used in the formulation of the statement of Theorem~\ref{thm:as_convergence}.
\end{remark}

We conclude this section with a result that will be needed to prove Theorem~\ref{thm:as_convergence}.
\begin{prop} \label{prop:finalbounds}
Let $P^{(N)}_{\vec k}$ and $Q^{(N)}_{\vec k}$ be the biorthogonal functions in \eqref{eq:biothogonality} constructed for the weights in \eqref{defPNmultiple}, where we recall that $V$ is a polynomial with $\deg V=m$. Let $(\vec n_N)$ be a up-right path as in \eqref{eq:up_right} with corresponding value $d$. Set $\kappa=\max \{m-2,d\}$. If the zeros of $P^{(N)}_{\vec n_N+\kappa}$ remain bounded as $N\to\infty$, then the integrals
$$
\int x^s P^{(N)}_{\vec n_N}(x)Q^{(N)}_{\vec n_N}(x)dx, \quad s=1,\dots, m-2
$$
all remain bounded as $N\to \infty$. The same conclusion holds true if in the integral above we replace $P^{(N)}_{\vec n_N}(x)Q^{(N)}_{\vec n_N}(x)$ by either one of
$$
P^{(N)}_{\vec n_N - \vec e_j}\, Q^{(N)}_{\vec n_N},\quad \quad P^{(N)}_{\vec n_N}\, Q^{(N)}_{\vec n_N+ \vec e_j}, \quad \mbox{or} \quad P^{(N)}_{\vec n_N-\vec e_j}\, Q^{(N)}_{\vec n_N+ \vec e_i} ,\quad i,j=1,2.
$$
\end{prop}

\begin{proof}
Equation~\eqref{eq:biothogonality} says that for an up-right path $(\vec n_k)$,
$$
\int P^{(N)}_{\vec n_k}(x)Q^{(N)}_{\vec n_j}(x)dx=\delta_{k,j-1}.
$$
Using \eqref{eq:gener_recurr} we then get that if $k\geq d s$, 
\begin{align*}
\int x^s P^{(N)}_{\vec n_N}(x) & Q^{(N)}_{\vec n_N}(x)dx  
 \\
 & = \sum_{j=-d(s-1)}^{s} F^{(s)}_{j}\left( \Theta^{(N)}_{\vec n_{N-(s-1)d}},\Theta^{(N)}_{\vec n_{N-(s-1)d+1}},\hdots, \Theta^{(N)}_{\vec n_{N+s-1}}  \right)    \int  P^{(N)}_{\vec n_{N+j}}(x) Q^{(N)}_{\vec n_N}(x)dx \\
 & = F^{(s)}_{-1}\left( \Theta^{(N)}_{\vec n_{N-(s-1)d}},\Theta^{(N)}_{\vec n_{N-(s-1)d+1}},\hdots, \Theta^{(N)}_{\vec n_{N+s-1}}  \right) ,
\end{align*}
where in the above and for the rest of the proof we have set
$$
F_s^{(s)}=1 \qquad \text{and}\qquad F_j^{(s)}=0 \quad \text{for } j> s.
$$
The conclusion is now a consequence of Corollary~\ref{cor:generalizedRecurr} and Proposition~\ref{propZerosbdd}.

Consider the second integral of interest, namely
$$
\int x^s P^{(N)}_{\vec n_N - \vec e_j}(x)\, Q^{(N)}_{\vec n_N} (x)dx, \quad s=1,\dots, m-2.
$$
If $\vec n_N - \vec e_j= \vec n_{N-1}$, the result follows from what we just proved. So, we only need to consider the case (see notation \eqref{def:Symbols1}-- \eqref{def:Symbols2})
$$
\int x^s P_{\vec n_N-\vec e_{\perp_{N-1}}}(x) \, Q^{(N)}_{\vec n_N} (x)dx,
$$
which by Proposition~\ref{prop:rec_off_path} and the biorthogonality reduces to 
$$
\sum_{i=0}^{\nu_N} p^{(i)}_{\vec n_{N}} \int x^s P_{\vec n_{N-i-1}} (x)\, Q^{(N)}_{\vec n_N} (x)dx=\sum_{i=0}^{\nu_N}p^{(i)}_{\vec n_{N}} F_i^{(s)}\left(\Theta^{(N)}_{\vec n_{N-i-1-(s-1)d}},\hdots, \Theta^{(N)}_{\vec n_{N-i+s-2}}\right).
$$
Proposition~\ref{prop:zerosboundedsequencePN} and interlacing then tells us that the $p$'s in the sum above are bounded as long as the zeros of $P^{(N)}_{\vec n_{N+d}}$  are bounded, whereas Proposition~\ref{propZerosbdd} assures the $F$'s above are bounded as long as the zeros of $P^{(N)}_{\vec n_{N+m-3}}$are bounded.

Analogously, for
$$
\int x^s P^{(N)}_{\vec n_N} (x)\, Q^{(N)}_{\vec n_N+ \vec e_j} (x)dx
$$
we again only need to analyze 
\begin{align*}
\int x^s P^{(N)}_{\vec n_N} (x)\, Q^{(N)}_{\vec n_N+ \vec e_{\perp_N}} (x)dx & = \sum_{i=0}^{\ell_N} q^{(i)}_{\vec n_N} \int x^s P^{(N)}_{\vec n_N} (x)\,  Q_{\vec n_{N+i+1}}(x)dx \\
& = \sum_{i=0}^{\ell_N}q^{(i)}_{\vec n_N}F_i^{(s)}\left(\Theta^{(N)}_{\vec n_{N-(s-1)d}},\hdots \Theta^{(N)}_{\vec n_{N+s-1}}\right),
\end{align*}
which is bounded in the same manner as in the previous case.

Finally, for the last integral,
$$
\int x^s P^{(N)}_{\vec n_N-\vec e_j}(x) Q^{(N)}_{\vec n_N+ \vec e_i}(x)dx
$$
we need only to consider the case
$$
\int x^s P^{(N)}_{\vec n_N-\vec e_{\perp_{N-1}} }(x) Q^{(N)}_{\vec n_N+ \vec e_{\perp_{N+1}}}(x)dx,
$$
as all the other possible choices for $\vec e_j$ and $\vec e_i$ reduce to one of the previous cases. This last integral, in turn, equals
\begin{multline*}
\sum_{j=0}^{\nu_N}\sum_{i=0}^{\ell_N}p^{(j)}_{\vec n_{N}}q^{(i)}_{\vec n_N}\sum_{k=-d(s-1)}^{s}F_k^{(s)}\left(\Theta^{(N)}_{\vec n_{N-j-1-(s-1)d}},\hdots, \Theta^{(N)}_{\vec n_{N-j+s-2}}\right)\int P^{(N)}_{\vec n_{N-j-1+k}}(x)Q^{(N)}_{\vec n_{N+i+1}}(x)dx \\
= \sum_{j=0}^{\nu_N}\sum_{i=0}^{\ell_N}p^{(j)}_{\vec n_{N}}q^{(i)}_{\vec n_N}F_{i+j+1}^{(s)}\left(\Theta^{(N)}_{\vec n_{N-j-1-(s-1)d}},\hdots, \Theta^{(N)}_{\vec n_{N-j+s-2}}\right).
\end{multline*}
With the same argument as before, we are assured that the $p$'s and $q$'s are bounded because the zeros of $P^{(N)}_{\vec n_{N+d}}$ remain bounded, whereas the $F$'s remain bounded because the zeros of $P^{(N)}_{\vec n_{N+m-3}}$ remain bounded, completing the proof.
\end{proof}

\section{Characterization of the eigenvalue distribution} \label{sec:generalcase}

Our main goal in this final section is to prove Theorems~\ref{prop:wave_functions} and \ref{thm:as_convergence} from Section~\ref{sec:general}. We thus emphasize that now we will always assume that $V$ is a polynomial of even degree $m$ and positive leading coefficient.

\subsection{The associated Riemann-Hilbert problems and the proof of Theorem~\ref{prop:wave_functions}}\label{sec:rhp}

\

The starting point is the Riemann-Hilbert problem characterizing the MOP's in \eqref{defPNmultiple}, obtained for the first time by Geronimo, Kuijlaars and Van Assche \cite{Assche01}: find a $3\times 3$ matrix-valued function $\bm Y=\bm Y_{\vec k}=\bm Y_{\vec k}^{(N)}$ such that
\begin{enumerate}[i)]
	\item $\bm Y:\C\setminus \R\to \C^{3\times 3}$ is analytic;
	\item $\bm Y$ has continuous boundary values $\bm Y_\pm$ on $\R$ from  $\mathbb H_\pm$, and they satisfy the jump condition 
	$$
	\bm Y_+(x)=\bm Y_-(x) (\bm I+e^{-NV_1(x)}\bm E_{12}+e^{-NV_2(x)}\bm E_{13}),\quad x\in \R,
	$$
	where $\bm E_{ij}= \bm e_i \bm e^T_j$ is the matrix whose only non-zero entry is $1$ in the position $(i,j)$ and $V_1$ and $V_2$ were defined in \eqref{defV1};
	\item the following asymptotic expansion is valid,
	$$
\bm 	Y(z)=(\bm I+\Boh(z^{-1}))\diag(z^{k_1+k_2},z^{-k_1},z^{-k_2}),
	$$
	as $z\to \infty$ in any direction non-tangential to $\R$.
\end{enumerate}

The unique solution $\bm Y$ satisfying i)--iii) is given by
\begin{equation} \label{Y}
\bm Y(z)=\bm \Lambda_{\vec k}^{(N)}
\begin{pmatrix}
P^{(N)}_{\vec k}(z) & \mathcal C_1[P^{(N)}_{\vec k}](z) & \mathcal C_2[P^{(N)}_{\vec k}](z)  \\
-2\pi i P^{(N)}_{\vec k-\vec e_1}(z) & \mathcal C_1[P^{(N)}_{\vec k-\vec e_1}](z)& \mathcal C_2[P^{(N)}_{\vec k-\vec e_1}](z) \\
-2\pi i  P^{(N)}_{\vec k-\vec e_2}(z) & \mathcal C_1[P^{(N)}_{\vec k-\vec e_2}](z) & \mathcal C_2[P^{(N)}_{\vec k-\vec e_2}](z) 
\end{pmatrix}
\diag\left(1,2\pi i,2\pi i\right)^{-1}
\end{equation}
with
\begin{equation}\label{eq:norming_constants_2}
\bm \Lambda^{(N)}_{\vec k}:=\diag\left(1, \gamma_{\vec k}^{(N,1)}, \gamma_{\vec k}^{(N,2)}\right),
\end{equation} 
where the coefficients above are as in \eqref{eq:norming_constants}.

In the formula above and in what follows, $\mathcal C$ and $\mathcal C_j$ denote the Cauchy and weighted Cauchy operators, defined for a suitable function $f$ by
$$
\mathcal C[f](z)=\int \frac{f(x)}{x-z}dx,\quad \mathcal C_j[f](z):= \mathcal C[f e^{-NV_j}](z),\quad z\in \C\setminus \R,\quad  j=1,2.
$$

The biorthogonal functions \eqref{secondkindQ} can also be encoded in a Riemann-Hilbert problem, usually denoted by the letter $\bm X=\bm X_{\vec k}$. This Riemann-Hilbert problem formulation is not needed here, but rather its explicit solution, which is known to satisfy $\bm X^T=\bm Y^{-1}$, and consequently $\bm Y^{-1}$ can be explicitly computed in terms of the biorthogonal functions $(Q^{(N)}_{\vec k})$, taking the form \cite{Assche01}
\begin{equation}\label{RHP_X}
\bm Y^{-1}(z) = 
\begin{pmatrix}
-\mathcal C[Q_{\vec k}^{(N)}](z) & -\frac{1}{2\pi i}\mathcal C[Q^{(N)}_{\vec k+\vec e_1}](z) & -\frac{1}{2\pi i}\mathcal C[Q^{(N)}_{\vec k+\vec e_2}](z) \\
2\pi i A^{(N,1)}_{\vec k}(z) & A^{(N,1)}_{\vec k+\vec e_1}(z) & A^{(N,1)}_{\vec k+\vec e_2}(z) \\
2\pi i A^{(N,2)}_{\vec k}(z) & A^{(N,2)}_{\vec k+\vec e_1}(z) & A^{(N,2)}_{\vec k+\vec e_2}(z) \\
\end{pmatrix}
\diag\left( 1, \gamma^{(N,1)}_{\vec k+\vec e_1},\gamma^{(N,2)}_{\vec k+\vec e_2} \right)^{-1}.
\end{equation}

\begin{remark}
We want to stress that we are dealing with {\it varying weights}, so all the quantities just introduced depend on $N$, $k_1$ and $k_2$, and these are {\it independent parameters}. Of particular interest for us is the choice $k_1=n_1,k_2=n_2$ and $N=k_1+k_2=n_1+n_2$ simultaneosly, as this leads to the average characteristic polynomial \eqref{averageCharPoly}. In what follows, we always denote $N=n_1+n_2$, but we need to treat $k_1$ and $k_2$ as independent of $N$. To avoid a more cumbersome notation, we will drop temporarily the $N$ dependence and simply write $P_{\vec k}^{(N)}=P_{\vec k}$, $A^{(N,1)}_{\vec k}=A^{(1)}_{\vec k}$ etc, and reserve the indices $n_1$ and $n_2$ for when $N=n_1+n_2$. Whenever the $N$-dependence becomes relevant, we write it explicitly. We will also deliberately drop the argument $z$ from these functions when the variable is irrelevant. 
\end{remark}

The Riemann-Hilbert problem for $\bm Y$ and the corresponding non-linear steepest descent method of Deift and Zhou are powerful tools in proving the macro and microscopic scale asymptotic results for \eqref{external_source_model} and several other random matrix models \cite{aptekarev_lysov_tulyakov_2, duits_kuijlaars_two_matrix_model, kuijlaars_multiple_orthogonal_random_matrix_theory_ICM, bertola_bothner_chain_matrix_model, bleher_delvaux_kuijlaars_external_source, bleher_its, deift_book}. We will not explore this asymptotic side of Riemann-Hilbert problems, but rather use it as an algebraic tool that nicely encodes many quantities on the model.

Define
\begin{equation}\label{defT}
\bm T(z) = \bm T_{\vec k} (z) :=\bm Y(z) \bm F(z)^{-1}, \; \quad z\in \C\setminus \R, 
\end{equation}
with
\begin{equation*}
\bm F(z):=\diag(e^{NV(z)},e^{Na_1z},e^{Na_2z}), \; z\in \C, \quad \mbox{so}\quad  \bm F(z)^{-1}=\diag(e^{-NV(z)},e^{-Na_1z},e^{-Na_2z}). 
\end{equation*}
Further denoting
$$
\bm D(z):= \diag\left( V'(z), a_1,a_2 \right), \quad \bm K=\bm K_{\vec k}^{(N)}:= \frac{1}{N}  \diag\left( k_1+k_2, -k_1,-k_2 \right),
$$
the matrix $\bm F$ satisfies
\begin{equation}\label{eq:ode_F}
\bm F'(z)=N\bm D(z)\bm F(z),\quad (\bm F(z)^{-1})'=-N \bm D(z) \bm F(z)^{-1}.
\end{equation}

\begin{prop} \label{propRH1}
The matrix-valued function $\bm T$ satisfies the first-order differential equation
\begin{equation} \label{ODEmatrix_T}
\frac{d}{dz}\, 	\bm T(z) = N \bm R_{\vec k}(z) \bm T(z), \quad z\in \C\setminus \R,
\end{equation}
where $\bm R_{\vec k}$ is a matrix-valued polynomial of the form 
\begin{equation}\label{Rexplicit_T}
\bm R_{\vec k}(z)=
\begin{pmatrix}
- V'(z)+\Boh(z^{m-2}) & \Boh(z^{m-2}) & \Boh(z^{m-2}) \\
\Boh(z^{m-2}) & \Boh(z^{m-3}) & \Boh(z^{m-3}) \\
\Boh(z^{m-2}) & \Boh(z^{m-3}) & \Boh(z^{m-3}) 
\end{pmatrix},
\end{equation}
and $V$ was defined in \eqref{def:potential_v_coeff}. Moreover, 
\begin{align} \label{traceR}
\tr \bm R_{\vec k}(z) & =-V'(z), \\
\tr \bm R^2_{\vec k}(z) & = \left(V'(z)\right)^2 +\Boh(z^{m-2}),  
\label{traceR2} \\
\det  \bm R_{\vec k}(z) & = \det\left( \frac{1}{N} \bm T'(z) \bm T(z)^{-1}\right)= v_m a^2 z^{m-1}+\left( v_{m-1}a +v_m  \frac{k_1-k_2}{N} \right) a z^{m-2}+\Boh(z^{m-3}). \label{detR}
\end{align}
\end{prop}
\begin{proof}
It is straightforward to check that $\bm T$ is the only solution of the following Riemann-Hilbert problem:
\begin{enumerate}[(i)]
\item $\bm T$ is analytic on $\C\setminus \R$;
\item $\bm T$ has continuous boundary values $\bm T_\pm$ on $\R$ from  $\mathbb H_\pm$, and they satisfy the jump condition 
$$
\bm T_+(x)=\bm T_-(x) (\bm I+ \bm E_{12}+ \bm E_{13}),\quad x\in \R;
$$

\item we have that
\begin{equation}\label{asymptT}
\bm 	T(z)=(\bm I+\bm S(z))\diag(z^{k_1+k_2},z^{-k_1},z^{-k_2}) \bm F(z)^{-1}
\end{equation}
where
$$
\bm S(z)=\Boh(z^{-1}) 
$$
as $z\to \infty$ in any direction non-tangential to $\R$. 
\end{enumerate}

Since the jump matrix for $\bm T$ across $\R$ is constant, standard arguments yield \eqref{ODEmatrix_T}, that is,
\begin{equation}\label{ODEmatrixBis}
\frac{d}{dz}\, 	\bm T(z) = N \bm R_{\vec k}(z) \bm T(z), \quad z\in \C\setminus \R,
\end{equation}
where $\bm R_{\vec k}(z)$ is an entire matrix-valued function.  We can find $\bm R_{\vec k}$ by comparing the asymptotics at $z\to \infty$ of both sides. 
Indeed, with the notation in \eqref{defT}--\eqref{eq:ode_F} we have that
\begin{align*}
\frac{d}{dz}\, 	\bm T(z) & = N\left(\frac{1}{N}\bm Y'(z)-\bm Y(z)\bm D(z)\right)\bm F(z)^{-1}\\
& = N \left(\frac{1}{N}\bm S'(z) +(\bm I+\bm S(z)) \left(\frac{1}{z}\bm K-\bm D(z)\right) \right) \diag\left(z^{k_1+k_2},z^{-k_1},z^{-k_2}\right)\bm F(z)^{-1}.
\end{align*}
Using it in \eqref{ODEmatrixBis} together with \eqref{asymptT}, and recalling that $\det \bm T(z)\equiv 1$ for all $z\in \C\setminus\R$ so that $\bm T(z)^{-1}$ always exists, we get that
\begin{equation}\label{Rtemp}
\bm R_{\vec k}(z) =  \frac{1}{N} 	\bm T'(z)  	\bm T^{-1}(z) = \left [ \frac{1}{N} \bm S'(z) + (\bm I+\bm S(z)) \left( \frac{1}{z} \bm K - \bm D(z) \right)  \right]  (\bm I+\bm S(z))^{-1}.
\end{equation}
Liouville's theorem implies that $\bm R_{\vec k}$ is a polynomial; since 
$$
(\bm I+\bm S(z))^{-1} =  \bm I+ \Boh(z^{-1}), \quad z\to \infty,
$$
we have in fact that 
\begin{equation} \label{expressionforR}
\begin{split}
\bm R_{\vec k}(z) & =  \polyn \left [  (\bm I+\bm S(z)) \left( \frac{1}{z} \bm K - \bm D(z) \right)   (\bm I+\bm S(z))^{-1} \right]  \\
& = - (\bm I+\bm S(z))  \bm D(z)    (\bm I+\bm S(z))^{-1}  +\widetilde {\bm S}(z),
\end{split}
\end{equation}
where 
$$
\widetilde {\bm S}(z) =  \Boh(z^{-1}) , \quad z\to\infty,
$$
and $\polyn [\cdot ]$ stands for the polynomial part of the expansion of the argument at infinity. Direct computations yield that
\begin{equation}\label{eq:expansion_R_D}
\bm R_{\vec k}(z) =-\bm D(z)+
\begin{pmatrix}
\Boh(z^{m-2}) & \Boh(z^{m-2}) & \Boh(z^{m-2}) \\
\Boh(z^{m-2}) & \Boh(z^{m-3}) & \Boh(z^{m-3}) \\
\Boh(z^{m-2}) & \Boh(z^{m-3}) & \Boh(z^{m-3}) 
\end{pmatrix},
\end{equation}
giving us \eqref{Rexplicit_T}. 
Using the cyclic property of the trace in \eqref{expressionforR} we get that
\begin{align*}
\tr \bm R_{\vec k}(z) = &   - \tr \bm D(z)=-V'(z),
\end{align*}
and 
\begin{align*}
\tr \bm R^2_{\vec k} & =  \tr  \bm D^2(z)    -2 \tr (\bm D(z) \widetilde {\bm S}(z)) + \Boh(z^{m-2}) .
\end{align*}
 This proves \eqref{traceR}--\eqref{traceR2}.
 
Finally, by \eqref{Rtemp},
$$
\det \bm R_{\vec k}(z) =   \det \left( \frac{1}{z} \bm K - \bm D(z) \right)  \det  \left ( \bm I + \frac{1}{z}  \bm E(z)    \right), 
$$
where
$$
\bm E(z):= \frac{z}{N}   \left( \frac{1}{z} \bm K - \bm D(z) \right)^{-1} (\bm I+\bm S(z))^{-1}  \bm S'(z).
$$
Notice that
$$
\left( \frac{1}{z} \bm K - \bm D(z) \right)^{-1} = \diag \left( \frac{1}{-V'(z)+\frac{k_1+k_2}{Nz}}, \frac{-1}{a_1+\frac{k_1}{N z}},  \frac{-1}{a_2+\frac{k_2}{N z}} \right)= \Boh(1), \quad z\to \infty,  
$$
so that
$$
\tr \bm E(z) =  \Boh(z^{-1}) , \quad z\to\infty.
$$
Using that 
$$
\det \left ( \bm I + \frac{1}{z}  \bm E(z)    \right) =1+\frac{1}{z}\tr \bm E(z) +\Boh(z^{-2})=1+\Boh(z^{-2}),
$$
and since
$$
\det \left( \frac{1}{z} \bm K - \bm D(z) \right) = \left( -V'(z)+\frac{k_1+k_2}{Nz}  \right) \left(  a_1+\frac{k_1}{N z} \right) \left(  a_2+\frac{k_2}{N z} \right), 
$$
we conclude that 
$$
\det \bm R_{\vec k}(z) =  \polyn \left[   \left( -V'(z)+\frac{k_1+k_2}{Nz}  \right) \left(  a_1+\frac{k_1}{N z} \right) \left(  a_2+\frac{k_2}{N z} \right)    \left ( 1+\Boh(z^{-2})   \right) \right] ,
$$
which gives \eqref{detR}. 
\end{proof}
 
\begin{proof}[Proof of Theorem~\ref{prop:wave_functions}]
With $k_1=n_1$ and $k_2=n_2$, the explicit expression for $\bm Y$ in \eqref{Y} and relation \eqref{defT} shows that the first column of $\bm T$ is the vector of wave functions $\bm \Psi_{\vec n}$, and \eqref{ODEmatrix} is a consequence of \eqref{ODEmatrix_T}. 

Equation \eqref{Rexplicit} is the same as \eqref{Rexplicit_T}. 

To verify that the coefficients in \eqref{eq:finite_spectral_curve} satisfy \eqref{eq:normalization_p0}, recall that for any $3 \times 3$ matrix $\bm M$,
\begin{equation} \label{charpolM}
\det(\xi \bm I + \bm M)=\xi^3+  \left( \tr \bm M \right)\xi^2  +\frac{1}{2}\left((\tr \bm M)^2-\tr(\bm M^2) \right)\xi + \det   \bm M 
\end{equation}
and 
\begin{equation} \label{detM}
\det   \bm M=\frac{1}{6}\left((\tr \bm M)^3-3\tr \bm M \tr(\bm M^2)+2\tr(\bm M^3)\right).
\end{equation}
The conditions on the coefficients of \eqref{eq:finite_spectral_curve} are then immediate from \eqref{traceR}--\eqref{detR}.
\end{proof} 
 
\subsection{Asymptotic distribution}\label{sec:counting_measures}

\

In order to prove Theorem~\ref{thm:as_convergence}, we need to analyze the entries of the coefficient matrix $\bm R_{\vec k}$ in \eqref{ODEmatrix} in more detail. This analysis is split into some lemmas. 

\begin{lem}\label{lem_cauchy_transforms}
The biorthogonal functions $P_{\vec k}$ and $Q_{\vec j}$ in \eqref{eq:biothogonality} enjoy the following properties.
\begin{enumerate}[(i)]
\item Whenever $|\vec k|<|\vec j|$,
$$
P_{\vec k}\, \mathcal C [Q_{\vec j}]=\mathcal C[P_{\vec k}Q_{\vec j}],
$$
and 
$$
P_{\vec k}\, \mathcal C [Q_{\vec k}]=-1+\mathcal C[P_{\vec k}Q_{\vec k}].
$$

\item For any $j_2,k_2$,
\begin{equation*}
A_{\vec j}^{(1)}\mathcal C_1[P_{\vec k}]=
\begin{cases}
\mathcal C_1[A^{(1)}_{\vec j} P_{\vec k}],&\mbox{if } j_1\leq k_1+1, \\
-1+\mathcal C_1[A^{(1)}_{\vec j} P_{\vec k}],&\mbox{if } j_1= k_1+2.
\end{cases}
\end{equation*}

\item For any $j_1,k_1$,
\begin{equation*}
A_{\vec j}^{(2)}\mathcal C_2[P_{\vec k}]=
\begin{cases}
\mathcal C_2[A^{(2)}_{\vec j} P_{\vec k}],&\mbox{if } j_2\leq k_2+1, \\
-1+\mathcal C_2[A^{(2)}_{\vec j} P_{\vec k}],&\mbox{if } j_2= k_2+2.
\end{cases}
\end{equation*}
\end{enumerate}
\end{lem}

\begin{proof}
Write
$$
P_{\vec k}(z) \mathcal C [Q_{\vec j}](z)=-\int Q_{\vec j}(x) \frac{P_{\vec k}(x)-P_{\vec k}(z)}{x-z}dx+\mathcal C[P_{\vec k}Q_{\vec j}](z).
$$
The fraction on the integrand is a monic polynomial in $x$ of degree $k_1+k_2-1$, and part (i) then follows from \eqref{eq:biothogonality}. 

To prove (ii), we write similarly
\begin{equation*}
A^{(1)}_{\vec j}(z)\mathcal C_1[P_{\vec k}](z)=-\int P_{\vec k}(x) \frac{A^{(1)}_{\vec j}(x)-A^{(1)}_{\vec j}(z)}{x-z}e^{-NV_1(x)}dx +\mathcal C_1[A^{(1)}_{\vec j}P_{\vec k}](z)
\end{equation*}
From \eqref{eq:normalization_typeI} the quotient in the integrand is a polynomial of degree $j_1-2$ with leading coefficient $\gamma_{\vec j-\vec e_1}^{(1)}$. The orthogonality relations \eqref{defPNmultiple} imply the first identity in (ii), and also that
\begin{equation*}
A^{(1)}_{k_1+2,j_2}(z)\mathcal C_1[P_{\vec k}](z) = -\gamma_{\vec k+\vec e_1}^{(1)}\int x^{k_1} P_{\vec k}(x)e^{-NV_1(x)}dx+\mathcal C_1[A^{(1)}_{\vec j}P_{\vec k}](z),
\end{equation*}
and the second identity then follows from \eqref{eq:norming_constants}. Part (iii) follows analogously.
\end{proof}

\begin{lem}\label{lem:coeff_cauchy_transform}
If $(\vec n_N)$ is a sequence of multi-indices that satisfies the assumption of Theorem~\ref{thm:as_convergence}, then the coefficients of the polynomial (of degree $\le m-2$)
$$
\polyn\left(V'(z) \mathcal C[P_{\vec n_N}\, Q_{\vec n_N}] (z)\right)
$$
remain bounded as $N\to \infty$. 

The same holds true if we replace $P_{\vec n_N}\, Q_{\vec n_N}$ in the expression above either by 
\begin{equation}\label{otherExpr}
P_{\vec n_N - \vec e_j}\, Q_{\vec n_N},\quad \quad P_{\vec n_N}\, Q_{\vec n_N+ \vec e_j}, \quad \mbox{or} \quad P_{\vec n_N-\vec e_j}\, Q_{\vec n_N+ \vec e_i} ,\quad i,j=1,2.
\end{equation}
\end{lem}
\begin{proof}
For simplicity we omit the $N$-dependence on the multi-indices and write $\vec n_N=\vec n$. We start with $P_{\vec n}\, Q_{\vec n}$. As the coefficients of $V$ are independent of $N,n_1,n_2$, it is enough to verify that the coefficients of the polynomial part of $z^jC[P_{\vec n}\, Q_{\vec n}]$ are bounded for $j=0,\hdots, m-1$. The case $j=0$ is obvious because the Cauchy transform is $\Boh(z^{-1})$ as $z\to \infty$. 

For $1\leq j\leq m-1$, we proceed similarly as in the previous proof and write
\begin{equation}\label{eq:cauchy_transform_moments}
z^j\mathcal C[P_{\vec n}Q_{\vec n}](z)=-\int P_{\vec n}(x)Q_{\vec n}(x)\frac{x^j-z^j}{x-z}dx+\mathcal C[x^jP_{\vec n}Q_{\vec n}  ](z),
\end{equation}
so
\begin{equation}\label{eq:expansion_pol_cauchy}
\begin{split}
 \polyn\left(z^j\mathcal C[P_{\vec n}Q_{\vec n}](z)\right) & =-\sum_{s=0}^{j-1}z^{j-1-s}\int x^s P_{\vec n}(x)Q_{\vec n}(x)dx\\ & =-\sum_{s=1}^{j-1}z^{j-1-s}\int x^s P_{\vec n}(x)Q_{\vec n}(x)dx,
\end{split}
\end{equation}
where we have used also biorthogonality \eqref{eq:biothogonality}. By Proposition~\ref{prop:finalbounds}, integrals
$$
\int x^s P_{\vec n_k}(x)Q_{\vec n_k}(x)dx, \quad s=1,\dots, m-2
$$
are uniformly bounded (at least, for $k$ large enough, e.g.~$k\geq 2m$). 

Finally, for the expressions in \eqref{otherExpr} we apply the same procedure and reduce to integrals also of the forms considered in Proposition~\ref{prop:finalbounds}; we skip the details.

%
%
\end{proof}

The entries of the matrix $\bm R_{\vec k}$ involve norm constants for $P_{\vec k}$ and as such they may not be bounded. So instead of looking at $\bm R_{\vec k}$, we modify it into
\begin{equation}\label{def:matrixW}
\bm W_{\vec k}(z)=\bm\Lambda_{\vec{k}}^{-1}\bm R_{\vec{k}}(z)\bm \Lambda_{\vec k},
\end{equation}
where $\bm \Lambda_{\vec{k}}$ is as in \eqref{eq:norming_constants_2}. 

Consequently, all the equations \eqref{Rexplicit_T}--\eqref{detR} remain valid if we replace $\bm R_{\vec k}$ by $\bm W_{\vec k}$. In particular, the entries of $\bm W_{\vec k}$ are polynomials of degree at most $m-1$. 

\begin{prop}\label{prop:bounded_entries}
Under the assumptions of Theorem~\ref{thm:as_convergence}, the coefficients of the polynomial entries of the matrix $\bm W_{\vec n_N}$ remain bounded as $N\to\infty$.
\end{prop}
\begin{proof}
For ease of notation, during this proof we again suppress the $N$ dependence of the multi-indices and simply denote $\vec n_N=\vec n$.

We start rewriting the first identity in \eqref{Rtemp}, which gives us
$$
\bm W_{{\vec n}}(z)=\frac{1}{N}\bm \Lambda_{{\vec n}}^{-1} \left(\bm Y (z)\bm F(z) \right)' \bm F(z)^{-1}\bm Y(z)^{-1} \bm \Lambda_{{\vec n}} = \bm \Lambda_{{\vec n}}^{-1} \polyn\left(\bm Y(z)\bm D(z)\bm Y^{-1}(z)\right)\bm \Lambda_{{\vec n}},
$$
where we also used that the entries of $\bm W_{{\vec n}}$ are polynomials in $z$. Using \eqref{Y}, \eqref{RHP_X} and \eqref{eq:recurrence_coeff_norming_constants}, we compute
\begin{multline*}
\bm W_{{\vec n}}=\polyn\left[ \begin{pmatrix}
P_{{\vec n}} & \mathcal C_1[P_{{\vec n}}] & \mathcal C_2[P_{{\vec n}}] \\
-2\pi i P_{{\vec n}-\vec e_1} & \mathcal C_1[P_{{\vec n}-\vec e_1}]& \mathcal C_2[P_{{\vec n}-\vec e_1}] \\
-2\pi i  P_{{\vec n}-\vec e_2} & \mathcal C_1[P_{{\vec n}-\vec e_2}] & \mathcal C_2[P_{{\vec n}-\vec e_2}]
\end{pmatrix}
\diag\left(1,2\pi i,2\pi i\right)^{-1}\bm D \right. \\
\left. 
\times
\begin{pmatrix}
-\mathcal C[Q_{{\vec n}}] & -\frac{1}{2\pi i}\mathcal C[Q_{{\vec n}+\vec e_1}] & -\frac{1}{2\pi i}\mathcal C[Q_{{\vec n}+\vec e_2}] \\
2\pi i A^{(1)}_{{\vec n}} & A^{(1)}_{{\vec n}+\vec e_1} & A^{(1)}_{{\vec n}+\vec e_2} \\
2\pi i A^{(2)}_{{\vec n}}& A^{(2)}_{{\vec n}+\vec e_1} & A^{(2)}_{{\vec n}+\vec e_2} \\
\end{pmatrix}
\diag(1,a^{(1)}_{{\vec n}},a_{{\vec n}}^{(2)}) \right].
\end{multline*}
From this identity we can compute the entries of $\bm W_{\vec n}=\left(\bm W(i,j)\right)_{i,j=0}^2$. The $(1,1)$ entry, denoted $\bm W_{\vec n}(0,0)$, reads
$$
\bm W(0,0) =\polyn\left[ - V' \, P_{\vec n}\, \mathcal C[Q_{\vec n}]  +a_1 A^{(1)}_{\vec n}\mathcal C_1[P_{\vec n}]+a_2A^{(2)}_{\vec n}\mathcal C_2[P_{\vec n}]  \right]=-V'-\polyn\left[V'  P_{\vec n} Q_{\vec n}\right],
$$
where in the last step we used Lemma~\ref{lem_cauchy_transforms}. By Lemma~\ref{lem:coeff_cauchy_transform} this last expression has bounded polynomial coefficients as $N\to\infty$. The other entries are dealt with similarly: with
$$
(\mathfrak c_0,\mathfrak c_1,\mathfrak c_2)= (1,-2\pi i,-2\pi i),\quad a^{(0)}_{\vec n}=1,
$$
the other entries are
\begin{equation*}
\bm W(i,j) =\frac{a^{(j)}_{\vec n}}{\mathfrak c_j}\polyn\left(-\mathfrak c_i V' P_{\vec n-\vec e_i}\mathcal C[Q_{n+\vec e_j}] +a_1A^{(1)}_{\vec n+\vec e_j}\mathcal C_1[P_{\vec n-\vec e_i}] +a_2 A^{(2)}_{\vec n+\vec e_j}\mathcal C_2[P_{\vec n-\vec e_i}] \right)
\end{equation*}
which can be shown to be bounded using Lemmas~\ref{lem_cauchy_transforms} and \ref{lem:coeff_cauchy_transform} again.
\end{proof}

We are finally ready to prove Theorem~\ref{thm:as_convergence}.

\begin{proof}[Proof of Theorem~\ref{thm:as_convergence}]
We start by rewriting \eqref{ODEmatrix} in terms of the matrix $\bm W_{\vec n_N}$ in \eqref{def:matrixW}, namely
\begin{equation}\label{eq:ODE_final_proof}
\frac{d}{dz}\bm \Phi_{\vec n_N}(z)=N\bm W_{\vec n_N}(z)\bm \Phi_{\vec n_N}(z),\quad \bm \Phi_{\vec k}(z):=\bm \Lambda_{\vec k}^{-1} \bm \Psi_{\vec k}(z).
\end{equation}
Using \eqref{def:wavefunctionPsi} we obtain that for any $z$ for which the entries of $\bm \Phi_{\vec n_N}$ do not vanish
\begin{equation}\label{eq:ODE_final_proof_2}
\frac{d}{dz}\bm \Phi_{\vec n_N}(z)=
N\left(\bm Q_{\vec n_N}(z)-V'(z)\bm  I\right)\bm \Phi_{\vec n_N}(z)
\end{equation}
where
$$
 \bm Q_{\vec k}(z):=\diag \left(\frac{P'_{\vec k}(z)}{N P_{\vec k}(z)}, \frac{P'_{\vec k-\vec e_1}(z)}{N P_{\vec k-\vec e_1}(z)}, \frac{P'_{\vec k-\vec e_2}(z)}{N P_{\vec k-\vec e_2}(z)} \right).
$$
Comparing \eqref{eq:ODE_final_proof} and \eqref{eq:ODE_final_proof_2}, we conclude that
$$
\left( \bm Q_{\vec n_N}(z) -V'(z) \bm I -\bm W_{\vec n_{N}}(z)    \right)\bm \Phi_{\vec n_N}(z)=\bm 0. 
$$
which is valid for $z\in \C$ for which the entries of $\bm\Phi_{\vec n_N}$ are nonzero. Thus, for such $z$ we conclude that
\begin{equation} \label{detIdentity}
\det \left( V'(z) \bm I -\bm Q_{\vec n_N}(z) +\bm W_{\vec n_{N}}(z)    \right)=0.
\end{equation}

The last step is to take the limit in this equation. First, Proposition~\ref{prop:bounded_entries} tells us that we can extract a subsequence $(\vec n_k)\subset (\vec n_N)$ for which all the entries of $\bm W_{\vec n_k}$ converge, say to $\bm W_{\infty}$. Because  $\bm W_{\vec n_k}$ also satisfies \eqref{traceR}--\eqref{detR}, for any $\xi,z\in \C$ in compacts we get the uniform convergence
$$
\det\left(\xi\bm I+\bm W_{\vec n_k}(z)\right)\stackrel{k\to\infty}{\to} \det(\xi \bm I+\bm W_\infty(z))=\xi^3+p_2(z)\xi^2+p_1(z)\xi+p_0(z),
$$
where the polynomial coefficients satisfy \eqref{eq:normalization_p0}.

Next, under the assumptions of Theorem~\ref{thm:as_convergence}, 
$$
\lim  \bm  Q_{\vec n_N}(z) = -C^\lambda(z) \bm I, \quad z\in \C\setminus\R,
$$
and \eqref{detIdentity} combined with continuity of the determinant gives
$$
\det((C^\lambda(z)+V'(z))\bm I+\bm W_\infty(z))=0,\quad z\in \C\setminus \R,
$$
which by analytic continuation extends to $z\in \C\setminus \supp\lambda$. This proves the claim that $\xi(z)$ is indeed a solution to a spectral curve in the sense of Definition~\ref{spectral_curve}. 

Finally, as $\bm Q_{\vec n_k}$ and $\bm \Lambda_{\vec n_k}$ are diagonal matrices, the relation \eqref{def:matrixW} tells us that
$$
\det(\xi\bm I+ \bm R_{\vec n_k}(z))=\det(\xi\bm I+ \bm W_{\vec n_k}(z))\to \det (\xi\bm I  +\bm W_\infty(z)),
$$
showing that the spectral curve for $\xi(z)=C^\lambda(z)+V'(z)$ indeed is a limit of \eqref{eq:finite_spectral_curve}, concluding the proof.
\end{proof}


\section*{Acknowledgments}


The first author was partially supported Simons Foundation Collaboration Grants for Mathematicians (grant 710499) and by the European Regional Development Fund along with Spanish Government (grant MTM2017-89941-P) and Junta de Andaluc\'{\i}a (grant UAL18-FQM-B025-A, as well as research group FQM-229 and Instituto Interuniversitario Carlos I de F\'{\i}sica Te\'orica y Computacional), and by the University of Almer\'{\i}a (Campus de Excelencia Internacional del Mar  CEIMAR). 

Most of this work was carried out while the second author was a Postdoctoral Assistant Professor in the Department of Mathematics at the University of Michigan, and also while he was visiting Fudan University, China. He thanks the hospitality and excellent research atmosphere at these institutions. The second author also acknowledges his current support by São Paulo Research Foundation under grants \# 2019/16062-1 and \# 2020/02506-2.

Both authors have benefited from fruitful discussions with several colleagues, such as Andrew Celsus, Adrien Hardy, Arno Kuijlaars, Ken McLaughlin, Walter Van Assche, Maxim Yattselev and Lun Zhang, to mention a few. Last but not least, we are deeply indebted to the anonymous referee, who did an outstanding job.


\printbibliography

%

\end{document}